%% file: main.tex
\documentclass[review]{elsarticle}
\usepackage[utf8]{inputenc}
\usepackage{booktabs}
\usepackage{longtable}
\usepackage{tabularx,arydshln,hyperref,amssymb,amsmath}

\usepackage{todonotes}

\usepackage{todonotes}
\usepackage{paralist}
\usepackage{xspace}
\usepackage{times}
\usepackage{wrapfig}
\usepackage{subcaption}
\usepackage{enumitem}
\usepackage{pgfplots}

\input{macros-common}

\input{macros}

\input{macros-tikz}

\begin{document}

    \title{Enjoy the Silence:\\ Analysis of Stochastic Petri Nets with Silent Transitions}

    \author[1]{Sander J.J. Leemans}
    \ead{s.leemans@bpm.rwth-aachen.de}
    \author[2]{Fabrizio Maria Maggi}
    \ead{maggi@inf.unibz.it}
    \author[2]{Marco Montali}
    \ead{montali@inf.unibz.it}

    \address[1]{RWTH Aachen, Germany}
    \address[2]{Free University of Bozen-Bolzano,
    Italy}
    %
    
    \begin{abstract}
Capturing stochastic behaviors in business and work processes is essential to quantitatively understand how nondeterminism is resolved when taking decisions within the process. This is of special interest in process mining, where event data tracking the actual execution of the process are related to process models, and can then provide insights on frequencies and probabilities. Variants of stochastic Petri nets provide a natural formal basis to represent stochastic behaviors and support different data-driven and model-driven analysis tasks in this spectrum. However, when capturing processes, such nets need to be labelled with (possibly duplicated) activities, and equipped with silent transitions that model internal, non-logged steps related to the orchestration of the process. At the same time, they have to be analyzed in a finite-trace semantics, matching the fact that each process execution consists of finitely many steps. These two aspects impede the direct application of existing techniques for stochastic Petri nets, calling for a novel characterization that incorporates labels and silent transitions in a finite-trace semantics. In this article, we provide such a characterization starting from generalized stochastic Petri nets and obtaining the framework of labelled stochastic processes (\pwp{s}). On top of this framework, we introduce different key analysis tasks on the traces of \pwp{s} and their probabilities. We show that all such analysis tasks can be solved analytically, in particular reducing them to a single method that combines automata-based techniques to single out the behaviors of interest within a \pwp, with techniques based on absorbing Markov chains to reason on their probabilities. Finally, we demonstrate the significance of how our approach in the context of stochastic conformance checking, illustrating practical feasibility through a proof-of-concept implementation and its application to different datasets.

\end{abstract}


    \maketitle

\input{introduction}
  
  \input{relatedwork}

  \input{nets}

\input{outcome}

\input{verification}

  \input{hybrid}

    \input{evaluation}
    
    \input{conclusion}
    
    \bibliographystyle{splncs04}
    \bibliography{references}  

    \appendix
    
    \section{Evaluation Results}
    \label{app:results}
    
        {
            \tiny
            \input{results.tex}
        }
        
        \begin{figure}
            \includegraphics[width=\linewidth]{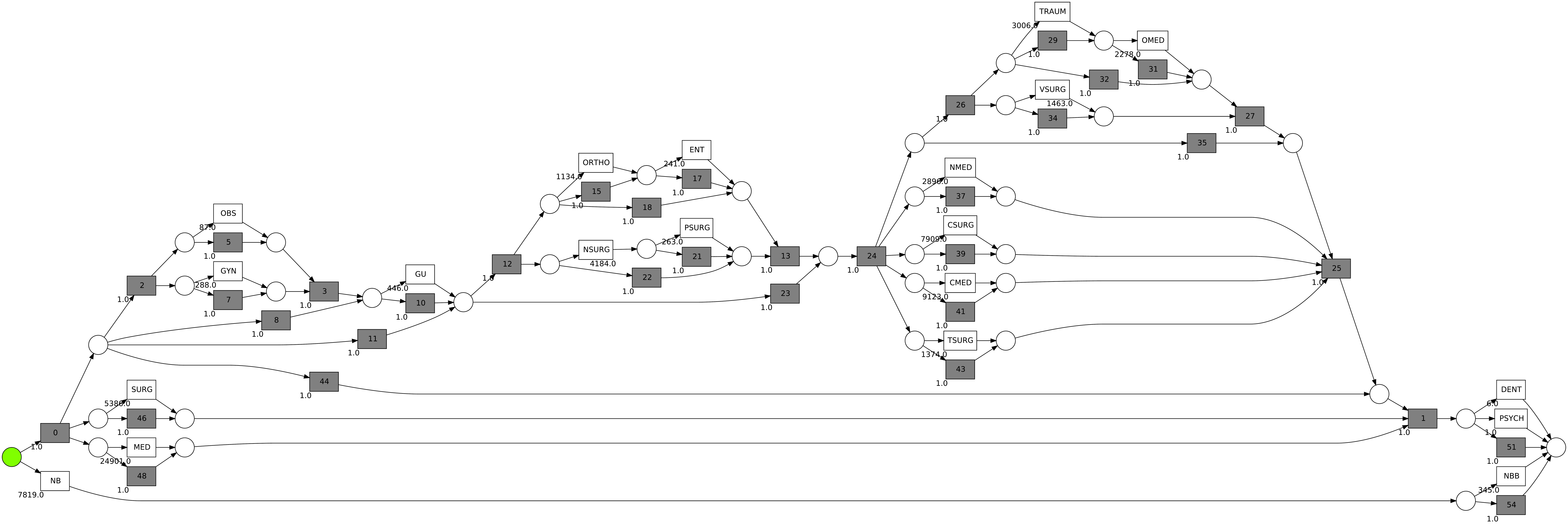}
            \caption{mimic-serv IMf BFE model.}
            \label{fig:eva:mimic-serv}
        \end{figure}

\end{document}

%% file: macros-common.tex
\newcommand{\set}[1]{\{#1\}}                      

\newcommand{\tup}[1]{\langle #1\rangle}            

\newcounter{dummy} 
\newcounter{dummy1} 
\newcounter{dummy2}
\newcounter{dummy3} 
\newcounter{dummy4}
\newcounter{dummy5} 
\newcounter{dummy6}
\newcounter{dummy7}

\usepackage[thmmarks,amsmath]{ntheorem}

\usepackage{setspace}

\theoremseparator{.}
\theorembodyfont{\itshape}
\theoremsymbol{$\triangleleft$}
\qedsymbol{\ensuremath{\Box}}

\theorembodyfont{\itshape\doublespacing}

\newtheorem{theorem}[dummy]{Theorem}
\newtheorem{lemma}[dummy1]{Lemma}

\theorembodyfont{\normalfont\doublespacing}
\newtheorem{definition}[dummy2]{Definition}
\newtheorem{proposition}[dummy3]{Proposition}
\newtheorem{example}[dummy4]{Example}
\newtheorem{remark}[dummy5]{Remark}
\newtheorem{corollary}[dummy6]{Corollary}
\newtheorem{problem}[dummy7]{Problem}

\theoremstyle{nonumberplain}
\theoremheaderfont{\itshape}
\theorembodyfont{\normalfont\doublespacing}

\theoremseparator{.}
\theoremsymbol{\ensuremath{\dashv}}
\newtheorem{proof}[dummy6]{Proof}
\theorembodyfont{\normalfont\doublespacing}

\qedsymbol{\ensuremath{\Box}}


%% file: macros.tex
\newcommand{\bpwp}{bounded LSP\xspace}
\newcommand{\Bpwp}{Bounded LSP\xspace}
\newcommand{\spwp}{GSPNP\xspace}
\newcommand{\pwp}{LSP\xspace}

\newcommand{\gspn}{GSPN\xspace}

\newcommand{\pval}[1]{\ensuremath{\rho_{#1}}}
\newcommand{\sstate}[1]{\ensuremath{m_{#1}}}

\newcommand{\taskname}[1]{\textsf{#1}}

\newcommand{\LTLf}{\ensuremath{\text{LTL}_f}\xspace}
\newcommand{\LDLf}{\ensuremath{\text{LDL}_f}\xspace}

\newcommand{\declare}{Declare\xspace}
\newcommand{\pdeclare}{ProbDeclare\xspace}

\newcommand{\mathname}[1]{\ensuremath{\text{\textit{#1}}}}

\newcommand{\funsym}[1]{\mathtt{#1}}

\newcommand{\restr}[2]{{
  \left.\kern-\nulldelimiterspace 
  #1 
  \vphantom{\big|} 
  \right|_{#2} 
  }}

%% file: macros-tikz.tex
\definecolor{darkblue}{rgb}{0.2,0.2,0.6}
\definecolor{darkred}{rgb}{0.6,0.1,0.1}
\definecolor{darkgreen}{rgb}{0.2,0.6,0.2}

\definecolor{acsiorange}{RGB}{180,88,26}
\definecolor{acsired}{RGB}{150,0,0}
\definecolor{acsigreen}{RGB}{128,198,54}
\definecolor{acsiblue}{RGB}{43,80,150}

\definecolor{swimmyred}{RGB}{206,19,55}
\definecolor{marcofucsia1}{RGB}{203,41,123}
\definecolor{marcofucsia2}{RGB}{153,25,94}
\definecolor{marcofucsia3}{RGB}{102,29,70}
\definecolor{marcoblue1}{RGB}{17,183,225}
\definecolor{marcoblue2}{RGB}{16,163,201}
\definecolor{marcoblue3}{RGB}{5,132,165}
\definecolor{marcoblue4}{RGB}{4,105,131}
\definecolor{marcoblue5}{RGB}{0,77,128}
\definecolor{marcoorange}{RGB}{255,147,0}
\definecolor{marcogreen1}{RGB}{82,174,139}
\definecolor{marcogreen2}{RGB}{56,122,101}

\definecolor{ocolor}{HTML}{123C69}
\definecolor{pcolor}{HTML}{AC3B61}

\tikzset{connect with angle/.style = {
  to path= {
  let \p1=(\tikztostart), \p2=(\tikztotarget) in -- ++({\x2-\x1-(\y1-\y2)*tan(#1-90)},0) -- (\tikztotarget)}
  }
}

\usetikzlibrary{arrows,arrows.meta,shapes,automata,petri,positioning,calc}

\newcommand{\tasklabel}[1]{\phantom{[}\textsf{#1}\phantom{]}}
\newcommand{\wval}[1]{\ensuremath{w_{{#1}}}}
\newcommand{\rval}[1]{\ensuremath{\lambda_{{#1}}}}

\tikzstyle{tlabel}=[
  xshift=-1mm,
  yshift=-2.9mm
]

\tikzstyle{arc}=[
  very thick,
  -to,
  rounded corners = 5pt,
]

\tikzstyle{darc}=[
  very thick,
  to-to,
]

\tikzset{place/.style args={#1}{
    circle,
    very thick,
    fill=marcofucsia3!10,
    draw=marcofucsia3!100,
    font=\sf,
    minimum width=6mm,
    inner ysep=2pt,
    label=below:{#1},
  }
}

\tikzstyle{transition}=[
  rectangle,
  very thick,
  fill=marcoblue3!20,
  draw=marcoblue3!100,
  font=\sf,
  minimum width=8mm,
  minimum height=8mm,
]

\tikzstyle{silent}=[
  rectangle,
  very thick,
  text=white,
  fill=black,
  draw=black,
  font=\sf,
  label=center:\textcolor{white}{$\tau$},
  minimum width=6mm,
  minimum height=6mm,
  align=center,
]

\tikzset{tweight/.style = {
  label={[label distance=-1mm]above:{#1}}
  }
}

\tikzset{tsilent/.style={
    rectangle,
    fill=black,
    tweight={#1},
    label=below:$\tau$, 
  }
}

\tikzset{wtransition/.style 2 args={
    rectangle,
    tweight={#2},
    very thick,
    fill=marcoblue3!20,
    draw=marcoblue3!100,
    label={[yshift=.5mm,font=\sffamily]below:{#1}},
  }
}

\tikzset{wsilenttransition/.style={
    rectangle,
    tweight={#1},
    fill=black,
    draw=black,
    label={[yshift=.5mm,font=\sffamily]below:{$\tau$}},
    font=\color{white},
  }
}

\tikzstyle{immediate}=[
    minimum width = .5mm,
    inner sep = .1mm,
    minimum height = 9mm,
]

\tikzstyle{timed}=[
    minimum width=8mm,
    minimum height=8mm,
]

\newcommand{\tname}[1]{\ensuremath{t_{{#1}}}}
\newcommand{\pname}[1]{\ensuremath{q_{{#1}}}}

\tikzstyle{hcolor} = [ black ]

\tikzstyle{invisible} = [
  minimum width=10mm,
  minimum height=10mm,
]

\tikzstyle{state} = [
  circle,
  draw=black!100,
  very thick,
  fill=marcoorange!20,
  minimum height=10mm,
]

\tikzstyle{autstate} = [
  circle,
  draw=black!100,
  very thick,
  fill=pink,
  minimum height=8mm,
]

\tikzstyle{prstate} = [
  circle,
  draw=black!100,
  top color=orange!20,
  bottom color=pink,
  very thick,
  fill=white,
  minimum height=8mm,
]

\tikzset{
  lstate/.style 2 args={
    state,
    label={[yshift=2mm]center:\ensuremath{#1}},
    label={[yshift=-2mm,font=\footnotesize]center:\ensuremath{#2}}  
  }
}

\tikzset{
  lprstate/.style 2 args={
    prstate,
    label={[yshift=1mm]center:\ensuremath{#1},},
    label={[yshift=-2mm]center:\ensuremath{#2}}  
  }
}

\tikzstyle{downloop} = [
  arc,
  loop below,
  looseness=5,
  out=-60,
  in=-120,
]

\tikzstyle{uploop} = [
  arc,
  loop below,
  looseness=5,
  out=60,
  in=120,
]

\tikzstyle{rightloop} = [
  arc,
  loop below,
  looseness=5,
  out=30,
  in=-30,
]

\tikzstyle{stransition}=[
  transition,
  minimum width=6mm,
  minimum height=6mm
]

\tikzstyle{silent}=[
draw=magenta,
]

\tikzstyle{deadend}=[draw=gray,opacity=.5]

%% file: introduction.tex
\section{Introduction}
    

The study of stochastic dynamic systems is a long-standing line of research, essential when one wants to quantitatively analyse how agents in AI and business/work processes in BPM resolve the nondeterminism intrinsically present when taking decisions. In the context of Petri nets, seminal works extending Petri nets with stochastic decision making date back to the early 80s \cite{natkin1980reseaux,molloy1982integration,molloy1982performance}, finally evolving into the framework of \emph{generalized stochastic Petri nets} (\gspn{s}) \cite{MarsanCB84,marsan1988stochastic,ABCD95}.

In BPM in particular, interest in stochastic processes and their analysis has been recently revived in the context of process mining, exploiting the fact that event data provide the basis for relating behaviors to frequencies and probabilities.  This is in turn key to extract useful information on such behaviors: a quality control process with 30\% failed checks is a considerably different process than a process with 2\% failed checks, even though they are expressed by the control flow.
    Explicitly enriching process models with a  stochastic dimension, which indicates how likely every model trace is, provides in this light the basis for extracting quantifiable insights, as well as increasing the quality of further process analysis and mining tasks such as simulation, prediction and recommendation. Not surprisingly, several stochastic process mining contributions have been provided in recent times, covering a variety of different tasks:
\begin{compactitem}[$\bullet$]
\item discovery of stochastic process models enriched with likelihood information on the different behaviors \cite{spdwe,RoggeSoltiAW13,DBLP:conf/icpm/BurkeLW20,DBLP:conf/apn/BurkeLW21};
\item process model repair~\cite{andreas-repair-otm2013};
\item stochastic conformance checking, either comparing the stochastic behavior implicitly represented by an event log with that induced by a stochastic process model \cite{DBLP:conf/icpm/PolyvyanyyK19,DBLP:journals/tosem/PolyvyanyySWCM20,DBLP:journals/is/AlkhammashPMG22,LABP21,StocasticCC}, or considering the likelihood of model traces when aligning the observed with the prescribed behavior \cite{BMMP21b}. 
\end{compactitem}
While all these works employ variants of Petri nets as control-flow backbone of the process, a parallel thread of research  infuses stochasticity in declarative process mining \cite{AMMP22}. 

Petri net-based approaches to stochastic process modeling typically come with explicit indications on the relative probability and timing of single steps in the process, which only implicitly yield the overall likelihood of its traces. \gspn{s} do so by supporting immediate and timed transitions. Immediate transitions have priority over the timed ones, and come with (relative) weights that are used to compute the probability of deciding which immediate transition to fire among the enabled ones. Each timed transition comes with the rate of a corresponding exponential distribution capturing the firing delay of the transition once it becomes enabled. Continuous and discrete Markov chains are traditionally used to analytically attack key analysis tasks such as calculating the expected time and probability of evolving a marking into another marking \cite{marsan1988stochastic,ABCD95}. Unfortunately, this well-understood, solid foundational basis cannot be readily applied to solve analysis tasks in BPM and process mining, due to three special requirements when modeling business processes that do not directly match with \gspn{s} and their execution semantics:
\begin{compactenum}
\item \emph{repeatable labels} -- transitions in the net must be labelled with corresponding activities in the process, possibly using the same label in multiple transitions;
\item \emph{silent transitions} -- there may be transitions that do not correspond to any visible activity in the process, but represent instead internal orchestration steps, such as those capturing gateways in the process;
\item \emph{finite-trace semantics} -- the interesting dynamics exhibited by the process are those representing the flow of cases/instances therein, starting from an initial state and reaching one among possibly multiple final states after having executed a finite (yet unbounded) amount of steps, for which only the visible activities are of interest.
\end{compactenum}
The interplay of such three requirements is far from trivial, in particular considering that a single trace of activities may correspond to infinitely many different finite-length runs in the process. This makes it impossible even to compute directly the probability of a model trace via enumeration of runs, and explains why stochastic process mining techniques have so far dealt with this problem through approximation \cite{LABP21,BMMP21b}. 

 To get an intuitive understanding of the problem, consider the labelled stochastic Petri net shown in 
    Figure~\ref{fig:lspn}(a), where silent transitions are shown in black (we will introduce these nets formally in Section~\ref{sec:lsp}). 
 This process has two model traces, but their likelihood may sound counterintuitive. For example, the probability that the process will generate the model trace consisting of $\taskname{a}$ followed by $\taskname{b}$ is $\frac{2}{3}$~\cite{LeeSA19}.
    The challenge here stems from the loop of silent transitions, which ``favours" $\taskname{b}$ over $\taskname{c}$. Technically, one can only compute that probability analytically by noticing that the probability of generating $\taskname{a},\taskname{b}$ corresponds to the sum of infinitely many probabilities, each obtained by executing a different number of iterations within the ``silent loop'' consisting of the two silent transitions. Specifically, the probability is $\frac{1}{2}+\frac{1}{2}\frac{1}{4}+\frac{1}{2}\left(\frac{1}{4}\right)^2+\ldots = \frac{1}{2}\sum_{i=0}^\infty \left(\frac{1}{4}\right)^i = \frac{\frac{1}{2}}{1-\frac{1}{4}}=\frac{2}{3}$.

  \begin{figure}[t]
        \centering
         \resizebox{\textwidth}{!}{
        \begin{tikzpicture}[node distance=1cm and .5cm]

            \node (source) [
              place={}
            ] {};
%
            \node (a) [
                wtransition={a}{1},
                timed, 
                right = of source
            ] {};
            \node (p1) [
                right = of a,  
                place
            ] {};
            \node (t1) [
                wsilenttransition={$1/2$},
                timed, 
                right = of p1
            ] {};
            \node (t2) [
                wsilenttransition={$1/2$}, 
                timed,
                above = of t1
            ] {};
            \node (b) [
                wtransition={b}{$1/2$},
                timed, 
                below = of t1] {};
            \node (p2) [
              right = of t1, 
              place] {};
            \node (c) [
                wtransition={c}{$1/2$},
                timed, 
                right = of p2] {};
            \node (sink) [
                place, 
                right = of c] {};

            \draw [arc] (source) to (a);
            \draw [arc] (a) to (p1);
            \draw [arc] (p1) to (t1);
            \draw [arc] (t1) to (p2);
            \draw [arc] (p2) to (c);
            \draw [arc] (c) to (sink);
            \draw [arc, rounded corners=5pt] (p1) |- (b);
            \draw [arc, rounded corners=5pt] (b) -| (sink);
            \draw [arc, rounded corners=5pt] (p2) |- (t2);
            \draw [arc, rounded corners=5pt] (t2) -| (p1);

            \node[below=5mm of b] (caption1) {(a) Stochastic net adapted from~\cite{LeeSA19}.};

            \def\dy{10mm}

			\node (source) [
                place,
                right = 10mm of sink] {};
			\node (a) [
                wtransition={a}{1},
                timed, 
                right = of source] {};
			\node (p1) [
                place, 
                right = of a, 
                yshift=\dy] {};
			\node (p2) [
                place, 
                right = of a, 
                yshift=-\dy] {};
			\node (b) [
                wsilenttransition={1},
                timed, 
                right = of p1] {};
			\node (p3) [
                place, 
                right = of b] {};
			
			\node (d) [
                wtransition={d}{1},
                timed,
                right = of p3] {};
			\node (p5) [place, right of=d] {};
			
            \node (c) [
                wtransition={c}{1}{},
                timed, 
            ] at (p2-|d) {};
			\node (p4) [place, right of=c] {};

			\draw [arc] (source) to (a);
			\draw [arc] (a) to (p1);
			\draw [arc] (a) to (p2);
			\draw [arc] (p1) to (b);
			\draw [arc] (b) to (p3);
			\draw [arc] (p3) to (d);
			\draw [arc] (d) to (p5);
			\draw [arc] (p2) to (c);
			\draw [arc] (c) to (p4);
			
			\draw [arc] (p2) to (d);
        
        \node at(caption1-|b) {(b) Stochastic net with confusion, adapted from~\cite{chiola1993generalized}.};

        \end{tikzpicture}
        }
        \caption{Two examples of labelled stochastic Petri nets.}
        \label{fig:lspn}
    \end{figure}

    Another example of potentially counterintuitive trace probability is shown in Figure ~\ref{fig:lspn}(b).
    In this net, the likelihood of $\taskname{a}$ followed by $\taskname{c}$ is $\frac{3}{4}$.
    The challenge in this example is again the silent transition, which is used here in a semi-concurrent context: the transition $\taskname{c}$ is mutually exclusive with transition $\taskname{d}$, but as $\taskname{c}$ is part of two runs (that is, executed before or after the silent transition), its probability is higher than one might expect~\cite[confusion]{chiola1993generalized}. In Section~\ref{sec:relatedwork}, we describe how existing techniques address or circumvent these challenges.

In this article, we start from \gspn{s} and propose the framework of \emph{(bounded) labelled stochastic processes} (\bpwp{s}) to capture processes that incorporate labels and silent transitions in a finite-trace semantics. On top of this framework, we define and analytically solve the following core problems:

\begin{compactenum}
\item \textbf{Outcome probability}: what is the probability that the \bpwp evolves from the initial marking to one (or a subset) of its final markings?  
\item \textbf{Trace probability}: what is the probability of a given trace of the \bpwp?
\item \textbf{Specification probability}: what is the probability that the \bpwp produces a trace that satisfies a given qualitative specification that captures desired behaviour?
\item \textbf{Stochastic compliance}: is the \bpwp of interest compatible, in behavioural and stochastic terms, to a probabilistic declarative specification \cite{AMMP22} indicating which temporal constraints are expected to hold, and with which probability?
\item \textbf{Stochastic conformance checking}: how can we employ the previous analysis questions, in particular trace probability, to improve the correctness and applicability of existing stochastic conformance checking techniques \cite{LeeSA19} relating a reference stochastic process model to a recorded log?
\end{compactenum}

Specifically, we show how outcome probability can be analytically solved by building on and suitably revising the connection~\cite{MarsanCB84,ABCD95} between \gspn{s} and discrete-time Markov chains, considering in particular the class of absorbing Markov chains \cite[Chapter 11]{GriS97}. We then demonstrate how the other analysis tasks can be reduced to outcome probability, relying on automata-based techniques inspired from those in qualitative verification of Markov chains against temporal properties~\cite[Ch.10]{BaiK08}, here adapted to handle finite traces. Such techniques are also used to elegantly deal with silent transitions. As a by-product, we hence provide a solution to an existing problem related to the removal of $\varepsilon$-transitions in stochastic finite-state automata, so far only solved through ad-hoc algorithms \cite{hanneforth2010epsilon}.

Finally, we demonstrate the significance of how our approach in the context of stochastic conformance checking, infusing an existing measure \cite{LeeSA19} with our analytical computation of trace probabilities. We illustrate practical feasibility through a proof-of-concept implementation and its application to different datasets.
 
This paper is a revised and largely extended version of~\cite{DBLP:conf/bpm/LeemansMM22}. In comparison with~\cite{DBLP:conf/bpm/LeemansMM22}:
\begin{compactitem}  
\item We here handle the full features of arbitrary \gspn{s}.
\item We extend all the foundational results, providing a complete characterisation of the considered analysis tasks, a more explicit connection with Markov chains, and a characterisation of \pwp{s} and their runs through stochastic languages, also examining the impact of livelocks.
\item We provide an implementation of the techniques in the ProM framework~\cite{DBLP:conf/apn/DongenMVWA05} and an evaluation of stochastic conformance checking techniques, made possible by our technique and its implementation.
\end{compactitem}

The article is organized as follows. Section~\ref{sec:relatedwork} discusses related work. In Section~\ref{sec:lsp}, we introduce the framework of \bpwp{s} based on \gspn{s} and its execution semantics. Section~\ref{sec:state-probabilities} provides a brief roadmap of the different analysis tasks and how they are inter-reduced to each other. In Section~\ref{sec:outcome} we attack the outcome probability problem showing how it can be analytically solved through a connection with absorbing Markov chains. In Section~\ref{sec:spec} we introduce the qualitative verification of \bpwp{s}, show how it can be used to compute the probability of model traces, and reduce it to outcome probability. Section~\ref{sec:delta} deals with stochastic compliance of \bpwp{s} against declarative probabilistic constraints, also in this case providing a reduction to the outcome probability problem. In Section~\ref{sec:evaluation} we describe the proof-of-concept implementation of our techniques, demonstrate their application to compute stochastic conformance checking metrics, and report on a corresponding evaluation. Conclusions follow.



%% file: relatedwork.tex
\section{Related Work}
\label{sec:relatedwork}
    
    Stochastic process-based models have been studied extensively in literature. In the context of this work, we are interested in formal, Petri net-based stochastic models that are at the basis of the recent series of approaches in stochastic process discovery \cite{spdwe,RoggeSoltiAW13,DBLP:conf/apn/BurkeLW21} and conformance checking \cite{LABP21,StocasticCC,Bergami21,DBLP:journals/is/AlkhammashPMG22,DBLP:conf/icpm/BurkeLWAH22}. Such approaches all refer to the model of (generalised) stochastic Petri nets, or fragments thereof. A first version of this model was proposed in~\cite{molloy1982performance}, extending Petri nets by assigning exponentially distributed firing \emph{rates} to transitions. 
    This was extended in~\cite{MarsanCB84} by distinguishing timed (as in~\cite{molloy1982performance}) and immediate transitions. 
    Immediate transitions have priority over timed ones, and have \emph{weights} to define their relative likelihood.
    As these two types of transitions, abstracting from time, behave homogeneously, we may capture the stochastic behaviour of the net through a discrete-time Markov chain \cite{MarsanCB84}.

    Several variants of stochastic Petri nets have been investigated starting from the seminal work in~\cite{MarsanCB84}. These variants differ from each other depending on the features they support (e.g., arbiters to resolve non-determinism, immediate vs timed transitions) and the way they express probabilities. Such nets may aid modellers in expressing certain constructs. An orthogonal, important dimension is to ensure that probabilities and concurrency interact properly. This can be achieved through good modelling principles~\cite{MarsanCB84,chiola1993generalized} or automated techniques~\cite{BrMM19}. 
    

Contrasting these formal models with recent works in stochastic process mining, key differences exist. Traditional stochastic nets do not support transition labels nor silent transitions, and put emphasis on recurring, infinite executions and the so-called steady-state analysis, focused on calculating the probability that an execution is currently placed in a given state. 
This is done by constructing a discrete-time Markov chain that characterises the stochastic behaviour of the net \cite{molloy1982performance,MarsanCB84}. 
Finding the probability of a finite-length trace in such nets is trivial, as every trace corresponds to a single path. However, no transition labels or silent steps are supported, which limits their usefulness for process mining due to the omnipresence of such transitions in process models. 
On the other hand, when these features are incorporated in stochastic Petri nets, which is precisely what we target in this paper, computing the probability of a trace cannot be approached directly anymore, as infinitely many paths would potentially need to be inspected. 
At the same time, in business processes we are interested in behaviour at the trace level rather than at the process level -- that is, we are not interested in the \emph{state} that a process can be in, but rather on the \emph{path} that a trace follows through the model -- thus the large body of work on steady-state-based analyses on Markov models does not apply for our purposes. 
This explains why reasoning on the stochastic behaviour of such extended nets has been conducted in an approximated way \cite{LABP21,StocasticCC}, or by imposing restrictions on the model \cite{Bergami21}. 

To bridge this gap, in this paper we take the most basic stochastic Petri nets: we do not consider time or priority, but we add (duplicate) labels and silent transitions. Importantly, \emph{our results seamlessly carry over to bounded, generalised stochastic Petri nets}, thanks to the fact that incorporating priorities in bounded nets is harmless, and that timed and immediate transitions are homogeneous from the stochastic point of view.
%
%
 %
 To the best of our knowledge, outside of recent work using stochastic Petri nets with silent transitions~\cite{LABP21,LeeSA19,Bergami21}, such nets have not been defined or studied before.

    


%
    
    While intuitively stochastic conformance checking techniques need to obtain the probability of a given trace in a stochastic process model (for instance, \cite{LeeSA19} explicitly obtains this probability to compute a distance measure between a log and a stochastic process model), some stochastic conformance checking techniques avoid computing the probability for a single trace, for instance by playing out the model to obtain a sample of executions~\cite{LABP21}, or by assuming that the model is deterministic~\cite{StocasticCC}.
    The results presented in this paper therefore enable the practical application of~\cite{LeeSA19}, and may enable further stochastic conformance checking techniques and, consequently, new types of analysis.
    
    Silent steps have been studied in the context of automata. 
    For instance, in~\cite{hanneforth2010epsilon} an ad-hoc method is described to iteratively remove all silent steps from a stochastic automaton.
    Due to concurrency and confusion (see for instance Figure~\ref{fig:lspn}(b)), such techniques are not directly applicable to stochastic Petri nets. A result of this paper is that silent steps can be handled directly, without the need for ad-hoc techniques.



%% file: nets.tex
\newcommand{\runningwproc}{\ensuremath{\wproc_{\text{order}}}}

\newcommand*\xbar[1]{%
  \hbox{%
    \vbox{%
      \hrule height 0.5pt 
      \kern0.3ex
      \hbox{%
        \kern-0.2em
        \ensuremath{#1}%
        \kern-0.0em
      }%
    }%
  }%
}

\newcommand{\multisets}[1]{\mathbb{M}(#1)}

\newcommand{\const}[1]{\mathsf{#1}}

\newcommand{\hidden}{\ensuremath{\tau}}

\newcommand{\uswn}{SWN\xspace}
\newcommand{\tg}{\ensuremath{G}}
\newcommand{\closed}[1]{\overline{#1}}
\newcommand{\enaset}[2]{E_{#2}(#1)}
\newcommand{\fire}[4]{#1\xrightarrow{#2}_{#4}#3}
\newcommand{\prob}[2]{\mathbb{P}_{#2}(#1)}
\newcommand{\cprob}[3]{\prob{#1|#2}{#3}}

\newcommand{\rg}[1]{RG(#1)}
\newcommand{\srg}[1]{SRG(#1)}
\newcommand{\ind}[1]{\textnormal{\texttt{#1}}}
\newcommand{\run}{\eta}
\newcommand{\trace}{\sigma}
\newcommand{\traces}[1]{\mathit{traces}(#1)}
\newcommand{\ptraces}[2]{\mathit{ptraces}_{#2}(#1)}

\newcommand{\nreach}[3][]{#2 \overset{#1}{\rightsquigarrow} #3}
\newcommand{\runs}[2]{runs_{#2}(#1)}
\newcommand{\seqs}[2]{seqs_{#2}(#1)}
\newcommand{\transp}[1]{#1^\top}
\newcommand{\embed}{\phi}
\newcommand{\trembed}{{\embed^{\text{tr}}}}
\newcommand{\gorgembed}{{\embed^{g}}}

\newcommand{\pa}{\rho_{23}}
\newcommand{\pb}{\rho_{24}}
\newcommand{\pc}{\rho_{55}}
\newcommand{\pd}{\rho_{65}}
\newcommand{\pe}{\rho_{67}}
\newcommand{\pf}{\rho_{57}}

\newcommand{\logtrace}{\trace}
\newcommand{\nonlogtrace}{{\sigma}}

\newcommand{\marked}[2]{\tup{#1,#2}}
\newcommand{\markedcat}[3]{\tup{#1,#2,#3}}

\newcommand{\enabledcat}[3]{#1[#2\rangle_{#3}}

\newcommand{\firecat}[4]{\enabledcat{#1}{#2}{#4}#3}
\newcommand{\mtrans}[1]{\Gamma_{#1}}
\newcommand{\dtrans}[1]{\Lambda_{#1}}
\newcommand{\cardeq}{\approx}

\newcommand{\pre}[1]{{^\bullet{#1}}}
\newcommand{\post}[1]{{#1^\bullet}}
\newcommand{\tcontext}[1]{^\circ{#1}}
\newcommand{\pcontext}[1]{#1^\circ}
\newcommand{\guard}[1]{[\![#1]\!]}
\newcommand{\guardf}{\mathname{G}}
\newcommand{\mcardp}[2]{m(#1)(#2)}

\newcommand{\tsys}[1]{\Lambda_{#1}}
\newcommand{\dbfun}{\funsym{db}}
\newcommand{\mfun}{\funsym{mrk}}

\newcommand{\places}{Q}
\newcommand{\place}{q}
\newcommand{\transitions}{T}
\newcommand{\itransitions}{\transitions^{im}}
\newcommand{\ttransitions}{\transitions^{ti}}
\newcommand{\transition}{t}
\newcommand{\flow}{F}
\newcommand{\labeling}{\ell}
\newcommand{\net}{N}
\newcommand{\wproc}{\mathcal{N}}

\newcommand{\alphabet}{\Sigma}

\newcommand{\tasks}{\mathcal{A}}
\newcommand{\marking}{m}
\newcommand{\imarking}{m_0}
\newcommand{\fmarkings}{M_f}
\newcommand{\weightfun}{w}
\newcommand{\probfun}{p}

\newcommand{\states}{S}
\newcommand{\istate}{s_0}
\newcommand{\fstates}{S_f}
\newcommand{\state}{s}
\newcommand{\transrel}{\varrho}

\newcommand{\mynet}{\net_{\text{order}}}
\newcommand{\myproc}{\wproc_{\text{order}}}

\newcommand{\outgoing}[2]{\mathit{succ}_{#2}(#1)}

\newcommand{\lgspn}{LGSPN\xspace}
 
\newcommand{\labels}{\mathcal{L}}

\newcommand{\expo}[2]{\ensuremath{#1\sim\mathsf{Expo}(#2)}}

\section{Labelled Stochastic Processes based on Petri Nets}
\label{sec:lsp}
In this section, we introduce the class of stochastic process models that can be analysed with our techniques. This class essentially builds on (bounded) generalised stochastic Petri nets (\gspn{s} \cite{ABCD95}), where  transitions (representing atomic units of work) are \cite{Mars88}:
\begin{compactitem}[$\bullet$]
\item \emph{immediate} or \emph{timed};
\item \emph{weighted}, where weights are used to define firing probabilities for immediate transitions, or firing delays for timed ones.
\end{compactitem}
Our model introduces two additional distinguishing features:
\begin{compactenum} 
\item transitions are \emph{silent} or \emph{visible}, in the latter case denoting \emph{process tasks}, and correspondingly come with a \emph{label} defining the task name;
\item we focus on \emph{finite traces}, progressing process instances through the net from an initial to a final state.
\end{compactenum}
This requires to go through the different features of \gspn{s}, carefully readapting them to our context. 

Before entering into the technical exposition, we give some preliminary notions on multisets. A multiset $a$ over a set $U$ is a function $a: U \rightarrow \mathbb{N}$, where for $u \in U$, $a(u)$ indicates the multiplicity (i.e., the number of occurrences) of $u$. Set $U$ is called the \emph{support} of the multiset $a$.

Given two multisets $a$ and $b$ over $U$, we define:
\begin{compactitem}[$\bullet$]
\item the \emph{union} of $a$ and $b$, denoted by $a+b$, as the multiset that assigns to each $u \in U$ multiplicity $a(u)+b(u)$;
\item that $a$ is \emph{included in} $b$, denoted by $a \leq b$, if for every $u \in U$, we have $a(u) \leq b(u)$;
\item assuming $a \leq b$, the \emph{difference} of $b$ and $a$, denoted by $b-a$, as the multiset that assigns to each $u \in U$ multiplicity $b(u) - a(u)$.
\end{compactitem}

The set of all multisets over $U$ is denoted by $\multisets{U}$. Multiset $a$ is explicitly represented by, in between squared brackets $[\ldots]$, each element $u$ with non-zero multiplicity, using notation $u^{a(u)}$.

\subsection{Labelled Petri Nets}
We capture the control-flow backbone of a work process using \textbf{Petri nets} with \textbf{labelled transitions}. Labels are used as follows:
\begin{compactitem}
\item \textbf{(labels as tasks)} labels generally represent (atomic) tasks to be executed within the process; 
\item \textbf{(silent transitions)} a special label is used to indicate that a transition is silent, i.e., does not correspond to any task.
\item \textbf{(repeated labels)} the same task can label distinct transitions.
\end{compactitem} 
We refer to a \emph{visible} transition if it is not silent. The firing of a visible transition conceptually captures the execution of its task; this is explicitly recorded as an event in a trace in an event log. The firing of a silent transition instead indicates the execution of an internal step, not perceived by the external environment, and thus not explicitly recorded. This is useful to capture a number of modelling patterns. We mention two:
\begin{compactitem}
\item \textbf{(control-flow patterns)} silent transitions can be employed to model internal control-flow patterns used to properly orchestrate visible transitions. Examples are skips, loopbacks, split and join points for concurrent branches, and more in general \emph{gateways} in contemporary process modelling notations (such as BPMN).
\item \textbf{(non-loggable tasks)} a silent transition can replace a visible one when its corresponding task cannot be logged. This is typically the case of a manual activity that is not backed up by a corresponding user-interaction activity, whose purpose is to inform the supporting information system  that the activity has been performed.
\end{compactitem}
Silent transitions are also relevant in their full generality in the context of process mining, since many process discovery algorithms produce as output a Petri net with silent transitions~\cite{DBLP:conf/apn/LeemansFA13,DBLP:conf/cec/BuijsDA12,AugustoCDRP19}.

To capture transition labels, we thus assume a given finite set $\tasks$ of tasks, and a special label $\hidden \not\in \tasks$ to indicate a silent step. The whole alphabet of labels is denoted by $\alphabet = \tasks \cup \set{\hidden}$. 

We recall the standard definition of labelled Petri nets, where we assume unweighted arcs for simplicity of presentation.
 
\begin{definition}[Labelled Petri net]
\label{def:pn}
A \emph{(labelled Petri) net} $\net$ is a tuple $\tup{\places, \transitions, \flow, \labeling}$, where:
\begin{compactitem}[$\bullet$]
\item $\places$ is a finite set of \emph{places};
\item $\transitions$ is a finite set of \emph{transitions}, disjoint from $\places$ (i.e., $\places \cap \transitions = \emptyset$);
\item $\flow \subseteq (\places \times \transitions) \cup (\transitions \times \places)$ is a \emph{flow relation} connecting places to transitions and transitions to places;
\item $\labeling: \transitions \rightarrow \alphabet$ is a \emph{labelling function} mapping each transition $\transition \in \transitions$ to a corresponding label $\labeling(\transition)$ that is either a task name from $\tasks$ or the silent label $\hidden$. 
\end{compactitem}
 \end{definition}
We adopt a \emph{dot notation} to extract the component of interest from a tuple. For example, given a net $\net$, its places are denoted by $\net.\places$. We will adopt the same notational convention for the other definitions as well. Given a net $\net$ and an element $x \in \net.\places \cup \net.\transitions$, the \emph{preset} and \emph{post-set} of $x$ are respectively denoted by $\pre{x} = \set{y \mid \tup{y,x} \in \flow}$ and $\post{x} = \set{y \mid \tup{x,y} \in \flow}$. If $x$ is a transition, then its pre- and post-set respectively denote its \emph{input} and \emph{output places}. 

We now turn to the execution semantics of nets, with the goal of formally describing how they represent the executions of process instances. We start with the notion of state. It is worth noting that our definitions are built on standard notions in Petri nets, such as that of marking, firing, run, reachability graph, but with a key difference: process instances evolve through \emph{finite} sequences of steps (of unbounded lenght) in the process. Focusing on finite runs and traces differ both conceptually and technically from the typical, infinite-trace execution semantics of Petri nets. 

As usual in Petri nets, an execution state of a net is described by a marking. Mathematically, a marking is a multiset of places. Conceptually, it represents a distribution of tokens over places. Each token denotes an execution thread, whose current local state is the place to which the token belongs. All tokens belong to a process instance (also called case), whose current global state is collectively described by the local state of each execution thread.

The distribution of token in a marking determines which transitions are enabled and, in turn, can fire and update the state of the process instance. Specifically, a transition is enabled in a marking if its input places contain at least one token each. Firing an enabled transition produces a new marking where one token per input place is consumed, and each output place gets one token more. This is formally defined using multiset operations as follows.

\begin{definition}[Marking]
\label{def:marking}
A \emph{marking} $\marking$ of a net $\net$ is a multiset over the places of $\net$, mapping each place $\place \in \net.\places$ to the number $\marking(\place)$ of tokens in $\place$.
\end{definition}

\begin{definition}[Transition enablement, transition firing]
\label{def:enabled}
Given a marking $\marking$ of a net $\net$, a transition $\transition \in \net.\transitions$ is \emph{enabled} in $\marking$, written $\enabledcat{\marking}{\transition}{\net}$, if $\pre{\transition} \leq \marking$. We denote by
$\enaset{\marking}{\net}$ the set of enabled transitions in a marking $\marking$.

Assuming $\enabledcat{\marking}{\transition}{\net}$, if $\transition$ \emph{fires} in $\marking$ of $\net$, a new marking $\marking'$ of $\net$ is produced, written $\firecat{\marking}{\transition}{\marking'}{\net}$, if $\marking' = (\marking - \pre{\transition}) + \post{\transition}$.
\end{definition}

Transition firings can be chained in a sequence, obtaining an execution. Before discussing how this is done, we need to introduce a different aspect of transitions: their timing.

\subsection{Immediate and Timed Transitions}
\label{sec:immediate-timed}
Orthogonally to the classification of transitions as visible or silent, we consider a second, temporal dimension, which comes from GNPSs \cite{ABCD95}: the distinction between immediate and timed transitions, and the interpretation of weights attached to them.

We start by extending Definition~\ref{def:pn} to accommodate these aspects. We stress that our definition presents some minor differences with other definitions of GSPNs in the literature, such as those in \cite{ABCD95,Mars88}. First, we do not consider general transition priorities, but we explicitly deal with priorities in the execution semantics. Second, we do not consider inhibitor arcs. This is just for simplicity of presentation; in fact, since we will focus on bounded GSPNs, inhibitor arcs can be seamlessly inserted without any impact on the computational properties of the model. At the same time, differently from the literature, our definition explicitly accounts for transition labels and, in particular, silent transitions, which is a minor change in the definition, but has a major conceptual and technical impact when it comes to the execution semantics and analysis of the stochastic behaviour of these nets.

\begin{definition}[Labelled, generalised stochastic Petri net]
\label{def:stochastic-net}
A \emph{labelled, generalised stochastic Petri net (\lgspn)} $\net$ is a tuple $\tup{\places, \transitions, \flow, \labeling, \weightfun}$, where:
\begin{compactitem}[$\bullet$]
\item $\tup{\places, \transitions, \flow, \labeling}$ is a labelled Petri net;
\item $\transitions = \itransitions \cup \ttransitions$ with $\itransitions \cap \ttransitions = \emptyset$, where $\itransitions$ is the set of \emph{immediate transitions}, and $\ttransitions$ is the set of \emph{timed transitions};
\item  $\weightfun\colon \transitions \to \mathbb{R}^+$ is a \emph{weight function} that assigns a positive number to each transition in $\net$.
\end{compactitem}
\end{definition}
With some abuse of terminology, we sometimes refer to an \lgspn by implicitly referring to its labelled Petri net.  

An immediate transition that becomes enabled fires instantaneously without advancing time, while a timed transition that becomes enabled fires with some delay. The question is then what happens if multiple transitions are enabled.

First and foremost, immediate transitions (being instantaneous) have \emph{priority} over timed ones. This calls for redefining the notion of enablement in Definition~\ref{def:enabled} as follows.

\begin{definition}[Transition enablement in a \lgspn]
\label{def:enabled}
Given a marking $\marking$ of a \lgspn $\net$, a transition $\transition \in \net.\transitions$ is \emph{enabled} in $\marking$, written $\enabledcat{\marking}{\transition}{\net}$, if $\transition$ is enabled in the sense of Definition~\ref{def:enabled}, and either:
\begin{compactitem}[$\bullet$]
\item $\transition \in \net.\itransitions$ (i.e., it is immediate), or
\item $\transition \in \net.\ttransitions$ (i.e. it is timed), and $\net.\itransitions \cap \enaset{\net}{\marking} = \emptyset$ (i.e., no immediate transition is enabled in $\marking$). 
\end{compactitem}
\end{definition}

If multiple immediate transitions are enabled, one is selected, resolving nondeterminism stochastically using the (relative) weights associated to the enabled transitions. If multiple timed transitions are enabled, one is selected based on a stochastic choice on which transition samples the shortest delay.

More specifically, the weight $\weightfun(\transition)=\rval{}$ of a timed transition $\transition$ denotes the rate of an exponentially distributed random variable describing the probability that, once enabled, the transition samples a certain delay. 
Technically, consider an exponential random variable $X$ with rate $\rval{}$, written $\expo{X}{\rval{}}$, denoting the firing delay of an enabled transition. The probability that such transition fires between $a$ and $b$ time units after becoming enabled is captured by
$$
\prob{a \leq X \leq b}{} = \int_{a}^{b}f_X(x)dx
$$
where
$$
f_X(x) = \begin{cases}
\rval{}e^{-\rval{}x} & x \geq 0
\\
0 & x < 0
\end{cases}
$$
is the probability density function of $X$.

In Section~\ref{sec:stochastic} we will use weights and rates to define the stochastic semantics of our model, recasting directly the original approach in \cite{ABCD95,Mars88}.

\paragraph{Representation of Nets}
\input{legend}
As customary in Petri net literature, we graphically represent \gspn{s} as bipartite directed graphs with place nodes depicted as circles, immediate transition nodes depicted as segments, and timed transition nodes depicted as squares. Figure~\ref{fig:legend} shows the main graphical conventions we adopt in decorating places and transitions with relevant information. The name of a place is indicated below the place circle, while that of a transition is shown inside the transition square, or at the bottom right of the transition segment. The weight (resp., rate) of an immediate (resp., timed) transition is shown on top of the transition icon, while its label is shown at the bottom. Silent transitions labelled by $\hidden$ are filled in black, while visible transitions labelled by a task from $\tasks$ are filled in light blue.

\begin{example}
\label{ex:sample-net}
Figure~\ref{fig:sample-net} shows a \lgspn for a simplified order-to-cash process. Each process instance refers to a distinct purchase order, evolved through visible transitions that capture tasks under the responsibility of the customer, and silent transitions denoting internal steps of the seller. Consistently with this distribution of responsibility, in a given state the order can be progressed exclusively by the customer or by the seller. We describe next intuitively how the process works.

A new process instance starts when the customer opens an order, and fills it with at least one item. Further items can be then added. Crucially, the seller determines if the customer is actually forced to add further items to progress the order, or whether instead adding further items is just an option. This is a free choice of the seller, which may be concretized in different ways (e.g., requiring certain items to be present together with other items). 
The first situation is determined when the seller executes the silent transition $\tname{s_1}$ (capturing a loopback), the second when the seller executes the silent transition $\tname{s_2}$. 

In the latter case, the customer can decide how to continue the process: 
\begin{inparaenum}[\itshape (i)]
\item by adding another item, in turn letting the seller choose whether a further item must or may be added;
\item by canceling the order, which leads to complete the process instance in the \emph{canceled} state (place $\pname{c}$);
\item by finalizing the order, signalling that the customer wants to proceed to the payment phase.
\end{inparaenum}
Whether a finalized order actually goes through the payment phase is determined by the seller, who decides whether the order is rejected or accepted. Both decisions have to be acknowledged by the customer. A rejection acknowledgement leads to conclude the process instance placing the order in the \emph{rejected} state (place $\pname{r}$). 

An acceptance acknowledgment gives control back to the customer, who can decide how to proceed:
\begin{inparaenum}[\itshape (i)]
\item by adding another item, which brings the order back to the previous phase, where it can be finalized, canceled, or filled with further items;
\item by canceling the order, which leads, as before, to complete the process instance in the \emph{canceled} state (place $\pname{c}$);
\item by paying the order (which we assume always succeeds for simplicity).
\end{inparaenum}
Upon payment, two tasks are executed by the seller concurrently (i.e., in no particular order): a receipt is emitted for the customer, and the order leaves the warehouse for delivery. Once both tasks are executed, the process instance successfully complete in the\emph{happy}, final state (place $\pname{h}$). Notice the usage of the two silent transitions $\tname{s_5}$ and $\tname{s_6}$ to respectively split and join the flow of control around the two concurrent tasks \taskname{emit receipt} and \taskname{ship}.

Throughout this example, transitions denoting tasks are all timed, whereas internal moves used for process orchestrations are captures as silent, immediate transitions. Weights and rates are represented symbolically. 
\end{example}

The \lgspn shown in Figure~\ref{fig:sample-net} only contains visible timed transitions and silent immediate ones. As shown next, this is not mandatory.

\input{sample-net}
\input{sample-net-silence}

\begin{example}
\label{ex:sample-net-silence}
Consider an information system supporting the execution of the net shown in Figure~\ref{fig:sample-net}, in a setting where the addition of items and the cancelation of orders cannot be logged (for example because they are manual tasks executed by the customer without interacting with the information system itself). To account for this key aspect, the two tasks \taskname{add item} and \taskname{cancel order} should be replaced by $\hidden$, turning the corresponding transitions into silent ones. This leads to the \lgspn shown in  Figure~\ref{fig:sample-net-silence}. Notice that the newly introduced silent transitions are all timed, as they denote the execution of non-logged tasks, each introducing some delay upon progressing the execution of a process instance.
\end{example}

\subsection{Labelled Stochastic Processes and their Finite-Trace Execution Semantics (Without Probabilities)}
\label{sec:lsp-semantics}

We now use \lgspn{s} as a basis for our model of \emph{labelled stochastic processes}, and describe their execution semantics based on finite traces. For the moment, we consider choices as purely nondeterministic, while in Section~\ref{sec:stochastic} we handle their stochastic behaviour.

Specifically, we are interested in using \lgspn{s} to describe the execution of process instances. This is done by fixing an initial marking describing the initial state of the process, and by executing the net starting from this initial marking and culminating in one of the final states of the process, which we define to be all deadlock markings.

\begin{definition}[Deadlock marking]
\label{def:deadlock}
A \emph{marking} $\marking$ of an \lgspn $\net$ is a \emph{deadlock marking} if no transition is enabled in it: $\enaset{\marking}{\net} = \emptyset$.
\end{definition}

\begin{definition}[Execution, supporting marking sequence]
\label{def:execution}
Let $\net$ be an \lgspn, and let $\marking_s$ and $\marking_f$ be two markings of $\net$. An \emph{execution} of $\net$ from $\marking_s$ to  $\marking_f$ is a (possibly empty) finite sequence $\transition_0,\ldots,\transition_n$ of transitions in $\net.\transitions$ such that there exists a corresponding sequence of markings $\marking_0,\ldots,\marking_{n+1}$ of $\net$ satisfying the following conditions:
\begin{inparaenum}[\itshape (i)]
\item $\marking_0 = \marking_s$,
\item $\marking_{n+1} = \marking_f$,
\item for every $i \in \set{0,\ldots,n}$, we have $\firecat{\marking_i}{\transition_i}{\marking_{i+1}}{\net}$.
\end{inparaenum}
We call the (unique) sequence $\marking_0,\ldots,\marking_{n+1}$ the \emph{supporting marking sequence} of $\transition_0,\ldots,\transition_n$.
\end{definition}

\begin{definition}[Labelled stochastic process]
A \emph{labelled stochastic process (\pwp)} is a triple $\tup{\net,\imarking,\fmarkings}$, where:
\begin{compactitem}[$\bullet$]
\item $\net$ is an \lgspn representing its \emph{supporting net};
\item $\imarking$ is a marking of $\net$ representing the \emph{initial marking};
\item $\fmarkings$ is a finite set of deadlock markings of $\net$ representing its \emph{final markings}.
\end{compactitem}
\end{definition}

\begin{example}
\label{ex:running-process}
From now on, we will consider, as a running example, the \pwp $\wproc_{\text{order}}$ defined as follows:
\begin{compactitem}
\item its supporting net $\wproc_{\text{order}}.\net$ is the \lgspn shown in Figure~\ref{fig:sample-net-silence};
\item its initial marking $\wproc_{\text{order}}.\imarking$ is the marking $[\place_s]$;
\item its final markings $\wproc_{\text{order}}.\fmarkings$ are the three markings $[\place_h]$, $[\place_r]$, and $[\place_c]$, representing the happy completion of the process where the order is paid and shipped, and the unsuccessful completions where the order is either rejected or canceled.
\end{compactitem}
\end{example}

We will interchangeably assign markings to nets and to \pwp{s}, simply meaning the following: a marking of \pwp $\wproc$ is a marking of $\wproc.\net$.

Runs are defined as \pwp executions linking the initial state to one of the final states.

\begin{definition}[Run of an \pwp]
\label{def:run}
A run of an \pwp $\wproc$ is an execution (as per Definition~\ref{def:execution}) that starts in the initial marking $\wproc.\imarking$ of $\wproc$ and ends in one of its final markings $\wproc.\fmarkings$.
\end{definition}

\begin{example}
\label{ex:runs}
Consider the \pwp $\runningwproc$ from Example~\ref{ex:running-process}. The transition sequence
$$
\run_{iair} = \tname{o},\tname{i_1},\tname{s_2},\tname{f},\tname{s_4},\tname{a},\tname{i_3},\tname{f},\tname{s_3},\tname{r}
$$
is a run of $\runningwproc$, leading from the initial marking $[\pname{s}]$ to the final, rejection marking $[\pname{r}]$. It does so by adding an item to the order, passing once through acceptance, then adding an other item, finally being rejected.

The transition sequence:
$$
\run_{iiair} = \tname{o},\tname{i_1},\tname{s_2},\tname{i_2},\tname{s_2},\tname{f},\tname{s_4},\tname{a},\tname{i_3},\tname{f},\tname{s_3},\tname{r}
$$ 
is another run of $\runningwproc$, which resembles $\run_{iair}$ with the only difference that two items are added before the first acceptance (later turned in a rejection).
\end{example}

When logging runs, silent transitions disappear, while visible ones leave their corresponding task as a footprint. We capture this by lifting runs to traces: each trace $\trace = e_0,\ldots,e_n \in \tasks^*$ is a sequence of \emph{events} over $\tasks$ where, for simplicity, each event $e_i$ simply tracks the task that has been executed, together with its qualitative position in the trace with respect to the other events. A trace is a \emph{model trace} for an \pwp if it is produced by one of its runs, stripping off the silent transitions and considering the labels of the visible ones, keeping their relative positioning.

\begin{definition}[Model trace, induced trace]
\label{def:model-trace}
A trace $\trace$ is a \emph{model trace} of \pwp $\wproc$ if there exists a run $\run = \transition_0,\ldots,\transition_m$ of $\wproc.\net$ whose corresponding sequence of labels $\wproc.\net.\labeling(\run) = \wproc.\net.\labeling(\transition_0),\ldots,\wproc.\net.\labeling(\transition_m)$ coincides with $\trace$ once all $\hidden$ elements are removed. In this case, we say that $\run$ \emph{induces} $\trace$. 
\end{definition}

The following remark highlights one of the main phenomena arising when considering silent transitions. This will become a key challenge when reasoning over \pwp{s} and their stochastic behaviour.

\begin{remark}
\label{rem:runs-traces}
A model trace $\trace$ of a (bounded) \pwp $\wproc$ can in general be induced by multiple, possibly infinitely many, runs of $\wproc$.
\end{remark}

We denote the (possibly infinite) set of runs of \pwp $\wproc$ that induce trace $\trace$ by $\runs{\trace}{\wproc}$. The following example substantiates Remark~\ref{rem:runs-traces}. 

\begin{example}
\label{ex:traces}
Consider the two runs $\run_{iair}$ and $\run_{iiair}$ from Example~\ref{ex:runs}. They both induce the following trace of $\runningwproc$:
$$
\trace_{ar} = \taskname{open},\taskname{finalize},\taskname{ack accept},\taskname{finalize},\taskname{ack reject}
$$
In fact, the addition of items is always handled via silent transitions, hence being completely invisible at the trace level.

More in general, trace $\trace_{ar}$ is induced by the infinitely many runs captured through the following regular expression:
$$
\runs{\trace_{ar}}{\runningwproc} = 
\left\{
\run 
\mid 
\run \text{ matches }
\left(
\begin{array}{@{}l@{}}
\taskname{open}
;
(\taskname{add item})^+
;
\taskname{finalize}
;\\
(
\taskname{ack accept};
(\taskname{add item})^+;
\taskname{finalize}
)^*
;\\
\taskname{ack reject}
\end{array}
\right)
\right\}
$$
where 
\begin{inparaenum}[\it (i)]
\item the atomic regular expression is a task name from $\tasks$;
\item expression ``$e_1;e_2$'' is the concatenation of expression $e_1$ followed by expression $e_2$;
\item ``$e^*$'' is the Kleene star-closure of expression $e$, capturing the set of regular expressions where $e$ is repeated zero or more times; 
\item ``$e^+$'' captures the set of expressions where $e$ is repeated at least once. 
\end{inparaenum}
\end{example}

The execution semantics of an \pwp is described by a (possibly infinite-state) labelled transition system typically named a \emph{reachability graph}. 
In our setting, labelled transition systems are deterministic, in the sense that given a state and a label, there can be at most one transition having that state as source and that label attached. This will be essential when turning to the stochastic behaviour of \pwp{s}.

\begin{definition}[Labelled transition system, run]
\label{def:lts}
A \emph{(deterministic) labelled transition system} $\tsys{}$ is a tuple $\tup{\labels,\states,\istate,\fstates,\transrel}$ where:
\begin{compactitem}[$\bullet$]
\item $\labels$ is a finite set of \emph{labels};
\item $\states$ is a (possibly infinite) set of \emph{states};
\item $\istate \in \states$ is the \emph{initial state};
\item $\fstates \subseteq \states$ is the set of \emph{accepting states};
\item $\transrel: (\states \times \labels) \rightarrow \states$ is a  \emph{transition function}, that is, a partial function that, given a state $\state \in \states$ and a label $l \in \labels$, is either undefined or returns a single successor state $\state' = \transrel(\state,l)$;
\item $\transrel$ is undefined for states in $\fstates$ (that is, final states to not have successors).
\end{compactitem}
A \emph{run} of $\tsys{}$ is a finite (possibly empty) sequence $l_0 \ldots l_n$ such that there exists a corresponding sequence $r_0,\ldots,r_{n+1}$ of states in $\states$ satisfying the following properties: 
\begin{inparaenum}[\itshape (i)]
\item $r_0 = \istate$;
\item $r_{n+1} \in \fstates$;
\item for every $i \in \set{0,\ldots,n}$, we have that $\transrel(r_i,l_i)$ is defined, and $\transrel(r_i,l_i) = r_{i+1}$.
\end{inparaenum}
\end{definition}

When defining the execution semantics of an \pwp, the resulting reachability graph is a labelled transition system with the following charachteristics:
\begin{inparaenum}[\itshape (i)] 
\item labels are built from the \pwp transitions;
\item states correspond to markings of the \pwp that are reachable from the initial state of the \pwp;
\item initial and accepting states correspond to initial and final markings of the \pwp;
\item transitions match transition firings of the \pwp, keeping track of the transition name for provenance.
\end{inparaenum}
The distinction between labels of the \pwp (task names or $\hidden$), and those of the reachability graph (transition names paired with their label) is key to guaranteeing determinism of the reachability graph itself. In fact, labelling the transition system with the labels of the \pwp would introduce nondeterminism.

\begin{definition}[Reachability graph]
\label{def:rg}
The \emph{reachability graph} $\rg{\wproc}$ of an \pwp $\wproc$ is a labelled transition system $\tup{\labels,\states,\istate,\fstates,\transrel}$ with $\labels = \wproc.\net.\transitions$ and whose other components are defined by mutual induction as the minimal sets satisfying the following conditions:
\begin{compactenum}
\item $\istate = \wproc.\imarking$, $\istate \in \states$;
\item for every state $\marking \in \states$, every transition $\transition \in \wproc.\net.\transitions$ with label $l = \wproc.\net.\labeling(t)$, and every marking $\marking' \in \multisets{\wproc.\net.\places}$, if $\firecat{\marking}{\transition}{\marking'}{\wproc.\net}$ we have that
\begin{inparaenum}[(a)]
\item $\marking' \in \states$;
\item if $\marking' \in \wproc.\fmarkings$, then $\marking' \in \fstates$;
\item $\transrel(\marking,\transition) = \marking'$.
\end{inparaenum}
\end{compactenum}
\end{definition}

It will be useful later to refer to outgoing transitions from a given state $\state$. We do so by defining: 
$\outgoing{\state}{\rg{\wproc}} = \set{\state' \mid \rg{\wproc}.\transrel(\state,(t,l))=\state'   \text{ for some } t,l}.
$

Given an \pwp $\wproc$, the runs of $\rg{\wproc}$ (in the sense of Definition~\ref{def:lts}) match with all and only the runs of $\wproc$ (in the sense of Definition~\ref{def:run}). In addition, due to our requirement that all final markings are deadlock markings, accepting states of $\rg{\wproc}$ have no outgoing transitions either.

In general, $\rg{\wproc}$ may contain infinitely many states, if it is possible to unconditionally fire transitions to create bigger and bigger markings. When capturing work processes, this is undesired, as the control flow characterising the possible evolutions of each single process instance is expected to be finite-state. In Petri net terms, this means that work processes should be represented by \textbf{bounded} nets, as defined next.

\begin{definition}[Bounded \pwp]
\label{def:bounded}
An \pwp $\wproc$ is \emph{bounded} if there exists a number $k$ such that, for every reachable marking $\marking \in \rg{\wproc}.\states$ and every place $\place \in \wproc.\net.\places$, we have $\marking(\place) \leq k$.
\end{definition}
As desired, a \bpwp induces a reachability graph that has finitely many states. Boundedness is a standard property assumed when capturing work processes. Verifying boundedness is decidable for Petri nets~\cite{Desel1998} and well-known techniques exist. It is instead undecidable for \gspn{s}, due to the implicit priority of immediate transitions over timed ones, but becomes decidable if these two types of transitions are separated, in the following sense: considering transition enablement only based on the presence of enough tokens in the transition input places, in every reachable marking it is never the case that both some timed and some immediate transition are enabled. The \lgspn{s} of Figures~\ref{fig:sample-net} and~\ref{fig:sample-net-silence} indeed satisfy this separation. 

Since a \bpwp reaches finitely many distinct markings, it also reaches finitely many distinct deadlock ones. We then single out the class of complete \bpwp{s}, which exhaustively consider all their deadlock markings as final ones.

\begin{definition}
\label{def:complete}
A \bpwp $\wproc$ is \emph{complete} if the set $\wproc.\fstates$ of its final markings coincide with the set of deadlock markings in $\rg{\wproc}$.
\end{definition}

In the remainder of this paper, \emph{we focus on \bpwp{s}, always implicitly assuming that they are complete}. This is just for convenience, and we remark that all the technical results presented hereafter seamlessly carry over the whole class of \bpwp{s}. 

\input{sample-rg}

\begin{example}
Without considering, for the moment, $\rho$-attachments on edges, Figure~\ref{fig:sample-rg} shows the finite reachability graph of the \pwp $\runningwproc$, which is 1-bounded. It also witnesses that $\runningwproc$ is a complete \bpwp.
\end{example}

We close this section by identifying undesired reachable markings that cannot be progressed to reach one of the final markings of the \pwp of interest. They will be of special interest when reasoning over \bpwp{s} and their stochastic behaviour. 
%
%
%
%
Such markings resemble the well-known notion of \emph{livelock}.

\begin{definition}[Livelock marking]
\label{def:livelock}
A marking $\state \in \rg{\wproc}.\states$ of a \bpwp $\wproc$ is a \emph{livelock marking} if it is not a deadlock marking, and there is no execution that starts from $\state$ and leads to some deadlock marking.
\end{definition}

Since \bpwp{s} have finitely many reachable markings, the only possibility of creating livelocks is to recur over the same set of markings, without ever having the possibility of exiting from the loop to reach a deadlock marking. 

The following example illustrates these notions.

\input{livelock}

\begin{example}
\label{ex:livelock-deadend}
Consider the \pwp $\wproc_{\text{live}}$  of Figure~\ref{fig:livelock-deadend}. The \pwp contains an initial choice. If one of the two transitions $\tname{01}$ and $\tname{04}$ is selected, the process instance completes in the only final marking. If $\tname{02}$ is selected, the process instance enters into a part of the process where it is forced to progress forever - a livelock situation. By looking at the reachability graph, it is in fact clear that the \pwp contains two livelock markings - when a token is assigned either to $\pname{2}$ or $\pname{3}$.
\end{example}

\subsection{Labelled Stochastic Processes and their Stochastic Behaviour}
\label{sec:stochastic}

We finally extend \pwp{s} with stochastic behaviour, by incorporating \textbf{stochastic decision making to determine which enabled transition to fire}. Technically, given a marking, we want to obtain a probability distribution over the transitions that are enabled therein, and use it to choose which one is fired next. This is where transition weights come into play \cite{Mars88}. 

The resulting execution semantics is then defined through a stochastic transition system.

\begin{definition}[Stochastic transition system]
\label{def:stochastic-ts}
A \emph{stochastic transition system} is a tuple $\tup{\labels,\states,\istate,\fstates,\transrel,\probfun}$ where $\tup{\labels,\states,\istate,\fstates,\transrel}$ is a labelled transition system (in the sense of Definition~\ref{def:lts}), while $\probfun: \transrel \rightarrow [0,1]$ is a \emph{transition probability function} mapping each transition in $\transrel$ to a corresponding probability value in $[0,1]$, such that for every non-deadlock marking $\state \in \states$, $\sum_{\xi=\tup{\state_1,\transition,\state_2} \in \transrel \text{~s.t.~}\state_1=\state} \probfun(\xi) = 1$.  
\end{definition}

To define the stochastic execution semantics of \pwp, we then need to indicate how to enrich its reachability graph with a transition probability function. For immediate transitions, weights are directly used to define the relative likelihood that one among the enabled transitions will fire: each firing probability is in fact simply obtained as the ratio between the weight of the corresponding transition and the sum of weights of all enabled transitions. 

An essential aspect of \gspn{s} is that exactly the same probability distribution is obtained also in the case of timed transitions. We briefly recall why this is the case.
Consider two exponential random variables $\expo{X}{\lambda}$ and $\expo{Y}{\mu}$.
We are interested in computing the probability that the delay sampled by $X$ is shorter than that sampled by $Y$, which in turn means that the transition associated by $X$ will \emph{fire first}. By recalling Section~\ref{sec:immediate-timed}, noticing that the two exponential distributions are independent from each other, and that they are 0 for negative values, we get $\prob{X\leq Y}{} =$
\begin{align*}
&\phantom{={}} \int_{0}^{\infty}\int_{x}^{\infty} f_X(x)f_Y(y) dydx 
= \int_{0}^{\infty}\int_{x}^{\infty} \lambda e^{-\lambda x} \mu e^{-\mu y} dydx =
\\
&= \int_{0}^{\infty}\lambda e^{-\lambda x} \int_{x}^{\infty} \mu e^{-\mu y} dydx
= \int_{0}^{\infty}\lambda e^{-\lambda x} \Big[ -e^{-\mu y} \Big]_x^\infty dx =
\\
&= \int_{0}^{\infty}\lambda e^{-\lambda x}e^{-\mu x}dx
= \int_{0}^{\infty}\lambda e^{-(\lambda + \mu)x}dx
= \left[ - \frac{\lambda}{\lambda + \mu}e^{-(\lambda+\mu)x}\right]_0^\infty = \frac{\lambda}{\lambda + \mu}\\ 
\end{align*}
Consequently, we get the following approach to compute the firing probability, which homogeneously applies to immediate and timed transitions.

\begin{definition}[Firing probability]
Let $\wproc$ be an \pwp and $\marking$ be a marking of $\wproc$.
The \emph{firing probability} of a transition $t \in \wproc.\net.\transitions$ in $\marking$, written $\prob{\transition \mid \marking}{\wproc}$, is defined as:
$$\cprob{\transition}{\marking}{\wproc} = 
\begin{cases}
\frac{\weightfun(\transition)}{\sum_{\transition' \in \enaset{\marking}{\wproc.\net}}\weightfun(t')}
 & \transition \in \enaset{\marking}{\wproc.\net}\\
0 & \text{otherwise}
\end{cases} 
$$
\end{definition}

The execution semantics of an \pwp is then defined by extending its reachability graph with a transition probability function based on firing probabilities.

\begin{definition}[Stochastic reachability graph]
\label{def:rg}
The \emph{stochastic reachability graph} $\srg{\wproc}$ of an \pwp $\wproc$ is a stochastic labelled transition system $\tup{\labels,\states,\istate,\fstates,\transrel,\probfun}$ where $\tup{\labels,\states,\istate,\fstates,\transrel} = \rg{\wproc}$, and $\probfun$ is defined as follows: for every transition $\xi = \tup{\marking,\transition,\marking'}\in \transrel$, we have that $\probfun(\xi) = \prob{\transition \mid \marking}{\wproc}$.
\end{definition}

\begin{example}
Consider the \pwp $\runningwproc$ from Example~\ref{ex:running-process}, whose supporting net is shown in Figure~\ref{fig:sample-net-silence}. Figure~\ref{fig:sample-rg} captures the stochastic reachability graph $\srg{\runningwproc}$ (using the symbolic weights/rates shown in Figure~\ref{fig:sample-net-silence}). If we fix the weight $\weightfun(\tname{s_4})$ of transition $\tname{s_4}$ to 80, and the weight $\weightfun(\tname{s_3})$ of transition $\tname{s_3}$ to 20, we are representing that, whenever a decision has to be taken on whether to accept or reject an order, then there is $0.2$ chance that the decision is of rejection, and $0.8$ chance that the decision is of acceptance.
\end{example}

A final step in the definition of the execution semantics of \pwp{s}, is to lift firing probabilities to the level of runs and traces. As customary in \gspn{s}, since every firing contained in a run corresponds to an independent choice, we simply need to compute the product of the firing probabilities.

\begin{definition}[Run probability]
The probability that a \pwp $\wproc$ produces a run $\run = t_0,\ldots,t_n$ of $\wproc$, called the \emph{run probability of $\run$ according to $\wproc$}, is: 
$$\prob{\run}{\wproc} = \prod_{i=1}^{n}\cprob{t_i}{\marking_{i}}{\wproc}$$
where $m_0,\ldots,m_{n+1}$ is the supporting marking sequence for $\run$ (cf.~Definition \ref{def:execution}). The same notion is equivalently defined given as input a stochastic transition system.
\end{definition}

The probability of a trace is then simply obtained by summing up the probabilities of all the runs inducing that trace. 

\begin{definition}[Trace probability]
\label{def:trace-probability}
The probability that a \pwp $\wproc$ produces a model trace $\trace$ of $\wproc$, called the \emph{trace probability of $\run$ according to $\wproc$}, is: 
$$\prob{\trace}{\wproc} = \sum_{\run \in \runs{\trace}{\wproc}} \prob{\run}{\wproc}$$
\end{definition}
The central problem we are now facing, which will be tackled in the second part of the article, is that due to Remark~\ref{rem:runs-traces}, this is in general an infinite sum, even in the case of bounded \pwp{s}.

\begin{example}
\label{ex:infinite-sum}
Consider the model trace $\trace_{ar}$ of $\runningwproc$ from Example~\ref{ex:traces}. The trace probability $\prob{\trace_{ar}}{\runningwproc}$ cannot be directly computed, as it requires to sum the probabilities of the infinitely many runs $\runs{\trace_{ar}}{\runningwproc}$. 
\end{example}

%% file: legend.tex
\begin{figure}[t]
\centering
\resizebox{.72\textwidth}{!}{
\begin{tikzpicture}[node distance=10mm and 10mm]

\node[
  place={$q$},
] (p) {};

\node[
  right = 5mm of p.east,
  anchor = west
] (p-label) {place with name $q$};

\node[ 
  wtransition={a}{$w$},
  immediate,
  below=of p,
  label={[tlabel]right:$t$}
] (wt-i) {};

\node[
  anchor = west
] (wt-i-label) at (p-label.west|-wt-i)
{(visible) immediate transition named $t$, with weight $w$ and label \taskname{a}};

\node[ 
  wsilenttransition={$w$},
  immediate,
  below=of wt-i,
  label={[tlabel]right:$t$},
] (wst-i) {};

\node[
  anchor = west
] (wst-i-label) at (p-label.west|-wst-i)
{silent immediate transition named $t$, with weight $w$};

\node[ 
  wtransition={a}{$\lambda$},
  timed,
  below=of wst-i,
] (wt-t) {$t$};

\node[
  anchor = west
] (wt-t-label) at (p-label.west|-wt-t)
{(visible) timed transition named $t$, with rate $\lambda$ and label \taskname{a}};

\node[ 
  wsilenttransition={$\lambda$},
  timed,
  below=of wt-t,
] (wst-t) {$t$};

\node[
  anchor = west
] (wst-t-label) at (p-label.west|-wst-t)
{silent timed transition named $t$, with rate $\lambda$};

\end{tikzpicture}
}
\caption{Notation for \spwp{s}, considering their three main dimensions: transition labels, silent vs visible transitions, timed vs immediate transitions}
\label{fig:legend}
\end{figure}

%% file: sample-net.tex
\begin{figure}[t]
\centering
\resizebox{\textwidth}{!}{
\begin{tikzpicture}[node distance=8mm and 5mm]

\node (qs) [
  place = {\pname{s}}
] {};
			
\node (to) [
  wtransition = {\taskname{open}}{\rval{o}},
  timed, 
  right = of qs,
] {\tname{o}};

\draw[arc] (qs) edge (to);

\node (q1) [
  place = {\pname{1}},
  right = of to,
] {};

\draw[arc] (to) edge (q1);
			
\node (ti1) [
  wtransition = {add item}{\rval{i_1}},
  timed, 
  right = of q1,
] {\tname{i_1}};

\draw[arc] (q1) edge (ti1);

\node (q2) [
  place = {\pname{2}},
  right = of ti1,
] {};

\draw[arc] (ti1) edge (q2);

\node (ts1) [
  wsilenttransition = {\wval{s_1}},
  immediate, 
  above = of ti1,
  label = {[tlabel]right:\tname{s_1}}
] {};

\draw[arc] (q2.north)++(-1mm,-.3mm) |- (ts1.east);
\draw[arc] (ts1.west) -| (q1);

\node (ts2) [
  wsilenttransition = {\wval{s_2}},
  immediate, 
  right = of q2,
  label = {[tlabel]right:\tname{s_2}}
] {};

\draw[arc] (q2) edge (ts2);

\node (q3) [
  place = {\pname{3}~~~~~~},
  right = of ts2,
] {};

\draw[arc] (ts2) edge (q3);

\node (ti2) [
  wtransition = {add item}{\rval{i_2}},
  timed, 
  above = of ts2,
] {\tname{i_2}};

\draw[arc] (q3.north)++(-1mm,-.3mm) |- (ti2.east);
\draw[arc] (ti2.west) -| ($(q2.north)+(1mm,-.3mm)$);

\node (tf) [
  wtransition = {finalize}{\rval{f}},
  timed, 
  right = of q3,
] {\tname{f}};

\draw[arc] (q3) edge (tf);

\node (q4) [
  place = {\pname{4}~~~~~~},
  right = of tf,
] {};

\draw[arc] (tf) edge (q4);

\node (ts3) [
  wsilenttransition = {\wval{s_3}},
  immediate, 
  right = 2mm of q4,
  yshift = -9mm,
  label = {[tlabel]right:\tname{s_3}}
]  {};

\draw[arc] (q4) |- (ts3);

\node (q5) [
  place = {\pname{5}},
  right = of ts3,
] {};

\draw[arc] (ts3) edge (q5);

\node (ts4) [
  wsilenttransition = {\wval{s_4}},
  immediate,
  right = 2mm of q4,
  yshift = 9mm,
  label = {[tlabel]right:\tname{s_4}}
] {};

\draw[arc] (q4) |- (ts4);

\node (q6) [
  place = {\pname{6}},
  right = of ts4,
] {};

\draw[arc] (ts4) edge (q6);			

\node (ta) [
  wtransition = {ack accept}{\rval{a}},
  timed, 
  right = of q6,
] {\tname{a}};

\draw[arc] (q6) edge (ta);

\node (q7) [
  place = {~~~~~~\pname{7}},
  right = of ta,
] {};

\draw[arc] (ta) edge (q7);

\node (ti3) [
  wtransition = {add item}{\rval{i_3}},
  timed, 
  above = of ta,
] {\tname{i_3}};

\draw[arc] (q7) |- (ti3);
\draw[arc] (ti3) -| ($(q3.north)+(1mm,-.3mm)$);

\node (tp) [
  wtransition = {pay}{\rval{p}},
  timed, 
  right = of q7,
] {\tname{p}};
	
\draw[arc] (q7) edge (tp);

\node (q8) [
  place = {\pname{8}~~~~~~},
  right = of tp,
] {};
	
\draw[arc] (tp) edge (q8);

\node (ts5) [
  wsilenttransition = {\wval{s_5}~~~~~~},
  immediate,
  right = of q8,
  label = {[tlabel,xshift=2.5mm]left:\tname{s_5}}
] {};
	
\draw[arc] (q8) edge (ts5);

\node (q9) [
  place = {\pname{9}},
  right = 2mm of ts5,
  yshift = 8mm, 
] {};

\draw[arc] (ts5) edge (q9);

\node (q10) [
  place = {\pname{10}},
  right = 2mm of ts5,
  yshift = -8mm, 
] {};

\draw[arc] (ts5) edge (q10);

\node (te) [
  wtransition = {emit receipt}{\rval{e}},
  timed, 
  right = of q9,
] {\tname{e}};

\draw[arc] (q9) edge (te);

\node (tl) [
  wtransition = {leave}{\rval{l}},
  timed, 
  right = of q10,
] {\tname{l}};

\draw[arc] (q10) edge (tl);

\node (q11) [
  place = {\pname{11}},
  right = of te,
] {};

\draw[arc] (te) edge (q11);

\node (q12) [
  place = {\pname{12}},
  right = of tl,
] {};

\draw[arc] (tl) edge (q12);

\node (ts6) [
  wsilenttransition = {\wval{s_6}},
  immediate,
  right = 2mm of q11,
  yshift = -8mm,
  label = {[tlabel]right:\tname{s_6}}
] {};

\draw[arc] (q11) edge (ts6);
\draw[arc] (q12) edge (ts6);

\node (qh) [
  place = {\pname{h}},
  right = of ts6,
] {};

\draw[arc] (ts6) edge (qh);

\node (tr) [
  wtransition = {ack reject}{\rval{r}},
  timed, 
  right = of q5,
] {\tname{r}};

\draw[arc] (q5) edge (tr);

\node (qr) [
  place = {\pname{r}},
] at (tr-|qh) {};

\draw[arc] (tr) edge (qr);

\node (tc1) [
  wtransition = {cancel}{\rval{c_1}},
  timed, 
  below = of tf,
  yshift = -12mm
] {\tname{c_1}};

\draw[arc] (q3) |- (tc1);

\node (qc) [
  place = {\pname{c}},
] at (tc1-|qh) {};

\draw[arc] (tc1) edge (qc);

\node (tc2) [
  wtransition = {cancel}{\rval{c_2}},
  timed, 
  below = of tp,
  yshift = -11mm
] {\tname{c_2}};

\draw[arc] (q7) |- (tc2);
\draw[arc] (tc2) -| (qc);

\end{tikzpicture}
}
\caption{Stochastic net of an order-to-cash process. Weights and rates are presented symbolically}
\label{fig:sample-net}
\end{figure}

%% file: sample-net-silence.tex

\begin{figure}[t]
\centering
\resizebox{\textwidth}{!}{
\begin{tikzpicture}[node distance=8mm and 5mm]

\node (qs) [
  place = {\pname{s}}
] {};
			
\node (to) [
  wtransition = {\taskname{open}}{\rval{o}},
  timed, 
  right = of qs,
] {\tname{o}};

\draw[arc] (qs) edge (to);

\node (q1) [
  place = {\pname{1}},
  right = of to,
] {};

\draw[arc] (to) edge (q1);
			
\node (ti1) [
  wsilenttransition = {\rval{i_1}},
  timed, 
  right = of q1,
] {\tname{i_1}};

\draw[arc] (q1) edge (ti1);

\node (q2) [
  place = {\pname{2}},
  right = of ti1,
] {};

\draw[arc] (ti1) edge (q2);

\node (ts1) [
  wsilenttransition = {\wval{s_1}},
  immediate, 
  above = of ti1,
  label = {[tlabel]right:\tname{s_1}}
] {};

\draw[arc] (q2.north)++(-1mm,-.3mm) |- (ts1.east);
\draw[arc] (ts1.west) -| (q1);

\node (ts2) [
  wsilenttransition = {\wval{s_2}},
  immediate, 
  right = of q2,
  label = {[tlabel]right:\tname{s_2}}
] {};

\draw[arc] (q2) edge (ts2);

\node (q3) [
  place = {\pname{3}~~~~~~},
  right = of ts2,
] {};

\draw[arc] (ts2) edge (q3);

\node (ti2) [
  wsilenttransition = {\rval{i_2}},
  timed, 
  above = of ts2,
] {\tname{i_2}};

\draw[arc] (q3.north)++(-1mm,-.3mm) |- (ti2.east);
\draw[arc] (ti2.west) -| ($(q2.north)+(1mm,-.3mm)$);

\node (tf) [
  wtransition = {finalize}{\rval{f}},
  timed, 
  right = of q3,
] {\tname{f}};

\draw[arc] (q3) edge (tf);

\node (q4) [
  place = {\pname{4}~~~~~~},
  right = of tf,
] {};

\draw[arc] (tf) edge (q4);

\node (ts3) [
  wsilenttransition = {\wval{s_3}},
  immediate, 
  right = 2mm of q4,
  yshift = -9mm,
  label = {[tlabel]right:\tname{s_3}}
]  {};

\draw[arc] (q4) |- (ts3);

\node (q5) [
  place = {\pname{5}},
  right = of ts3,
] {};

\draw[arc] (ts3) edge (q5);

\node (ts4) [
  wsilenttransition = {\wval{s_4}},
  immediate,
  right = 2mm of q4,
  yshift = 9mm,
  label = {[tlabel]right:\tname{s_4}}
] {};

\draw[arc] (q4) |- (ts4);

\node (q6) [
  place = {\pname{6}},
  right = of ts4,
] {};

\draw[arc] (ts4) edge (q6);			

\node (ta) [
  wtransition = {ack accept}{\rval{a}},
  timed, 
  right = of q6,
] {\tname{a}};

\draw[arc] (q6) edge (ta);

\node (q7) [
  place = {~~~~~~\pname{7}},
  right = of ta,
] {};

\draw[arc] (ta) edge (q7);

\node (ti3) [
  wsilenttransition = {\rval{i_3}},
  timed, 
  above = of ta,
] {\tname{i_3}};

\draw[arc] (q7) |- (ti3);
\draw[arc] (ti3) -| ($(q3.north)+(1mm,-.3mm)$);

\node (tp) [
  wtransition = {pay}{\rval{p}},
  timed, 
  right = of q7,
] {\tname{p}};
	
\draw[arc] (q7) edge (tp);

\node (q8) [
  place = {\pname{8}~~~~~~},
  right = of tp,
] {};
	
\draw[arc] (tp) edge (q8);

\node (ts5) [
  wsilenttransition = {\wval{s_5}~~~~~~},
  immediate,
  right = of q8,
  label = {[tlabel,xshift=2.5mm]left:\tname{s_5}}
] {};
	
\draw[arc] (q8) edge (ts5);

\node (q9) [
  place = {\pname{9}},
  right = 2mm of ts5,
  yshift = 8mm, 
] {};

\draw[arc] (ts5) edge (q9);

\node (q10) [
  place = {\pname{10}},
  right = 2mm of ts5,
  yshift = -8mm, 
] {};

\draw[arc] (ts5) edge (q10);

\node (te) [
  wtransition = {emit receipt}{\rval{e}},
  timed, 
  right = of q9,
] {\tname{e}};

\draw[arc] (q9) edge (te);

\node (tl) [
  wtransition = {leave}{\rval{l}},
  timed, 
  right = of q10,
] {\tname{l}};

\draw[arc] (q10) edge (tl);

\node (q11) [
  place = {\pname{11}},
  right = of te,
] {};

\draw[arc] (te) edge (q11);

\node (q12) [
  place = {\pname{12}},
  right = of tl,
] {};

\draw[arc] (tl) edge (q12);

\node (ts6) [
  wsilenttransition = {\wval{s_6}},
  immediate,
  right = 2mm of q11,
  yshift = -8mm,
  label = {[tlabel]right:\tname{s_6}}
] {};

\draw[arc] (q11) edge (ts6);
\draw[arc] (q12) edge (ts6);

\node (qh) [
  place = {\pname{h}},
  right = of ts6,
] {};

\draw[arc] (ts6) edge (qh);

\node (tr) [
  wtransition = {ack reject}{\rval{r}},
  timed, 
  right = of q5,
] {\tname{r}};

\draw[arc] (q5) edge (tr);

\node (qr) [
  place = {\pname{r}},
] at (tr-|qh) {};

\draw[arc] (tr) edge (qr);

\node (tc1) [
  wsilenttransition = {\rval{c_1}},
  timed, 
  below = of tf,
  yshift = -12mm
] {\tname{c_1}};

\draw[arc] (q3) |- (tc1);

\node (qc) [
  place = {\pname{c}},
] at (tc1-|qh) {};

\draw[arc] (tc1) edge (qc);

\node (tc2) [
  wsilenttransition = {\rval{c_2}},
  timed, 
  below = of tp,
  yshift = -11mm
] {\tname{c_2}};

\draw[arc] (q7) |- (tc2);
\draw[arc] (tc2) -| (qc);

\end{tikzpicture}
}
\caption{Stochastic net obtained from that of Figure~\ref{fig:sample-net} in the case where the activities \taskname{add item} and \taskname{cancel} are not logged.}
\label{fig:sample-net-silence}
\end{figure}

%% file: sample-rg.tex
\begin{figure}[t]
\centering
\resizebox{\textwidth}{!}{
\begin{tikzpicture}[node distance=5mm and 11mm]

\node [
  lstate={\sstate{0}}{[\pname{s}]}
] (s0) {};
\draw
  ($(s0.north west)+(-3mm,3mm)$)
  edge[arc]
  (s0.north west);

\node [
  lstate={\sstate{1}}{[\pname{1}]},
  below=of s0,
] (s1) {};

\draw[arc]
  (s0)
  edge
  node[right] {(\tname{o},\taskname{open})}
  node[left] {1}  
  (s1);

\node [
  lstate={\sstate{2}}{[\pname{2}]},
  right=of s1,
] (s2)  {};

\draw[arc]
  (s1)
  edge[silent]
  node[above] {(\tname{i_1},\hidden)}
  node[below] {1}  
  (s2);

\draw[arc,out=-100,in=-80]
  (s2)
  edge[silent]
  node[above] {(\tname{s_1},\hidden)}
  node[below] {$\pval{s_1} = \frac{\wval{s_1}}{\wval{s_1}+\wval{s_2}}$}  
  (s1);

\node [
  lstate={\sstate{3}}{[\pname{3}]},
  right=20mm of s2,
] (s3)  {};

\draw[arc]
  (s2)
  edge[silent]
  node[above] {(\tname{s_2},\hidden)}
  node[below] {$\pval{s_2} = \frac{\wval{s_2}}{\wval{s_1}+\wval{s_2}}$}  
  (s3);

\draw[arc,out=120,in=60]
  (s3)
  edge[silent]
  node[below] {(\tname{i_2},\hidden)}
  node[above] {$\pval{i_2} = \frac{\rval{i_2}}{\rval{i_2}+\rval{f}+\rval{c_1}}$}  
  (s2);

\node [
  lstate={\sstate{4}}{[\pname{4}]},
  right=23mm of s3,
] (s4)  {};

\draw[arc]
  (s3)
  edge
  node[above] {(\tname{f},\taskname{fin})}
  node[below,yshift=.5mm] {$\pval{f} = \frac{\rval{f}}{\rval{i_2}+\rval{f}+\rval{c_1}}$}  
  (s4);

\node[invisible,right=of s4] (xor) {};

\node [
  lstate={\sstate{5}}{[\pname{5}]},
  below= of xor,
] (s5)  {};

\draw[arc,out=-90,in=180]
  (s4)
  edge[silent]
  node[above,anchor=south,xshift=3mm] {(\tname{s_3},\hidden)}
  node[left,xshift=2mm,yshift=-2mm] {$\pval{s_3} = \frac{\wval{s_3}}{\wval{s_3}+\wval{s_4}}$}  
  (s5);

\node [
  lstate={\sstate{r}}{[\pname{r}]},
  double,
  below=20mm of s5,
] (sr)  {};

\draw[arc]
  (s5)
  edge
  node[near end,left] {(\tname{r},\taskname{rej})}
  node[near end,right] {$1$}  
  (sr);

\node [
  lstate={\sstate{6}}{[\pname{6}]},
  above= of xor,
] (s6)  {};

\draw[arc,out=90,in=180]
  (s4)
  edge[silent]
  node[below,anchor=north,xshift=3mm] {(\tname{s_4},\hidden)}
  node[left,xshift=2mm,yshift=3mm] {$\pval{s_4} = \frac{\wval{s_4}}{\wval{s_3}+\wval{s_4}}$}  
  (s6);

\node [
  lstate={\sstate{7}}{[\pname{7}]},
  right=of s6,
] (s7)  {};

\draw[arc]
  (s6)
  edge
  node[above] {(\tname{a},\taskname{acc})}
  node[below] {$1$}  
  (s7);

\node [
  lstate={\sstate{c}}{[\pname{c}]},
  double,
] at(sr-|s7) (sc)  {};

\draw[arc,out=-90,in=110]
  (s3)
  edge[silent]
  node[above] {(\tname{c_1},\hidden)}
  node[below] {$\pval{c_1} = \frac{\rval{c_1}}{\rval{i_2}+\rval{f}+\rval{c_1}}$}  
  (sc);

\draw[arc]
  (s7)
  edge[silent]
  node[anchor=west,yshift=16mm] {(\tname{c_2},\hidden)}
  node[anchor=east,yshift=16mm] {$\pval{c_2} = \frac{\rval{c_2}}{\rval{i_3}+\rval{p}+\rval{c_2}}$}  
  (sc);

\draw[arc,out=140,in=60]
  (s7)
  edge[silent]
  node[below] {(\tname{i_3},\hidden)}
  node[above] {$\pval{i_3} = \frac{\rval{i_3}}{\rval{i_3}+\rval{p}+\rval{c_2}}$}  
  (s3);

\node [
  lstate={\sstate{8}}{[\pname{8}]},
  right=25mm of s7,
] (s8)  {};

\draw[arc]
  (s7)
  edge
  node[above] {(\tname{p},\taskname{pay})}
  node[below] {$\pval{p} = \frac{\rval{p}}{\rval{i_3}+\rval{p}+\rval{c_2}}$}   
  (s8);

\node [
  lstate={\sstate{9}}{\big[\begin{array}{@{}l@{}}\pname{9},\\[-3pt]\pname{10}\end{array}\big]},
  below= of s8,
] (s9)  {};

\draw[arc]
  (s8)
  edge[silent]
  node[right] {(\tname{s_5},\hidden)}
  node[left] {$1$}  
  (s9);

\node[invisible,below=of s9] (split) {};

\node [
  lstate={\sstate{11}}{\big[\begin{array}{@{}l@{}}\pname{9},\\[-3pt]\pname{12}\end{array}\big]},
  left= 15mm of split,
] (s11)  {};

\draw[arc]
  (s9)
  edge
  node[below right,xshift=-3mm] {(\tname{l},\taskname{leave})}
  node[above,xshift=-7mm] {$\pval{l} = \frac{\rval{l}}{\rval{l}+\rval{e}}$}  
  (s11);

\node [
  lstate={\sstate{10}}{\big[\begin{array}{@{}l@{}}\pname{10},\\[-3pt]\pname{11}\end{array}\big]},
  right= 15mm of split,
] (s10)  {};

\draw[arc]
  (s9)
  edge
  node[below left,xshift=3mm] {(\tname{e},\taskname{emit})}
  node[above,xshift=7mm] {$\pval{e} = \frac{\rval{e}}{\rval{l}+\rval{e}}$}  
  (s10);

\node [
  lstate={\sstate{12}}{\big[\begin{array}{@{}l@{}}\pname{11},\\[-3pt]\pname{12}\end{array}\big]},
  below= of split,
] (s12)  {};

\draw[arc]
  (s11)
  edge
  node[above right,xshift=-3mm] {(\tname{e},\taskname{emit})}
  node[below left] {1}  
  (s12);

\draw[arc]
  (s10)
  edge
  node[above left,xshift=3mm] {(\tname{l},\taskname{leave})}
  node[below right] {1}  
  (s12);

\node [
  lstate={\sstate{h}}{[\pname{h}]},
  double,
  below= of s12,
] (sh)  {};

\draw[arc]
  (s12)
  edge[silent]
  node[right] {(\tname{s_6},\hidden)}
  node[left] {1}  
  (sh);

\end{tikzpicture}
}
\caption{Stochastic reachability graph of the order-to-cash \pwp. States are named. The initial state is shown with a small incoming edge. Final states have a double contour. To ease readability, in addition to transition names we also include the corresponding task names (or $\hidden$). Some task names are shortened, and we depict transitions that correspond to the firing of silent net transitions in magenta.}
\label{fig:sample-rg}
\end{figure}

%% file: livelock.tex
\begin{figure}[t]
\centering
\begin{subfigure}[b]{0.49\linewidth}
    \centering
    \resizebox{\textwidth}{!}{
    \begin{tikzpicture}[node distance=8mm and 6mm]
    
    \node[
      place={\pname{0}}
    ] (p0) {};
    
    \node[
      wtransition = {\tasklabel{b}}{\rval{b}},
      timed, 
      right=of p0,
    ] (b) {\tname{02}};
    
    \node[
      wtransition = {\tasklabel{a}}{\rval{a}},
      timed, 
      above=of b,
    ] (a) {\tname{01}};
    
    \draw[arc] (p0) edge (a);
    
    \node[
      place={\pname{1}},
      right=of a,
    ] (p1) {};
    
    \draw[arc] (a) edge (p1);

    \draw[arc] (p0) edge (b);
    
    \node[
      place={\pname{2}},
      right=of b,
    ] (p2) {};
    
    \draw[arc] (b) edge (p2);
    
  \node[
      wtransition = {\tasklabel{c}}{\rval{c}},
      timed, 
      right=of p2,
    ] (c) {\tname{23}};

    \draw[arc] (p2) edge (c);
    
    \node[
      place={\pname{3}},
      right=of c,
    ] (p3) {};
    
    \draw[arc] (c) edge (p3);
    
    \node[
      wtransition = {\tasklabel{d}}{\rval{d}},
      timed, 
      right=of p1,
    ] (d) {\tname{32}};

    \draw[arc] (p3) edge (d);
    \draw[arc] (d) edge (p2);

    \node[
      wtransition = {\tasklabel{e}}{\rval{e}},
      timed, 
      right=of p3,
    ] (e) {\tname{33}};
    
    \draw[darc] (p3) edge (e);

 \node[
      wtransition = {\tasklabel{f}}{\rval{f}},
      timed, 
      below=of b,
    ] (f) {\tname{04}};
    
    \draw[arc] (p0) edge (f);
    
    \node[
      place={\pname{4}},
      right=of f,
    ] (p4) {};
    
    \draw[arc] (f) edge (p4);

 \node[
      wtransition = {\tasklabel{g}}{\rval{g}},
      timed, 
      right=of p4,
    ] (g) {\tname{45}};
    
    \draw[arc] (p4) edge (g);
    
    \node[
      place={\pname{5}},
      right=of g,
    ] (p5) {};
    
    \draw[arc] (g) edge (p5);

    \end{tikzpicture}
    }
    \caption{Supporting net}
    
    \end{subfigure}%
\hfill
    \begin{subfigure}[b]{0.49\linewidth}
        \centering
        \resizebox{\textwidth}{!}{
        \begin{tikzpicture}[node distance=14mm and 6mm]


\node [
  lstate={\sstate{0}}{[\pname{0}]}
] (s0) {};
\draw
  ($(s0.north west)+(-3mm,3mm)$)
  edge[arc]
  (s0.north west);

\node [
  lstate={\sstate{2}}{[\pname{2}]},
  right=25mm of s0,
  deadend
] (s2) {};

\node [
  lstate={\sstate{1}}{[\pname{1}]},
  above=of s2,
  double
] (s1) {};

\node [
  lstate={\sstate{3}}{[\pname{3}]},
  right= 10mm of s2,
  deadend
] (s3) {};

\draw
  (s0)
  edge[arc,sloped] 
  node[below] {\tname{01}}
  node[above] {$\pval{a} = \frac{\rval{a}}{\rval{a}+\rval{b}+\rval{f}}$}
  (s1);

\draw
  (s0)
  edge[arc,sloped,deadend] 
  node[above] {\tname{02}}
  node[below] {~~~~~~$\pval{b} = \frac{\rval{b}}{\rval{a}+\rval{b}+\rval{f}}$}
  (s2);

\draw 
  (s2)
  edge[arc,deadend] 
  node[above] {\tname{23}}
  node[below] {1}
  (s3);

\draw 
  (s3)
  edge[arc,deadend,out=90,in=90] 
  node[below] {\tname{32}}
  node[above,yshift=-1mm] {$\pval{d} = \frac{\rval{d}}{\rval{d}+\rval{e}}$}
  (s2);

\draw 
  (s3)
  edge[arc,loop right,out=30,in=-30,looseness=5,deadend]
  node[left] {\tname{33}}
  node[below,yshift=-3mm] {$\pval{e} = \frac{\rval{e}}{\rval{e}+\rval{d}}$}
  (s3);

\node[
  lstate={\sstate{4}}{[\pname{4}]},
  below= of s2,
] (s4) {};

\draw 
  (s0)
  edge[arc]
  node[above,sloped] {\tname{04}}
  node[below,sloped] {$\pval{f} = \frac{\rval{f}}{\rval{a}+\rval{b}+\rval{f}}$}
  (s4);

\node[
  lstate={\sstate{5}}{[\pname{5}]},
  below= of s3,
  double
] (s5) {};

\draw[arc] 
  (s4)
  edge[arc]
  node[above] {\tname{45}}
  node[below] {1}
  (s5);

\end{tikzpicture}
}
\caption{Reachability graph}
\end{subfigure}
\caption{Reachability graph (b) of an \pwp with supporting net shown in (a), initial marking $[\place_0]$ and final markings $[\place_1]$ and $[\place_5]$. Markings $\state_2$ and $\state_3$ are livelock markings. The portion of the transition system that cannot reach a final marking is greyed out.}
\label{fig:livelock-deadend}
\end{figure}

%% file: outcome.tex
\newcommand{\outcomeprob}{\textsc{outcome-prob}\xspace}
\newcommand{\absorptionprob}{\textsc{absorption-prob}\xspace}
\newcommand{\verifprob}{\textsc{spec-prob}\xspace}
\newcommand{\traceprob}{\textsc{trace-prob}\xspace}

\newcommand{\finalstates}{F}

\newcommand{\varstate}[1]{x_{#1}}
\newcommand{\optsys}[2]{\mathcal{O}_{#1}^{#2}}
\newcommand{\eqsys}[2]{\mathcal{E}_{#1}^{#2}}
\newcommand{\equsys}[3]{\mathcal{E}_{#1}(#2,#3)}

\newcommand{\lang}[1]{\mathcal{L}(#1)}

\section{Analysis of Bounded Stochastic Labelled Processes}
\label{sec:state-probabilities}

\newcommand{\markovstates}{V}
\newcommand{\markovastates}{\markovstates^a}
\newcommand{\markovtstates}{\markovstates^t}
\newcommand{\markovstate}{v}
\newcommand{\walk}{w}
\newcommand{\markovfinalstates}{A}

\newcommand{\chain}{\mathcal{M}}
\newcommand{\chainfun}{\mathit{Chain}}

Now that we have settled the formal model of \bpwp{s}, we are ready to tackle their analysis. We recall from the introduction the five fundamental analysis tasks we tackle.
\begin{compactenum}
\item \textbf{Outcome probability}: what is the probability that the \bpwp evolves from the initial marking to one (or a subset) of its final markings?  
\item \textbf{Trace probability}: what is the probability of a given trace of the \bpwp?
\item \textbf{Specification probability}: what is the probability that the \bpwp produces a trace that satisfies a given qualitative specification that captures desired behaviour?
\item \textbf{Stochastic compliance}: is the \bpwp of interest compatible, in behavioural and stochastic terms, to a probabilistic declarative specification \cite{AMMP22} indicating which temporal constraints are expected to hold, and with which probability?
\item \textbf{Stochastic conformance checking}: how can we employ the previous analysis questions, in particular trace probability, to improve the correctness and applicability of existing stochastic conformance checking techniques \cite{LeeSA19} relating a reference stochastic process model to a recorded log?
\end{compactenum}

We address such problems as follows:
\begin{compactenum}[(a)]
\item In Section~\ref{sec:outcome}, we formalise the problem of \textbf{outcome probability} and show how it can can be solved using \emph{state-of-the-art analytic techniques for the analysis of absorbing Markov chains \cite{GriS97}}.
\item In Section~\ref{sec:spec}, we formalise the problem of \textbf{specification probability} and show how it can be reduced to that of \textbf{outcome probability}; we do so by employing \emph{state-of-the-art qualitative model checking techniques for Markov chains} \cite{BaiK08}, infused with a \emph{technique to handle silent transitions} and suitably defined over finite traces. We then argue that traces can be encoded as specifications, hence reducing \textbf{trace probability} to \textbf{specification probability}. 
\item In Section~\ref{sec:delta}, we introduce the problem of \textbf{stochastic compliance}, and show how it can be solved via  \textbf{specification probability}.
\item In Section~\ref{sec:evaluation}, we describe a \emph{proof-of-concept implementation} to solve the specification probability problem following the technique explained in Section~\ref{sec:spec}, and put the implementation into practice by showing how \textbf{stochastic conformance checking} measures can be \emph{effectively computed on a number of real-life logs}. 
\end{compactenum}

\section{Computing Outcome Probabilities}
\label{sec:outcome}

We tackle the problem of computing outcome probabilities, that is, determining the probability that a process instance of the \bpwp of interest evolves from the initial to one (or a subset) of the final markings, representing the \emph{outcome(s) of interest}. For example, we may be interested in knowing the probability that the \bpwp $\runningwproc$ of our running example (Figure~\ref{fig:sample-net}) evolves an order from \emph{opening} (marking $[\place{s}]$) to \emph{happy completion} (marking $[\place{h}]$). 

Formally, we define the problem as follows:

\begin{problem}[$\outcomeprob(\wproc,\finalstates)$]
~
\begin{compactdesc}
\item[Input:] \bpwp $\wproc$, non-empty set $\finalstates \subseteq \wproc.\fmarkings$;
\item[Output:] probability value $\prob
{\finalstates}{\wproc} = \sum_{\run \text{ run of }\wproc \text{ ending in some }\marking \in \finalstates} \prob{\run}{\wproc}$.
\end{compactdesc}
\end{problem}
The same problem can also get, as input, a stochastic transition system instead of a \bpwp.

$\outcomeprob$ cannot be solved through an enumeration of runs, as there may be infinitely many runs reaching the desired marking(s). So far, approximated enumerations have been used, by fixing a maximum threshold either on the length of runs~\cite{LABP21}, or on their minimum probability \cite{BMMP21b}. 

\subsection{Absorbing Markov Chains}
\label{sec:absorbing-markov-chains}

To obtain an analytically exact answer, we build on the direct connection, known from literature \cite{molloy1982performance,MarsanCB84}, between \gspn{s} and discrete-time Markov chains \cite{Durr12,BaiK08}. This is possible since if one does not consider labels, stochastic transition systems (in the sense of Definition~\ref{def:stochastic-ts}) are in fact discrete-time Markov chains. In the literature, the Markov chain corresponding to the reachability graph of a bounded \gspn is called the \emph{embedded/jump} chain, which can be seen as the core of the continuous-time Markov chain capturing the timed execution semantics of the net, once delays are stripped away and one is only interested in the firing probabilities \cite{molloy1982performance,MarsanCB84}. Forgetting labels has another effect: it blurs the distinction between silent and visible transitions. Hence, the same approach used in the literature for bounded \gspn{s} can be employed, in spirit, to deal with \bpwp{s}. There is a main different, though: we have to suitably mirror the finite-trace semantics of a \bpwp into its corresponding Markov chains. We do so by establishing a connection between \bpwp{s} and the class of \emph{absorbing Markov chains} \cite[Chapter 11]{GriS97}. 

Instead of using the classical definition of Markov chains via probability transition matrixes, we define them via their graph representation, which provides a direct correspondence with stochastic transition systems as per Definition~\ref{def:stochastic-ts}.

\begin{definition}[Markov chain, walk, walk probability]
\label{def:markov-chain}
An \emph{absorbing Markov chain} is a tuple $\chain = \tup{\markovstates,\transrel,\probfun}$ where:
\begin{compactitem}[$\bullet$]
\item $\markovstates$ is a finite set of \emph{states};
\item $\transrel \subseteq \markovstates \times \markovstates$ is a \emph{transition relation}, such that $\tup{\markovstates,\transrel}$ is a \emph{connected graph};
\item $\probfun: \transrel \rightarrow [0,1]$ is a \emph{transition probability function} mapping each transition in $\transrel$ to a corresponding probability value in $[0,1]$, such that for every state $\markovstate \in \markovstates$, we have $\sum_{\xi=\tup{\markovstate_1,\markovstate_2} \in \transrel \text{~s.t.~}\markovstate_1=\markovstate} \probfun(\xi) = 1$.
\end{compactitem}
A \emph{walk} over $\chain$ from state $\markovstate_0 \in \markovstates$ to state $\markovstate_n \in \markovstates$ is a sequence $\walk = \markovstate_0 \cdots \markovstate_n$ of states, such that for every $i \in \set{0,n-1}$, we have $\tup{\markovstate_i,\markovstate_{i+1}} \in \transrel$. The probability of $\walk$ is $\prob{\walk}{\chain} = \prod_{i=0}^{n-1} \probfun(\tup{\markovstate_i,\markovstate_{i+1}})$.
\end{definition}

\begin{definition}[Absorbing Markov chain, proper walk]
\label{def:absorbing-markov-chain}
An \emph{absorbing Markov chain} is a Markov chain $\chain = \tup{\markovstates,\transrel,\probfun}$, where the set $\markovstates$ of states is partitioned into two subsets $\markovastates$ and $\markovtstates$:
\begin{compactitem}[$\bullet$]
\item $\markovastates$ is the set of \emph{absorbing states}, so that $\markovstate \in \markovastates$ if and only if the only transition having $\markovstate$ as source is the self-loop $\tup{\markovstate,\markovstate}$ (thus implying $\probfun(\tup{\markovstate,\markovstate}) = 1$);
\item $\markovtstates = \markovstates \setminus \markovastates$ is the set of \emph{transient states}, so that $\markovstate \in \markovtstates$ if and only if there is a walk from $\markovstate$ to some absorbing state in $\markovastates$. 
\end{compactitem} 
A \emph{proper walk} over $\chain$ is a walk from a transient state $\markovstate \in \markovtstates$ to an absorbing state $\markovstate_f \in \markovastates$ , which does not recur over $\markovstate_f$ (i.e., the previous-last state of the walk is different from $\markovstate_f$).
\end{definition}

The notion of proper walk helps establishing a direct correspondence with runs of an \pwp. Notice that having longer walks recurring over the same absorbing state would not change at all their overall probability. 

Within the wide spectrum of analysis tasks for Markov chains, the main problem we are interested in is that of calculating the \emph{absorption probability}. Given a transient state $\markovstate$ and a set $\markovfinalstates$ of one or more absorption states, the absorption probability of $\markovstate$ is the probability that a walk from $\markovstate$ ends up in one of the states $\markovfinalstates$. This corresponds to the overall probability of proper walks from $\markovstate$ to one of the states in $\markovfinalstates$.

\begin{problem}[$\absorptionprob(\chain,\markovstate,\markovfinalstates)$]
~
\begin{compactdesc}
\item[Input:] Absorbing Markov chain $\chain$, transient state $\markovstate \in \chain.\markovtstates$, non-empty set $\markovfinalstates \subseteq \chain.\markovastates$ of absorbing states;
\item[Output:] probability value $\cprob{\markovfinalstates}{\markovstate}{\wproc} = \sum_{\walk \text{ proper walk from }\markovstate \text{ to some }\markovstate_f \in \markovfinalstates} \prob{\walk}{\chain}$.
\end{compactdesc}
\end{problem}

This problem can be analytically tackled by solving a system of linear equations, which encode the step-wise semantics of the input absorbing Markov chain, and how probabilities are propagated therein. In particular, every state $\markovstate_i \in \chain.\markovstates$ gets a corresponding state variable $\varstate{\markovstate_i}$, capturing the probability the probability $\cprob{\markovfinalstates}{\markovstate_i}{\chain}$ of reaching one of the absorbing states in $\markovfinalstates$ from state $\markovstate_i $; this means that variable $\varstate{\markovstate}$ represents the solution of the problem. Then, each equation defines the value of a state variable $\varstate{\markovstate_i}$ as follows:
\begin{compactdesc}
\item[Base case (targeted absorbing state)] if $\markovstate_i \in \markovfinalstates$, that is, $\markovstate_i$ is one of the targeted absorbing states), then  $\varstate{\markovstate_i} = 1$;
\item[Base case (non-targeted absorbing state)] if $\markovstate_i \in \markovastates \setminus \markovfinalstates$, that is, $\markovstate_i$ is a non-targeted absorbing state, then $\varstate{\markovstate_i} = 0$;  
\item[Inductive case] if $\markovstate_i$ is a transient state, its variable is equal to sum of the state variables of its successor states, each weighted by the correspnding transition probability to move to that successor.
\end{compactdesc} 

Formally, $\absorptionprob(\chain,\markovstate,\markovfinalstates)$   leads to the following \emph{system $\equsys{\chain}{\markovstate}{\markovfinalstates}$ of linear equations}, defined as follows:
\begin{align}
 \varstate{\markovstate_i} 
    & = 1 
    && \text{for each } \markovstate_i \in \markovfinalstates 
  \label{eq:lin:markov-final}
\\
  \varstate{\markovstate_j} 
    & = 0 
    && \text{for each } \markovstate_j \in \chain.\markovastates \setminus \markovfinalstates   \label{eq:lin:markov-rej}
\\
  \varstate{\markovstate_k} 
    & = \sum_{\tup{\markovstate_k,\markovstate_k'} \in \chain.\transrel} \chain.\probfun(\tup{\markovstate_k,\markovstate_k'}) \cdot \varstate{\markovstate_k'}
    && \text{for each } \markovstate_k \in \chain.\markovtstates
  \label{eq:lin:markov-flow}
\end{align}
We will introduce examples of systems of this kind in the next section. To conclude this introductory discussion on absorbing Markov chains, we recall the following central result.

\begin{theorem}[From \cite{GriS97}]
\label{thm:absorption}
Given an absorbing Markov chain $\chain$, a transient state $\markovstate \in \chain.\markovtstates$, and a non-empty subset $\markovfinalstates$ of absorbing states from $\chain.\markovastates$, we have:
\begin{compactenum}
\item $\equsys{\chain}{\markovstate}{\markovfinalstates}$ has exactly one solution;
\item $\absorptionprob(\chain,\markovstate,\markovfinalstates)$ coincides with the value of state variable $\varstate{\markovstate}$ extracted from the unique solution of $\equsys{\chain}{\markovstate}{\markovfinalstates}$. 
\end{compactenum} 

\end{theorem}

\subsection{Connecting \pwp{s} to (Absorbing) Markov Chains}
\label{sec:lsp-to-markov-chains}

We now encode stochastic transition systems and \pwp{s} into corresponding Markov chains. The encoding is straightforward: transition labels are stripped off, and final states get a self-loop with probability 1, thus becoming absorbing states in the chain.

\begin{definition}
\label{def:chain-encoding}
The embedded Markov chain of a stochastic transition system $\tsys{} = \tup{\labels,\states,\istate,\fstates,\transrel,\probfun}$ is the Markov chain $\chainfun(\tsys{}) = \tup{\markovstates,\transrel',\probfun'}$, where:
\begin{compactitem}[$\bullet$]
\item $\markovstates = \states$;
\item $\transrel' = \transrel \cup \set{\tup{\state,\state} \mid \state \in \fstates}$;
\item $\probfun' = \set{\tup{\state_1,\state_2} \mapsto \probfun(\tup{\state_1,\ell,\state_2}) \mid \tup{\state_1,\ell,\state_2} \in \transrel} \cup \set{\tup{\state,\state} \mapsto 1 \mid \state \in \fstates}$.
\end{compactitem}
The embedded Markov chain of a \bpwp $\wproc$ is the Markov chain $\chainfun(\wproc) = \chainfun(\rg{\wproc})$.
\end{definition}

By inspecting the notion of reachability graph of a \bpwp, as well as the notion of livelock marking, we immediately get the following important result.

\begin{lemma}
\label{lem:lsp-chain}
For a livelock-free \bpwp $\wproc$, we have that $\chainfun(\wproc)$ is an absorbing Markov chain.
\end{lemma}

\begin{proof}
From Definition~\ref{def:chain-encoding}, we have that the final markings of $\rg{\wproc}$, which are deadlock markings and hence do not have any successor, get encoded by $\chainfun(\tsys{})$ into absorbing states. By Definition~\ref{def:livelock}, livelock-freedom implies that every other marking in $\rg{\wproc}$ reaches some of such final marking, thus getting encoded by $\chainfun(\tsys{})$ into a transient state.
\end{proof}

\subsection{Outcome Probability and Stochastic Languages for Livelock-free \bpwp{s}}
\label{sec:outcome-probability-livelock-free}

Lemma~\ref{lem:lsp-chain} provides the basis for showing that the outcome probability problem can be casted as an absorption probability problem, and hence analytically solved through the technique mentioned in Theorem~\ref{thm:absorption}.

\begin{theorem}
\label{thm:outcome-probability}
For every livelock-free \bpwp $\wproc$ and set $\finalstates \subseteq \wproc.\fmarkings$, we have that
$\outcomeprob(\wproc,\finalstates)$ has a unique solution, which coincides with that of $\absorptionprob(\chainfun(\wproc),\wproc.\istate,\finalstates)$.
\end{theorem}

\begin{proof}
Lemma~\ref{lem:lsp-chain} guarantees that $\chainfun(\wproc)$ is an absorbing Markov chain. This implies, by Theorem~\ref{thm:absorption}, that $\absorptionprob(\chainfun(\wproc),\wproc.\istate,\finalstates)$ has a unique solution. By contrasting the definition of run over a \pwp as per Definition~\ref{def:run}, with the definition of proper run as per Definition~\ref{def:absorbing-markov-chain}, we have that every run over $\wproc$ corresponds one-to-one to a proper walk over $\chainfun(\wproc)$ from $\istate$. Hence:
\begin{align*}
 \outcomeprob(\wproc,\finalstates) &= 
\sum_{\run \text{ run of }\wproc \text{ ending in some }\marking \in \finalstates} \prob{\run}{\wproc}\\
&= \sum_{\walk \text{ proper walk of }\chainfun(\wproc)  \text{ from } \istate \text{ to some }\marking \in \finalstates} \prob{\walk}{\chainfun(\wproc)}\\
&= \absorptionprob(\chainfun(\wproc),\wproc.\istate,\finalstates)
\end{align*} 
\end{proof}

Considering Theorem~\ref{thm:outcome-probability}, the encoding of livelock-free \bpwp{s} into absorbing Markov chains, and the definition of the corresponding systems of equations, we can directly solve the outcome probability problem $\outcomeprob(\wproc,\finalstates)$ for a livelock-free \bpwp and a nonempty set $\finalstates \subseteq \wproc.\fmarkings$ by extracting the value of variable $\varstate{\wproc.\imarking}$ from the unique solution of the following system $\eqsys{\wproc}{\fmarkings}$ of equations:
\begin{align}
 \varstate{\marking_i} 
    & = 1 
    && \text{for each marking } \marking_i \in \finalstates 
  \label{eq:lin:final}
\\
  \varstate{\marking_j} 
    & = 0 
    && \text{for each marking } \marking_j \in \wproc.\fmarkings \setminus \finalstates   \label{eq:lin:rej}
\\
  \varstate{\marking_k} 
    & = \hspace*{-10mm}\sum_{\tup{\marking_k,l,\marking_k'} \in \outgoing{\marking_k}{\rg{\wproc}}} \hspace*{-14mm}\probfun(\tup{\marking_k,l,\marking_k'}) \cdot \varstate{\marking_k'}
    && \text{for each marking } \marking_k \in \rg{\wproc}.\states \setminus \fmarkings
  \label{eq:lin:flow}
\end{align}

\begin{example}
\label{ex:outcome}
Consider the \bpwp $\myproc$ (Figure~\ref{fig:sample-net}). We want to solve $\outcomeprob(\myproc,\marking_h)$ to compute the probability that a created order is eventually paid and shipped. To do so, we encode the reachability graph of Figure~\ref{fig:sample-rg} into the system $\eqsys{\myproc}{\marking_h}$. 
Specifically, noticing that the overall sum of all probability values appearing in the outgoing edges of a given marking is always 1, we get:
\begin{align*}
\varstate{\marking_h} &= 1\\
\varstate{\marking_c} &= \varstate{\marking_r} = \varstate{\marking_5} = 0\\
\varstate{\marking_{10}} &= \varstate{\marking_{11}} = 
\varstate{\marking_{12}} = \varstate{\marking_h} = 1\\
\varstate{\marking_{8}} &= \varstate{\marking_{9}} = \pval{e}\varstate{\marking_{10}} + \pval{l}\varstate{\marking_{11}} = \pval{e} + \pval{l} = 1\\
\varstate{\marking_{6}} &=   
\varstate{\marking_{7}} =
\pval{i_3}\varstate{\marking_{3}} +
\pval{p}\varstate{\marking_{8}} +
\pval{c_2}\varstate{\marking_{c}} = 
\pval{i_3}\varstate{\marking_{3}} +
\pval{p}\\
\varstate{\marking_{4}} &=   
\pval{s_3}\varstate{\marking_{5}}+
\pval{s_4}\varstate{\marking_{6}} = 
\pval{s_4}\pval{i_3}\varstate{\marking_{3}} +
\pval{s_4}\pval{p}\\
\varstate{\marking_{3}} &=   
\pval{i_2}\varstate{\marking_{2}} +
\pval{f}\varstate{\marking_4}+
\pval{c_1}\varstate{\marking_c} =
\pval{i_2}\varstate{\marking_{2}} +
\pval{f}\pval{s_4}\pval{i_3}\varstate{\marking_{3}} 
+\pval{f}\pval{s_4}\pval{p} = \\
&= \frac{\pval{i_2}}{1-\pval{f}\pval{s_4}\pval{i_3}}\varstate{\marking_2} + \frac{\pval{f}\pval{s_4}\pval{p}}{1-\pval{f}\pval{s_4}\pval{i_3}}\\
\varstate{\marking_{2}} &=   
\pval{s_1}\varstate{\marking_1} +
\pval{s_2}\varstate{\marking_3} =
\pval{s_1}\varstate{\marking_1} +
\frac{\pval{s_2}\pval{i_2}}{1-\pval{f}\pval{s_4}\pval{i_3}}\varstate{\marking_2} + \frac{\pval{s_2}\pval{f}\pval{s_4}\pval{p}}{1-\pval{f}\pval{s_4}\pval{i_3}} =
\\
&= 
\frac{\pval{s_1}-\pval{s_1}\pval{f}\pval{s_4}\pval{i_3}}{1-\pval{f}\pval{s_4}\pval{i_3}-\pval{s_2}\pval{i_2}}\varstate{\marking_1}
+ \frac{\pval{s_2}\pval{f}\pval{s_4}\pval{p}}{1-\pval{f}\pval{s_4}\pval{i_3}-\pval{s_2}\pval{i_2}}
\\
\varstate{\marking_{0}} &= \varstate{\marking_{1}} = \varstate{\marking_{2}} = 
\frac{\pval{s_1}-\pval{s_1}\pval{f}\pval{s_4}\pval{i_3}}{1-\pval{f}\pval{s_4}\pval{i_3}-\pval{s_2}\pval{i_2}}\varstate{\marking_0}
+ \frac{\pval{s_2}\pval{f}\pval{s_4}\pval{p}}{1-\pval{f}\pval{s_4}\pval{i_3}-\pval{s_2}\pval{i_2}}\\
&= \frac{\pval{s_2}\pval{f}\pval{s_4}\pval{p}}{1-\pval{s_1}-\pval{s_2}\pval{i_2}-\pval{f}\pval{s_4}\pval{i_3}+\pval{s_1}\pval{f}\pval{s_4}\pval{i_3}}
\\
&= \frac{\pval{s_2}\pval{f}\pval{s_4}\pval{p}}{1-\pval{s_1}-\pval{s_2}\pval{i_2}-\pval{f}\pval{s_4}\pval{i_3}(1-\pval{s_1})}
= \frac{\pval{s_2}\pval{f}\pval{s_4}\pval{p}}{1-\pval{s_1}-\pval{s_2}\pval{i_2}-\pval{s_2}\pval{f}\pval{s_4}\pval{i_3}}
\end{align*}
If we assume that the weights of $\myproc$ are all equal, the probability distributions for choosing the next transition are all uniform, so that $\pval{s_1} = \pval{s_2} = \frac{1}{2}$,
$\pval{i_2} = \pval{f} = \pval{c_1} = \frac{1}{3}$, $\pval{s_3} = \pval{s_4} = \frac{1}{2}$, $\pval{i_3} = \pval{p} = \pval{c_2} = \frac{1}{3}$, and $\pval{e} = \pval{l} = \frac{1}{2}$. Under this assumption, we thus get
 $\varstate{\imarking} = \frac{1/36}{1-1/2-1/6-1/36} = \frac{1/36}{11/36} = \frac{1}{11} \approx 0.09$.

With an analogous approach, the probability that an order gets cancelled is $\frac{4}{11}$, and the one that an order gets rejected is $\frac{6}{11}$. 
\end{example}

As the example shows, the total sum of probabilities associated to reaching the whole set of final markings is, as expected, 1: every order that finishes does so in one of the categories of completion (i.e., one of the final markings). This actually generalizes to any livelock-free \bpwp: the outcome probability calculated by considering the whole set of final markings corresponds to 1.

\begin{theorem}
\label{thm:total-absorption-one}
For every livelock-free \bpwp $\wproc$, we have that $\outcomeprob(\wproc,\wproc.\fmarkings) = 1$.
\end{theorem}

\begin{proof}
The claim is a direct consequence of Theorem~\ref{thm:outcome-probability}, and the fact that the probability of reaching an absorbing state in an absorbing Markov chain is 1 - as proved in \cite[Theorem 11.3]{GriS97}. 
\end{proof}

The effect of Theorem~\ref{thm:total-absorption-one} is that the trace language of a livelock-free \bpwp is indeed a stochastic language.

\begin{definition}[Stochastic language]
\label{def:stochastic-language}
A stochastic language over an alphabet $L$ of symbols is a total function $\Psi: L^* \rightarrow [0,1]$ mapping each sequence over $L$ onto a probability value, so that $\sum_{\xi \in L^*} \Psi(\xi) = 1$.
\end{definition}

\begin{definition}[Induced stochastic language]
\label{def:induced-stochastic-language}
The \emph{induced stochastic language} of a livelock-free \bpwp $\wproc$ is the function $\mathcal{P}_\wproc: \tasks^* \rightarrow [0,1]$ defined as follows:
for every $\sigma \in \tasks^*$, we have
$$
\mathcal{P}_{\wproc}(\sigma) = 
\begin{cases}
\prob{\sigma}{\wproc} & \text{ if } \sigma \text{ is a model trace of }\wproc\\
0 & \text{ otherwise}\\
\end{cases}
$$
\end{definition}

\begin{corollary}
\label{thm:lsp-stochastic-language}
Let $\wproc$ be a livelock-free \bpwp. Then $\mathcal{P}_\wproc$ is indeed a stochastic language in the sense of Definition~\ref{def:stochastic-language}.
\end{corollary}

\subsection{Dealing with Livelocks}
We now show how the technique introduced in Section~\ref{sec:outcome-probability-livelock-free} for livelock-free \bpwp{s} can be generalized to arbitrary \bpwp{s} with minor adaptations: the crux is to deal with livelocks by recognizing that they are not able to reach any final marking, and hence their states should be assigned to a probability value of 0. 

This problem is well-known in Markov chain analysis. In fact, the same approach used to solve the absorption probability problem can be adopted, in spirit, to deal with arbitrary Markov chains, in particular to compute the probability that a walk from a state $\markovstate_i$ culminates in another state $\markovstate_f$. However, as the next example shows, in case of livelocks the system of linear equations recalled before becomes indeterminate.

\begin{example}
\label{ex:livelock-probabilities}
Consider the \bpwp $\wproc_{\text{live}}$ from Example~\ref{ex:livelock-deadend} and Figure~\ref{fig:livelock-deadend}. If we try to compute the outcome probability of its two final markings $[\pname{1}]$ and $[\pname{5}]$ by applying verbatim the approach used so far, we get the following system of linear equations, directly derived from the reachability graph of Figure~\ref{fig:livelock-deadend}(b):
\begin{align*}
\varstate{\marking_5} &= 1 
& 
\varstate{\marking_2} &= \varstate{\marking_3}
\\
\varstate{\marking_4} &= \varstate{\marking_5} 
& 
\varstate{\marking_1} &= 1
\\
\varstate{\marking_3} &= \pval{d} \varstate{\marking_2} + \pval{e} \varstate{\marking_3} 
& 
\varstate{\marking_0} &= \pval{a} \varstate{\marking_1} + \pval{b} \varstate{\marking_2} + \pval{f} \varstate{\marking_4}
\end{align*}
By recalling that, by definition, $\pval{d} + \pval{e} = 1$, the system yields $\varstate{\marking_3} = \pval{d} \varstate{\marking_3} + \pval{e} \varstate{\marking_3} = (\pval{d}+\pval{e})\varstate{\marking_3} = \varstate{\marking_3}$, and so is indeterminate. 
This is caused by the livelock markings $[\pname{2}]$ and $[\pname{3}]$, which induce a cycle from which none of the target markings can be reached. The minimal non-negative solution of the system then assigns $\varstate{\marking_3} = \varstate{\marking_2} = 0$ and, in turn,  $\varstate{\marking_0} = \pval{a}$. 
\end{example}

The indeterminacy shown in Example~\ref{ex:livelock-probabilities} is solved in Markov chain analysis by turning the system of equalities into an optimization problem, where one is interested in extracting the minimum non-negative solution \cite{Durr12,BaiK08}.  This corresponds to assigning 0 to each state variable that is causing the indeterminacy. In the case of \bpwp{s}, such state variables correspond to livelock markings. To deal with arbitrary \bpwp{s} without resorting to an optimisation problem, we then proceed as follows. The outcome probability problem $\outcomeprob(\wproc,\finalstates)$ for an arbitrary \bpwp and a nonempty set $\finalstates \subseteq \wproc.\fmarkings$ is solved by extracting the value of variable $\varstate{\wproc.\imarking}$ from the unique solution of $\eqsys{\wproc}{\fmarkings}$, now generalized as follows:
\begin{align}
 \varstate{\marking_i} 
    & = 1 
    && \text{for each marking } \marking_i \in \finalstates 
  \label{eq:lin:final}
\\
  \varstate{\marking_j} 
    & = 0 
    && \text{for each marking } \marking_j \in \wproc.\fmarkings \setminus \finalstates   \label{eq:lin:rej}
\\
  \varstate{\marking_l} 
    & = 0 
    && \text{for each livelock marking } \marking_l \in \rg{\wproc}.\states   \label{eq:lin:rej}
\\
  \varstate{\marking_k} 
    & = \hspace*{-12mm}\sum_{\tup{\marking_k,l,\marking_k'} \in \outgoing{\marking_k}{\rg{\wproc}}} \hspace*{-15mm}\probfun(\tup{\marking_k,l,\marking_k'}) \cdot \varstate{\marking_k'}
    && \text{for each other  marking } \marking_k \in \rg{\wproc}.\states
  \label{eq:lin:flow}
\end{align}
Recall that detecting livelock markings can be done over the reachability graph of the input \bpwp via simple graph (non-)reachability checks. Such checks do not involve probabilities at all.

\begin{example}
\label{ex:livelock-probabilities-fixed}
We revise Example~\ref{ex:livelock-probabilities} by solving $\outcomeprob(\wproc_{\text{live}},\set{[\pname{1}],[\pname{5}]})$ as follows. The system $\eqsys{\wproc_{\text{live}}}{\set{[\pname{1}],[\pname{5}]}}$ corresponds to:
\begin{align*}
\varstate{\marking_5} &= \varstate{\marking_1} = 1 
& 
\varstate{\marking_2} &= \varstate{\marking_3} = 0
\\
\varstate{\marking_4} &= \varstate{\marking_5} 
& 
\varstate{\marking_0} &= \pval{a} \varstate{\marking_1} + \pval{b} \varstate{\marking_2} + \pval{f} \varstate{\marking_4}
\end{align*}
Hence, we get $\varstate{\marking_0} = \pval{a} + \pval{f}$.
\end{example}

As witnessed by Example~\ref{ex:livelock-probabilities-fixed}, in case of \bpwp{s} with livelocks, the outcome probability value $p$ computed over all final markings is not 1, as instead it is guaranteed for livelock-free \bpwp{s} (cf.~Theorem~\ref{thm:total-absorption-one}). This also means that every \bpwp with livelocks does not induce a stochastic language over its traces, but only a form of ``incomplete'' stochastic language, where the complement probability value $1-p$ indicates the probability that the process gets stuck in a livelock.

\begin{remark}
\label{rem:livelock-stochastic-language}
For a \bpwp $\wproc$ containing livelocks, we get the following:
\begin{compactitem}
\item $\outcomeprob(\wproc,\wproc.\fmarkings) < 1$;
\item $1-\outcomeprob(\wproc,\wproc.\fmarkings)$ is the probability that an execution starting from $\wproc.\imarking$ ends up into a livelock.
\end{compactitem}
\end{remark} 

\begin{example}
By continuing on Example~\ref{ex:livelock-probabilities-fixed}, we have that the probability that an execution of $\wproc_{\text{live}}$ starting from $\wproc_{\text{live}}.\imarking = [\pname{0}]$ ends up into a livelock is $1-\outcomeprob(\wproc_{\text{live}},\set{[\pname{1}],[\pname{5}]}) = 1 - \pval{a}-\pval{f} = \pval{b}$, which is in fact the probability of firing $\tname{02}$ from the initial marking. 
\end{example}

%% file: verification.tex
\section{Computing Specification and Trace Probabilities}
\label{sec:spec}

We now further leverage the connection between \bpwp{s} and discrete-time Markov chains (cf.~Definition~\ref{def:chain-encoding}), to deal with the computation of \textbf{specification probabilities}. A specification is a compact description of (possibly infinitely  many) traces of interest, and the problem is about measuring what is the probability that a trace produced by the \bpwp of interest indeed belongs to the specification. This problem is also known under the name of qualitative model checking \cite[Ch.~10]{BaiK08}, to stress that the specification is \emph{qualitative}, i.e., non-probabilistic.

To solve the problem, we borrow the main idea from \cite[Ch.~10]{BaiK08} that underlies the construction of the cross-product stochastic transition system that describes all and only the behaviour that belong to both the specification and the \bpwp under scrutiny. However, we have to adapt this construction to reflect two distinctive features of our setting:
\begin{compactitem}[$\bullet$]
\item traces have a finite length, and this has to be taken into account in the choice of formalism for describing the specification, as well as when constructing the cross-product;
\item the behaviour described by the specification are traces, while those of the \bpwp are runs over transitions that also include silent steps, a mismatch that has to be resolved before  constructing the cross product.
\end{compactitem}
We proceed by fixing and justifying the formalism we use to represent specifications. We then formalise the problem of specification probability, and describe how to solve it in three steps: 
\begin{inparaenum}[\it (i)]
\item infusing silent steps in the specification;
\item computing the cross-product of the altered specification and the reachability graph of the \bpwp; and
\item solving an outcome probability problem on the cross product.
\end{inparaenum}

Finally, we show that the solution of the specification probability problem yields, as a by-product, the solution of the \textbf{trace probability} problem (as described in Definition~\ref{def:trace-probability}).

\subsection{Automata-Based Specifications}
\label{sec:automata-specifications}
As already pointed out, a specification intentionally describes a (possibly infinite) set of desired finite-length traces. Such traces are defined over the task names in $\tasks$ (without $\hidden$). 
As a natural formalism to capture such specifications, we employ (deterministic) finite-state automata. To avoid confusion between the sequences accepted by automata and runs/traces of \pwp{s}, we use the standard term of \emph{words} for automata. 

\newcommand{\DFA}{DFA\xspace}
\newcommand{\NFA}{NFA\xspace}
\newcommand{\aut}{A}
\newcommand{\autstates}{\states}
\newcommand{\autstate}{\state}
\newcommand{\autistate}{\istate}
\newcommand{\autfstates}{\fstates}
\newcommand{\autfun}{\delta}
\newcommand{\silentaut}{\bar{\aut}}
\newcommand{\traceaut}[1]{\aut_{#1}}
\newcommand{\silenttraceaut}[1]{\bar{\aut}_{#1}}

\newcommand{\product}[2]{\Upsilon_{#1}^{#2}}

\begin{definition}[\DFA]
\label{def:automaton}
A \emph{deterministic finite-state automaton} (\DFA) over $L$ is a tuple $\aut = \tup{L,\autstates,\autistate,\autfstates,\autfun}$, where:
\begin{compactitem}[$\bullet$]
\item $L$ is a finite \emph{alphabet} of symbols;
\item $\autstates$ is a finite set of \emph{states}, with $\autistate \in \autstates$ the \emph{initial state} and $\autfstates \subseteq \autstates$ the set of \emph{final states};
\item $\autfun: \autstates \times L \rightarrow \autstates$ is a transition transition function that, given a state $\autstate \in \autstates$ and a label $l \in L$, returns the successor state $\autfun(\autstate,l)$.
\end{compactitem}
\end{definition}

\begin{definition}[Word, acceptance, language]
\label{def:automaton-language}
Given a finite alphabet $L$, a word over $L$ is a finite sequence of symbols from $L^*$. A \DFA $\aut$ over $L$ accepts a word $\trace = l_0,\ldots,l_n$ if there exists a sequence of states $\autstate_0,\ldots,\autstate_{n+1}$ such that:
\begin{compactitem}[$\bullet$]
\item $s_0 = \aut.\istate$;
\item $\autstate_{n+1} \in \autfstates$;
\item for every $i \in \set{0,\ldots,n}$, we have $\autstate_{i+1} = \autfun(\autstate_i,l_i)$.
\end{compactitem}
The \emph{language} $\lang{\aut}$ of $\aut$ is the set of all words accepted by $\aut$.
\end{definition}

\begin{remark}
\label{rem:dfa-properties}
Dealing with \DFA{s} in turn provides an operational solution for dealing with a wide range of formalisms to specify traces:
\begin{compactitem}[$\bullet$]
\item It expresses all non-deterministic automata (\NFA{s}), as each \NFA can be encoded into a corresponding \DFA \cite{HoMU07}.
\item For the same reason \cite{HoMU07}, it expresses all regular expressions.
\item It expresses all declarative specifications formalised in Linear Temporal Logic over finite traces (\LTLf) or the richer Linear Dynamic Logic over finite traces (\LDLf), as such logics have an automata-theoretic characterisation based on \NFA{s}/\DFA{s} \cite{DeVa13,DDMM22}.
\item In turn, it expresses all specifications captured in \declare \cite{PeSV07,MPVC10}, possibly extended with meta-constraints \cite{DDMM22}. 
\end{compactitem}

\end{remark}

\input{sample-automata}

\begin{example}
\label{ex:properties}
Figure~\ref{fig:sample-automata} shows three specifications of interest for our running example. Specifically:
\begin{compactitem}[$\bullet$]
\item The DFA in Figure~\ref{fig:sample-automata}(a) captures those words where every time an order is opened, it is later shipped (possibly executing other tasks in between).
\item The DFA in Figure~\ref{fig:sample-automata}(b) captures those words where the other is eventually rejected and later paid.
\item The DFA in Figure~\ref{fig:sample-automata}(c) captures those words where if an order is accepted, then it has been previously finalised.
\end{compactitem}
\end{example}

\subsection{Defining and Solving the Specification Probability Problem}
\label{sec:specification-probability}

%
Noticing that words over $\tasks$ correspond to traces of \pwp{s}, we are now ready to define the \textbf{specification probability} problem.

\begin{problem}[$\verifprob(\wproc,\aut)$] 
~
\begin{compactdesc}
\item[Input:] \bpwp $\wproc$, \DFA $\aut$ over $\tasks$;
\item[Output:] probability $\prob{\aut}{\wproc}= \sum_{\trace \text{ model trace of } \wproc { s.t.~}\trace \in \lang{\aut}} \prob{\trace}{\wproc}$. 
\end{compactdesc}
\end{problem}

To solve the problem, we need to account for three different aspects: 
\begin{compactenum}
\item deal with the mismatch between runs over $\wproc$ and words accepted by $\aut$;
\item single out all and only those model traces of $\wproc$ that are also words accepted by $\aut$;
\item compute the collective probability of all such traces.
\end{compactenum}

We tackle these three aspects with three corresponding steps.

\medskip
\noindent
\textbf{Silenced \DFA.} Given a specification DFA describing desired words over $\tasks$, we transform it into a corresponding \DFA describing corresponding sequences of labels over $\alphabet$. The idea is to make the resulting \DFA continues to describe the same desired traces, but accepts padding their visible steps with an arbitrary number of silent steps in between.

\begin{definition}[Silenced \DFA]
\label{def:silenced-dfa}
Given a \DFA $\aut = \tup{\tasks,\autstates,\autistate,\autfstates,\autfun}$ over $\tasks$, its \emph{silenced \DFA} $\silentaut$ is the \DFA $\tup{\alphabet,\autstates',\autistate',\autfstates',\autfun'}$ over $\alphabet$ satisfying the following conditions:
\begin{compactitem}
\item states (including initial and final ones) are identical to $\aut$, that is, $\autstates' = \autstates$, $\autistate' = \autistate$, and $\finalstates' = \finalstates$;
\item $\autfun'$ expands $\autfun$ by introducing a $\hidden$-labelled self-loop in each state, that is, $\autfun' = \autfun \cup \set{\tup{\autstate,\hidden} \rightarrow \autstate \mid \autstate \in \autstates}$
\end{compactitem}
\end{definition}

\input{sample-silenced-automata}

Figure~\ref{fig:sample-silenced-automata} shows the silenced \DFA{s} corresponding to the \DFA{s} in Figure~\ref{fig:sample-automata}.

We now establish a key correspondence between \DFA{s} and their silenced versions, when characterising traces of an \pwp.

\begin{lemma}
\label{lem:dfa-correspondence}
Let $\wproc$ be an \pwp, and $\aut$ be a \DFA over $\tasks$. For every model trace $\trace$ of $\wproc$ and for every run $\run$ of $\wproc$ inducing $\trace$, we have that $\trace$ is accepted by $\aut$ if and only if the sequence of labels $\wproc.\net.\labeling(\run)$ is accepted by $\silentaut$.
\end{lemma}


\begin{proof}
The claim directly follows from Definitions~\ref{def:model-trace} and~\ref{def:silenced-dfa}, noticing that the difference between a trace $\sigma$ induced by a run $\run$ and the sequence of labels $l = \wproc.\net.\labeling(\run)$ is that $l$ may insert arbitrarily many $\hidden$ symbols at the beginning and and of $\sigma$, as well as between every consecutive tasks appearing therein.
\end{proof}

\medskip
\noindent
\textbf{Product stochastic transition system.} Given a \bpwp $\wproc$ and a \DFA $B$ over $\alphabet$, we use $\rg{\wproc}$ and $B$ to construct a \emph{product} stochastic transition system generates all and only those runs of $\wproc$ whose label sequences are accepted by $B$, retaining the probabilities of those runs from $\wproc$. This can be done by the standard product automaton construction, with the only difference that we need to retain the stochastic information coming from $\wproc$. We then use this as a basis to solve the specification probability problem, relying on silenced-\DFA{s}.

\begin{definition}[Product system]
\label{def:product}
Let $\wproc$ be a \bpwp with $\rg{\wproc} = \tup{\states^1,\istate^1,\fstates^1,\transrel^1,\probfun^1}$, 
and $B = \tup{\alphabet^2,\autstates^2,\autistate^2,\autfstates^2,\autfun^2}$ a \DFA over $\alphabet$. 
The \emph{product system} $\product{\wproc}{B}$ of $\wproc$ and $B$ is a stochastic transition system $\tup{\states,\istate,\fstates,\transrel,\probfun}$ whose states  are pairs of states from $\states^1 \times \states^2$, and whose components are defined by mutual induction as the minimal sets satisfying the following conditions:
\begin{compactenum}
\item $\istate = \tup{\istate^1,\istate^2}$, $\istate \in \states$;
\item for every state $\tup{\state^1,\state^2} \in \states$ and every label $l \in \alphabet^2$ such that
\begin{compactenum}[\itshape (i)]
\item $\tup{\state^1,\transition^1,{\state^1}'} \in \transrel^1$ for some ${\state^1}' \in \states^1$ and transition $\transition^1 \in \wproc.\net.\transitions$ with $\wproc.\net.\labeling(\transition^1) = l $, and 
\item $\autfun^2(\state^2,l) = {\state^2}'$ for some ${\state^2}' \in \states^2$,
\end{compactenum}
 and taking $\state' = \tup{{\state^1}',{\state^2}'}$ we have: 
\begin{compactenum}
\item $\state' \in \states$,
\item $\tup{\state,\transition^1,\state'} \in \transrel$,
\item $\probfun(\tup{\state,\transition^1,\state'}) = \probfun_1(\tup{\state_1,\transition^1,\state_1'})$,
\item if ${\state^1}' \in \fstates^2$ and ${\state^2}' \in \autfstates^2$, then $\state' \in \fstates$.
\end{compactenum}
\end{compactenum}
\end{definition}

\begin{remark}
\label{rem:incomplete-product}
The product system as per Definition~\ref{def:product} is not a \emph{complete} stochastic transition system: there may be states whose successor probabilities do not add up to one. However, it can be made complete as follows:
\begin{compactitem}
\item introduce a fresh, non-final sink state $\state_{sink}$, and introduce a new net transition $\transition_{sink}$ that is not used in the input \pwp;
\item for every state $\state$ in the product system such that the sum of the probabilities of its outgoing edges correspond to $p < 1$, add an edge between $\state$ and $\state_{sink}$, labelling it with $\transition_{sink}$ and assigning probability $1-p$ to it.
\end{compactitem}
As we will see next, we are going to use the (Markov chain underlying the) product system for calculating probabilities. Hence, this completion can be avoided (as state $\state_{sink}$ would get assigned a probability of 0).
\end{remark}

\medskip
\noindent
\textbf{Specification probability as outcome probability.}
We are now ready to bring everything together, exploiting the notions of silenced \DFA{s} and product systems to show how the $\verifprob$ problem can be reduced to $\outcomeprob$, invoked on $\product{\wproc}{\silentaut}$, considering all its final states.

\begin{theorem}
\label{thm:correspondence}
For every \bpwp $\wproc$ and \DFA $\aut$ over $\tasks$, we have that 
$\verifprob(\wproc,\aut) = \outcomeprob(\product{\wproc}{\silentaut},\product{\wproc}{\silentaut}.\fstates)$.
\end{theorem}
\begin{proof}
First of all, we recall the definition of $\verifprob$, considering that the probability of a model trace of $\wproc$ is the sum of the probabilities of the runs of $\wproc$ inducing that trace: $\verifprob(\wproc,\aut) = \sum_{\trace \text{ model trace of } \wproc { s.t.~}\trace \in \lang{\aut}} \prob{\trace}{\wproc} = \sum_{\run \text{ run of } \wproc \text{ inducing } \trace { s.t.~}\trace \in \lang{\aut}} \prob{\run}{\wproc}$. By Lemma~\ref{lem:dfa-correspondence}, we have that the acceptance of a trace by $\aut$ can be lifted to acceptance of the label sequence induced by its corresponding runs over $\silentaut$. This yields:
$ \verifprob(\wproc,\aut) = \sum_{\run \text{ run of } \wproc { s.t.~}\wproc.\net.\labeling(\run) \in \lang{\aut}} \prob{\run}{\wproc}$. By Definition~\ref{def:product}, the runs mentioned in this sum are precisely those of $\product{\wproc}{\silentaut}$. By the same definition, we also know that the probabilities of such runs as assigned by $\product{\wproc}{\silentaut}$ are the same as those assigned by $\wproc$. Hence: $ \verifprob(\wproc,\aut) = \sum_{\run \text{ run of } \product{\wproc}{\silentaut}} \prob{\run}{\product{\wproc}{\silentaut}} = \outcomeprob(\product{\wproc}{\silentaut},\product{\wproc}{\silentaut}.\fstates)
$.
\end{proof}



\subsection{Computing Trace Probabilities}
\label{sec:computing-trace-probability}
As already pointed out, a key problem in stochastic conformance checking \cite{LABP21,BMMP21b} is that of computing the probability of a trace in a stochastic setting. Starting from Definition~\ref{def:trace-probability}, we cast such a problem in the context of \bpwp{s} as follows:
\begin{problem}[$\traceprob(\wproc,\trace)$]
~
\begin{compactdesc}
\item[Input:] \Bpwp $\wproc$, trace $\trace$ over $\tasks$;
\item[Output:] Probability $\prob{\trace}{\wproc}$. 
\end{compactdesc}
\end{problem}

To solve this problem, we straightforwardly encode a trace $\trace$ into a corresponding \DFA over $\tasks$ that accepts exactly $\trace$. We call this the \emph{trace \DFA} of $\sigma$:
\begin{definition}[Trace \DFA]
\label{def:trace-dfa}
Given a trace $\trace=a_0,\ldots,a_n$ over $\tasks$, its \emph{trace \DFA} $\traceaut{\trace}$ is the \DFA $\tup{\tasks,\autstates,\autistate,\autfstates,\autfun}$ over $\tasks$ such that:  
\begin{compactenum}
\item $\autstates = \set{\autstate_0,\ldots,\autstate_{n+1}}$ contains $n+1$ states;
\item $\autfstates = \set{\autstate_{n+1}}$;
\item for every $i \in \set{0,\ldots,n}$, $\autfun(\autstate_i,a_i)=\autstate_{i+1}$ (and nothing else is in $\autfun$). 
\end{compactenum}
\end{definition}

This immediately yields the following key result:

\begin{theorem}
For every \bpwp $\wproc$ and every trace $\trace$ over $\tasks^*$, we have that $\traceprob(\wproc,\trace)=\verifprob(\wproc,\traceaut{\trace})$.
\end{theorem}
\begin{proof}
Direct from the definitions of the problems, noticing that $\lang{\traceaut{\trace}} = \set{\trace}$.
\end{proof}
\vspace{-3mm}

\input{product}

\begin{example}
Consider the trace $\trace_{ar} = \taskname{open},\taskname{fin},\taskname{acc},\taskname{fin},\taskname{rej}$ from Example~\ref{ex:traces}. We now compute the probability that $\myproc$ generates trace $\trace_{ar}$, where an order is filled, finalised, accepted, then modified, finalised again, and this second time rejected. Following the described technique, we first transform $\trace_{ar}$ into its trace \DFA $\traceaut{\trace_{ar}}$, and then further into its silenced trace \DFA $\silenttraceaut{\trace_{ar}}$. This is shown in Figures~\ref{fig:product}(a) and (b). We then compute the product system $\product{\rg{\myproc}}{\silenttraceaut{\trace_{ar}}}$ of $\silenttraceaut{\trace_{ar}}$  and $\rg{\myproc}$ (from Figure~\ref{fig:sample-rg}), obtaining the (incomplete, cf.~Remark~\ref{rem:incomplete-product}) stochastic transition system of Figure~\ref{fig:product}(b) - notice how silent transitions unfold in this transition system. Finally, we construct $\eqsys{\product{\rg{\myproc}}{\silenttraceaut{\trace_{ar}}}}{\set{\tup{7,5}}}$. Recalling that by definition $\pval{s_1}+\pval{s_2}=1$, we then get:
\begin{align*}
x_{55} &=1\\
x_{44} &= \pval{s_3} x_{54} = \pval{s_3} x_{55} = \pval{s_3}\\
x_{23} &= \pval{s_1} x_{13} + \pval{s_2} x_{33} = \pval{s_1} x_{23} + \pval{s_2} x_{33} = \frac{\pval{s_2}}{1-\pval{s_1}}x_{33}\\ 
x_{33} &= \pval{f} x_{44} + \pval{i_2} x_{23} = \pval{f} \pval{s_3} + \frac{\pval{i_2}\pval{s_2}}{1-\pval{s_1}}x_{33}
= 
\frac{\pval{f} \pval{s_3}}{1-\pval{i_2}}\\
x_{42} &= \pval{s_4} x_{62} = \pval{s_4} x_{73} = \pval{s_4}\pval{i_3}x_{33} =\frac{\pval{s_4}\pval{i_3}\pval{f} \pval{s_3}}{1-\pval{i_2}}\\
x_{31} &= \pval{i_2}x_{21} + \pval{f}x_{42} = 
\pval{i_2}x_{21} + \frac{\pval{f}\pval{s_4}\pval{i_3}\pval{f} \pval{s_3}}{1-\pval{i_2}}\\
x_{21} &= \pval{s_1}x_{11} + \pval{s_2}x_{31} = 
\pval{s_1}x_{11} + \pval{s_2}\pval{i_2}x_{21} + \frac{\pval{s_2}\pval{f}\pval{s_4}\pval{i_3}\pval{f} \pval{s_3}}{1-\pval{i_2}} = \\
&= \frac{\pval{s_1}}{1-\pval{s_2}\pval{i_2}}x_{11} + \frac{\pval{s_2}\pval{f}\pval{s_4}\pval{i_3}\pval{f} \pval{s_3}}{(1-\pval{i_2})(1-\pval{s_2}\pval{i_2})}\\
x_{00} &= x_{11} = x_{21} = \frac{\pval{s_1}}{1-\pval{s_2}\pval{i_2}}x_{00} + \frac{\pval{s_2}\pval{f}\pval{s_4}\pval{i_3}\pval{f} \pval{s_3}}{(1-\pval{i_2})(1-\pval{s_2}\pval{i_2})}
\\
&= \frac{\pval{s_2}\pval{f}\pval{s_4}\pval{i_3}\pval{f} \pval{s_3}}{(1-\pval{i_2})(1-\pval{s_2}\pval{i_2})}
\frac{1-\pval{s_2}\pval{i_2}}{\pval{s_2}(1-\pval{i_2})}
= \frac{\pval{f}\pval{s_4}\pval{i_3}\pval{f} \pval{s_3}}{(1-\pval{i_2})^2}
\end{align*}
%
which yields the solution to the $\traceprob(\myproc,\trace)$ problem. If all transition weights have the same value (as in Example~\ref{ex:outcome}), we get as result probability $\frac{1/3\cdot 1/2 \cdot 1/3 \cdot 1/3 \cdot 1/2}{(1-1/3)^2} = \frac{1}{48} \sim 0.02$.
\end{example}

%% file: sample-automata.tex
\begin{figure}[t]
\centering
\resizebox{\textwidth}{!}{
\begin{tikzpicture}[node distance=5mm and 8mm]

\node[
  autstate,
  label=center:$\autstate_0$,
  double
] (s0) {};

\node[
  autstate,
  right=of s0,
  label=center:$\autstate_1$
] (s1) {};

\draw[arc] ($(s0.north west)+(-3mm,3mm)$) -- (s0);

\draw[
  downloop
] 
(s0) 
edge node[below] {$\tasks \setminus \set{\taskname{open}}$} 
(s0);

\draw[
  out=30,
  in=150,
  arc
] 
(s0)
edge node[above] {$\taskname{open}$}
(s1);

\draw[
  downloop
] 
(s1) 
edge node[below] {$\tasks \setminus \set{\taskname{ship}}$} 
(s1);

\draw[
  out=-150,
  in=-30,
  arc
] 
(s1)
edge node[below] {$\taskname{ship}$}
(s0);

\node[
  below=15mm of s0,
  xshift=10mm
] {(a) Every \taskname{open} is followed by \taskname{pay}};


\node[
  autstate,
  label=center:$\autstate_0$,
  right=34mm of s1
] (s0) {};

\node[
  autstate,
  right=of s0,
  label=center:$\autstate_1$
] (s1) {};

\node[
  autstate,
  below=of s1,
  label=center:$\autstate_2$,
  double
] (s2) {};

\draw[arc] ($(s0.north west)+(-3mm,3mm)$) -- (s0);

\draw[
  downloop
] 
(s0) 
edge node[below] {$\tasks \setminus \set{\taskname{can}}$} 
(s0);

\draw[
  arc
] 
(s0)
edge node[above] {$\taskname{rej}$}
(s1);

\draw[
  rightloop,
] 
(s1) 
edge node[right] {$\tasks \setminus \set{\taskname{pay}}$} 
(s1);

\draw[
  arc
] 
(s1)
edge node[left] {$\taskname{pay}$}
(s2);

\draw[
  rightloop
] 
(s2) 
edge node[right] {$\tasks$} 
(s2);

\node[
  below=15mm of s0,
  xshift=10mm
] {(b) At some point \taskname{rej} and then \taskname{pay}};


\node[
  autstate,
  label=center:$\autstate_0$,
  right=32mm of s1,
  double
] (s0) {};

\node[
  autstate,
  right=of s0,
  label=center:$\autstate_1$,
  double
] (s1) {};

\node[
  autstate,
  below=of s1,
  label=center:$\autstate_2$,
] (s2) {};

\draw[arc] ($(s0.north west)+(-3mm,3mm)$) -- (s0);

\draw[
  downloop
] 
(s0) 
edge node[below] {$\tasks \setminus \set{\taskname{fin},\taskname{acc}}$} 
(s0);

\draw[
  arc
] 
(s0)
edge node[above] {$\taskname{fin}$}
(s1);

\draw[
  rightloop,
] 
(s1) 
edge node[right] {$\tasks$} 
(s1);

\draw[
  arc
] 
(s0)
edge node[sloped,yshift=2mm] {$\taskname{acc}$}
(s2);

\draw[
  rightloop
] 
(s2) 
edge node[right] {$\tasks$} 
(s2);

\node[
  below=15mm of s0,
  xshift=10mm
] {(c) \taskname{Acc} only possible after \taskname{fin}};

\end{tikzpicture}
}
\caption{\DFA{s} of three properties for the order-to-cash example. A single edge labelled by a set $\tasks$ of task names compactly describes a set of edges, each labelled by a task name from $\tasks$.}
\label{fig:sample-automata}
\end{figure}

%% file: sample-silenced-automata.tex
\begin{figure}[t]
\centering
\resizebox{\textwidth}{!}{
\begin{tikzpicture}[node distance=5mm and 8mm]

\node[
  autstate,
  label=center:$\autstate_0$,
  double
] (s0) {};

\node[
  autstate,
  right=of s0,
  label=center:$\autstate_1$
] (s1) {};

\draw[arc] ($(s0.north west)+(-3mm,3mm)$) -- (s0);

\draw[
  downloop
] 
(s0) 
edge node[below] {$\tasks \setminus \set{\taskname{open}}$} 
(s0);

\draw[
  out=30,
  in=150,
  arc
] 
(s0)
edge node[above] {$\taskname{open}$}
(s1);

\draw[
  downloop
] 
(s1) 
edge node[below] {$\tasks \setminus \set{\taskname{ship}}$} 
(s1);

\draw[
  out=-150,
  in=-30,
  arc
] 
(s1)
edge node[below] {$\taskname{ship}$}
(s0);

\draw
(s0)
edge[uploop,silent] 
node[above] {$\hidden$}
(s0);

\draw
(s1)
edge[uploop,silent] 
node[above] {$\hidden$}
(s1);

\node[
  below=23mm of s0,
  xshift=10mm
] {(a) Every \taskname{open} is followed by \taskname{pay}};


\node[
  autstate,
  label=center:$\autstate_0$,
  right=34mm of s1
] (s0) {};

\node[
  autstate,
  right=of s0,
  label=center:$\autstate_1$
] (s1) {};

\node[
  autstate,
  below=of s1,
  label=center:$\autstate_2$,
  double
] (s2) {};

\draw[arc] ($(s0.north west)+(-3mm,3mm)$) -- (s0);

\draw[
  downloop
] 
(s0) 
edge node[below] {$\tasks \setminus \set{\taskname{can}}$} 
(s0);

\draw[
  arc
] 
(s0)
edge node[above] {$\taskname{rej}$}
(s1);

\draw[
  rightloop,
] 
(s1) 
edge node[right] {$\tasks \setminus \set{\taskname{pay}}$} 
(s1);

\draw[
  arc
] 
(s1)
edge node[left] {$\taskname{pay}$}
(s2);

\draw[
  rightloop
] 
(s2) 
edge node[right] {$\tasks$} 
(s2);

\draw
(s0)
edge[uploop,silent] 
node[above] {$\hidden$}
(s0);

\draw
(s1)
edge[uploop,silent] 
node[above] {$\hidden$}
(s1);

\draw
(s2)
edge[downloop,silent] 
node[below] {$\hidden$}
(s2);

\node[
  below=23mm of s0,
  xshift=10mm
] {(b) At some point \taskname{rej} and then \taskname{pay}};


\node[
  autstate,
  label=center:$\autstate_0$,
  right=32mm of s1,
  double
] (s0) {};

\node[
  autstate,
  right=of s0,
  label=center:$\autstate_1$,
  double
] (s1) {};

\node[
  autstate,
  below=of s1,
  label=center:$\autstate_2$,
] (s2) {};

\draw[arc] ($(s0.north west)+(-3mm,3mm)$) -- (s0);

\draw[
  downloop
] 
(s0) 
edge node[below] {$\tasks \setminus \set{\taskname{fin},\taskname{acc}}$} 
(s0);

\draw[
  arc
] 
(s0)
edge node[above] {$\taskname{fin}$}
(s1);

\draw[
  rightloop,
] 
(s1) 
edge node[right] {$\tasks$} 
(s1);

\draw[
  arc
] 
(s0)
edge node[sloped,yshift=2mm] {$\taskname{acc}$}
(s2);

\draw[
  rightloop
] 
(s2) 
edge node[right] {$\tasks$} 
(s2);

\draw
(s0)
edge[uploop,silent] 
node[above] {$\hidden$}
(s0);

\draw
(s1)
edge[uploop,silent] 
node[above] {$\hidden$}
(s1);

\draw
(s2)
edge[downloop,silent] 
node[below] {$\hidden$}
(s2);

\node[
  below=23mm of s0,
  xshift=10mm
] {(c) \taskname{Acc} only possible after \taskname{fin}};

\end{tikzpicture}
}
\caption{Silenced \DFA{s} corresponding to the three \DFA{s} in Figure~\ref{fig:sample-automata}.}
\label{fig:sample-silenced-automata}
\end{figure}

%% file: product.tex
\newcommand{\astate}[1]{\autstate_{#1}}

\begin{figure}[t]
\centering
\resizebox{\linewidth}{!}{
\begin{tikzpicture}[node distance=8mm and 8mm]

\node[
  autstate,
  label=center:$\autstate_0$,
] (s0) {};

\node[
  autstate,
  below=of s0,
  label=center:$\autstate_1$
] (s1) {};

\node[
  autstate,
  below=of s1,
  label=center:$\autstate_2$
] (s2) {};

\node[
  autstate,
  below=of s2,
  label=center:$\autstate_3$
] (s3) {};

\node[
  autstate,
  below=of s3,
  label=center:$\autstate_4$
] (s4) {};

\node[
  autstate,
  below=of s4,
  label=center:$\autstate_5$,
  double
] (s5) {};

\draw[arc] ($(s0.north west)+(-3mm,3mm)$) -- (s0);

\draw[
  arc
] 
(s0)
edge node[right] {$\taskname{open}$}
(s1);

\draw[
  arc
] 
(s1)
edge node[right] {$\taskname{fin}$}
(s2);

\draw[
  arc
] 
(s2)
edge node[right] {$\taskname{acc}$}
(s3);

\draw[
  arc
] 
(s3)
edge node[right] {$\taskname{fin}$}
(s4);

\draw[
  arc
] 
(s4)
edge node[right] {$\taskname{rej}$}
(s5);

\node[below = 3mm of s5] (captionaut)
{(a) $\traceaut{\trace_{ar}}$};


\node[
  autstate,
  right=of s0,
  label=center:$\autstate_0$,
] (s0) {};

\node[
  autstate,
  below=of s0,
  label=center:$\autstate_1$
] (s1) {};

\node[
  autstate,
  below=of s1,
  label=center:$\autstate_2$
] (s2) {};

\node[
  autstate,
  below=of s2,
  label=center:$\autstate_3$
] (s3) {};

\node[
  autstate,
  below=of s3,
  label=center:$\autstate_4$
] (s4) {};

\node[
  autstate,
  below=of s4,
  label=center:$\autstate_5$,
  double
] (s5) {};

\draw[arc] ($(s0.north west)+(-3mm,3mm)$) -- (s0);

\draw[
  rightloop,
  silent
] 
(s0) 
edge node[right] {$\hidden$} 
(s0);

\draw[
  rightloop,
  silent
] 
(s1) 
edge node[right] {$\hidden$} 
(s1);

\draw[
  rightloop,
  silent
] 
(s2) 
edge node[right] {$\hidden$} 
(s2);

\draw[
  rightloop,
  silent
] 
(s3) 
edge node[right] {$\hidden$} 
(s3);

\draw[
  rightloop,
  silent
] 
(s4) 
edge node[right] {$\hidden$} 
(s4);

\draw[
  rightloop,
  silent
] 
(s5) 
edge node[right] {$\hidden$} 
(s5);

\draw[
  arc
] 
(s0)
edge node[right] {$\taskname{open}$}
(s1);

\draw[
  arc
] 
(s1)
edge node[right] {$\taskname{fin}$}
(s2);

\draw[
  arc
] 
(s2)
edge node[right] {$\taskname{acc}$}
(s3);

\draw[
  arc
] 
(s3)
edge node[right] {$\taskname{fin}$}
(s4);

\draw[
  arc
] 
(s4)
edge node[right] {$\taskname{rej}$}
(s5);

\node[right=4mm of captionaut] (captionaut)
{(b) $\silenttraceaut{\trace_{ar}}$};

\node[
  lprstate={\sstate{s}}{\astate{0}},
  right= 15mm of s0,
] (00) {};

\draw[arc] ($(00.north west)+(-3mm,3mm)$) -- (00.north west);

\node[
  lprstate={\sstate{1}}{\astate{1}},
  below = of 00,
] (11) {};

\draw
  (00)
  edge[arc]
  node[right] {1}
  node[left] {\taskname{open}}
  (11);

\node[
  lprstate={\sstate{2}}{\astate{1}},
  right = of 11,
] (21) {};

\draw
  (11)
  edge[arc,silent]
  node[below] {1}
  node[above] {\hidden}
  (21);

\draw[out=-90,in=-90]
  (21)
  edge[arc,silent]
  node[below] {\pval{s_1}}
  node[above] {\hidden}
  (11);

\node[
  lprstate={\sstate{3}}{\astate{1}},
  right = of 21,
] (31) {};

\draw
  (21)
  edge[arc,silent]
  node[below] {\pval{s_2}}
  node[above] {\hidden}
  (31);

\draw[out=90,in=90]
  (31)
  edge[arc,silent]
  node[below] {\pval{i_2}}
  node[above] {\hidden}
  (21);

\node[
  lprstate={\sstate{4}}{\astate{2}},
  below = of 31,
] (42) {};

\draw
  (31)
  edge[arc]
  node[right] {\pval{f}}
  node[left] {\taskname{fin}}
  (42);

\node[
  lprstate={\sstate{6}}{\astate{2}},
  right = of 42,
] (62) {};

\draw
  (42)
  edge[arc,silent]
  node[below] {\pval{s_4}}
  node[above] {\hidden}
  (62);

\node[
  lprstate={\sstate{7}}{\astate{3}},
  below = of 62,
] (73) {};

\draw
  (62)
  edge[arc]
  node[right] {1}
  node[left] {\taskname{acc}}
  (73);

\node[
  lprstate={\sstate{3}}{\astate{3}},
  right = of 73,
] (33) {};

\draw
  (73)
  edge[arc,silent]
  node[below] {\pval{i_3}}
  node[above] {\hidden}
  (33);

\node[
  lprstate={\sstate{2}}{\astate{3}},
  right = of 33,
] (23) {};

\draw[out=90,in=90]
  (33)
  edge[arc,silent]
  node[below] {\pval{i_2}}
  node[above] {\hidden}
  (23);

\draw
  (23)
  edge[arc,silent]
  node[below] {\pval{s_2}}
  node[above] {\hidden}
  (33);

\node[
  lprstate={\sstate{1}}{\astate{3}},
  right = of 23,
] (13) {};

\draw[out=-90,in=-90]
  (23)
  edge[arc,silent]
  node[below] {\pval{s_1}}
  node[above] {\hidden}
  (13);

\draw
  (13)
  edge[arc,silent]
  node[below] {1}
  node[above] {\hidden}
  (23);

\node[
  lprstate={\sstate{4}}{\astate{4}},
  below = of 33,
] (44) {};

\draw
  (33)
  edge[arc]
  node[right] {\pval{f}}
  node[left] {\taskname{fin}}
  (44);

\node[
  lprstate={\sstate{5}}{\astate{4}},
  right = of 44,
] (54) {};

\draw
  (44)
  edge[arc,silent]
  node[below] {$\pval{s_3}$}
  node[above] {\hidden}
  (54);

\node[
  lprstate={\sstate{5}}{\astate{5}},
  double,
  below = of 54,
] (55) {};

\draw
  (54)
  edge[arc]
  node[right] {1}
  node[left] {\taskname{rej}}
  (55);


\node[
  lprstate={\sstate{c}}{\astate{1}},
  deadend,
] at (s0-|62) (c1) {};

\draw
  (31)
  edge[arc,deadend]
  node[below right] {\pval{c_1}}
  node[above left] {\hidden}
  (c1);

\node[
  lprstate={\sstate{5}}{\astate{2}},
  deadend,
] at (s0-|33) (52) {};

\draw
  (42)
  edge[arc,deadend]
  node[below right] {\pval{s_3}}
  node[above left] {\hidden}
  (52);

\node[
  lprstate={\sstate{c}}{\astate{3}},
  deadend,
] at (42|-55) (c3) {};

\draw
  (33)
  edge[arc,deadend]
  node[below right] {\pval{c_1}}
  node[above left] {\hidden}
  (c3);

\draw[out=-150,in=90]
  (73)
  edge[arc,deadend]
  node[right] {\pval{c_2}}
  node[left] {\hidden}
  (c3);

\node[
  lprstate={\sstate{5}}{\astate{6}},
  deadend,
] at (c3-|73) (56) {};

\draw
  (44)
  edge[arc,deadend]
  node[below right] {\pval{s_4}}
  node[above left] {\hidden}
  (56);


\node
at (captionaut-|62)
{(c) product system $\product{\rg{\myproc}}{\silenttraceaut{\trace_{ar}}}$};

\end{tikzpicture}
}
\caption{(a) Trace \DFA and (b) silenced trace \DFA for trace $\trace=\taskname{open},\taskname{fin},\taskname{acc},\taskname{fin},\taskname{rej}$, and (c) product system with the reachability graph of Figure~\ref{fig:sample-rg}; the portion of the transition system that cannot reach the final state is greyed out.}
\label{fig:product}
\end{figure}

%% file: hybrid.tex
\newcommand{\pdm}{\mathcal{D}}
\newcommand{\cset}{\mathcal{C}}
\newcommand{\sset}[1]{\mathbb{S}_{#1}}
\newcommand{\scenario}{\mathcal{S}}
\newcommand{\cformula}[1]{\Phi_{#1}}

\newcommand{\mypdm}{\pdm_{\text{order}}}

\newcommand{\scompliance}{\textsc{s-complies}}

\newcommand{\iprob}[2]{\mathbb{P}_{#1}^{#2}}

\section{Stochastic Compliance with Probabilistic Declare}
\label{sec:delta}
We now employ the verification machinery from Section~\ref{sec:specification-probability} to deal with probabilistic temporal properties. In particular, we introduce \emph{stochastic compliance}, where the stochastic behaviour induced by a \bpwp is contrasted with a set of probabilistic temporal constraints, each expressing a qualitative specification paired with a condition on the expected probability that the specification is indeed satisfied. Since multiple probabilistic constraints can be given, combined conditions on the probability of their satisfaction/violation must be obtained. To systematically handle this, we rely on the \pdeclare language from \cite{MMPA20}, which extends the well-known \declare language \cite{PeSV07,MPVC10} with constraint probabilities. Before defining and solving our stochastic conformance problem, we briefly recall \pdeclare.

\subsection{Probabilistic \declare}

\declare is a constraint-based declarative language for modelling flexible processes \cite{PeSV07,MPVC10}. The acceptable behaviours supported by a process are implicitly characterised using \LTLf constraints: a \declare specification consists of a set of \LTLf constraints, which must be all  respected in a complete process execution. This comes with a \emph{crisp} interpretation of constraints: a complete trace satisfies a \declare specification if it satisfies \emph{every} constraints contained therein. For an overall introduction of \declare and corresponding key analysis and process mining tasks, the interested reader is referred to \cite{DiCM22}.

By exploiting a probabilistic version of \LTLf as underlying temporal logic \cite{MaMP20}, the crisp semantics of \declare has been lifted to an uncertain one in \cite{MMPA20}, opening up a full range of new tasks for (probabilistic) declarative process mining \cite{AMMP22}. In the resulting \pdeclare framework, each constraint comes with a probability condition identifying a set of probabilities for which a trace generated by the process satisfies that constraint. More precisely, the semantics is based on stochastic languages, which assign a probability mass to each trace in a (possibly infinite) set, as recalled in Definition~\ref{def:stochastic-language}. Specifically, a  stochastic language satisfies a probabilistic constraint if, considering only those traces in the language that satisfy the constraint, their overall probability mass satisfies the probability condition attached to the constraint. This makes \pdeclare a natural candidate to express probabilistic temporal properties over livelock-free \bpwp{s}, whose model traces with their probabilities indeed capture stochastic languages (cf.~Corollary~\ref{thm:lsp-stochastic-language}). 

In \pdeclare, probabilities are assigned to traces, compatibly with the semantics of \bpwp{s}. On the other hand, as recalled in Remark~\ref{rem:dfa-properties}, the task of computing the probability that a \bpwp meets a specification provided with a \DFA can be seamlessly used to capture the case where the specification is given in \LTLf. Hence, we do not need to introduce \LTLf in the context of the current article - we simply rely on informal presentation of constraints, formalizing them directly using \DFA{s}. For a thorough introduction to \LTLf and its automata-theoretic characterisation, see \cite{DeVa13,DDMM22}.

The main question for \pdeclare is how to characterise which stochastic languages satisfy \emph{a set} of probabilistic constraints, considering how the temporal and probability components of each constraint interact with those of the others. We recall from \cite{AMMP22} the necessary definitions to obtain such a characterisation.

\begin{definition}[Probabilistic constraint]
\label{def:probabilistic-constraint}
  A \emph{probabilistic constraint} over $\tasks$ is a triple $\tup{\varphi,\bowtie,p}$, where:
  \begin{inparaenum}[\it (i)]
\item $\varphi$ is an \LTLf formula over $\tasks$ representing the \emph{constraint formula};
\item ${\bowtie}\in\{=,\neq,\le,\ge,<,>\}$ is the \emph{constraint probability operator};
\item  $p$ is a rational value in $[0,1]$ representing the \emph{constraint probability}.
\end{inparaenum}
\end{definition}

\begin{definition}[\pdeclare specification]
\label{def:probdeclare-specification}
A \emph{\pdeclare specification} is a set $\pdm$ of probabilistic constraints over $\tasks$.
\end{definition}

Differently from \cite{MMPA20,AMMP22}, we do not single out crisp constraints from genuinely probabilistic constraints. This is just for simplicity of presentation, recalling that a crisp constraint $\varphi$ can be represented as the corresponding probabilistic constraint $\tup{\varphi,=,1}$.

\begin{definition}[\pdeclare satisfaction relation]
\label{def:probdeclare-satisfaction}
A stochastic language $\Psi$ \emph{satisfies} a \pdeclare specification $\pdm$ if for every probabilistic constraint $\tup{\varphi,\bowtie,p} \in \pdm$, we have that:
$$
\sum_{\trace \in \tasks^* \text{ s.t. } \trace \in \lang{\aut_\varphi}} \Psi(\trace) \bowtie p
$$
where $\aut_\varphi$ is the (minimal) \DFA capturing the language of $\varphi$.
\end{definition}

\begin{example}
\label{ex:pdeclare}
Consider the order-to-cash process of our running example. We introduce a \pdeclare specification $\mypdm$ containing three probabilistic constraints:
\begin{compactitem}[$\bullet$]
\item $\Phi_{pr}$ is a crisp constraint expressing that the \emph{not coexistence} between \taskname{pay} and \taskname{ack reject} (namely that \taskname{pay} and \taskname{ack reject} cannot both occur in the same trace) must hold with probability 1;
\item $\Phi_{op}$ is a probabilistic constraint expressing that the \emph{response} from \taskname{open} to \taskname{pay} (namely that whenever \taskname{open} occurs then \taskname{pay} later occurs as well) must hold with a probability $\geq \frac{1}{20}$;
\item $\Phi_{or}$ is a probabilistic constraint expressing that the \emph{response} from \taskname{open} to \taskname{ack reject} (namely that whenever \taskname{open} occurs then \taskname{ack reject} later occurs as well) must hold with a probability $\leq \frac{1}{4}$.
\end{compactitem}
\end{example}

\subsection{Checking Stochastic Compliance}
To check whether a livelock-free \bpwp stochastically complies with a \pdeclare specification, we simply verify whether the stochastic language induced by the \pwp indeed satisfies all probabilistic constraints with their probability conditions.

\begin{problem}[$\scompliance(\wproc,\pdm)$]
~
\begin{compactdesc}
\item[Input:] livelock-free \bpwp $\wproc$, \pdeclare specification $\pdm$;
\item[Output:] \emph{yes} if the induced stochastic language of $\wproc$ (as per Definition~\ref{def:induced-stochastic-language}) satisfies $\pdm$ in the sense of Definition~\ref{def:probdeclare-satisfaction}, \emph{no} otherwise.
\end{compactdesc}
\end{problem}

We can solve this problem through iterated calls to the $\verifprob$ problem, using the \DFA of each probabilistic constraint, and checking whether the returned probability satisfies the probability condition associated to that constraint.

\begin{theorem}
\label{thm:semantic-compliance}
Consider a livelock-free \bpwp $\wproc$ and a \pdeclare specification $\pdm$. Then $\scompliance(\wproc,\pdm)$ returns yes if and only if for every probabilistic constraint $\tup{\varphi, \bowtie, p} \in \pdm$, we have that $\verifprob(\wproc,\aut_\varphi) \bowtie p$.
\end{theorem}

\begin{proof}
Fix a probabilistic constraint $\tup{\varphi, \bowtie, p} \in \pdm$. By Definitions \ref{def:probdeclare-satisfaction} and \ref{def:induced-stochastic-language}, we have that 
$$
\sum_{\trace \in \tasks^* \text{ s.t. } \trace \in \lang{\aut_\varphi}} \prob{\trace}{\wproc} = 
\sum_{\trace \text{ model trace of } \wproc \text{ s.t. } \trace \in \lang{\aut_\varphi}} \prob{\trace} {\wproc}
$$
Indeed, notice that the stochastic language induced by $\wproc$ is so that $\prob{\trace}{\wproc}>0$ only if $\trace$ is a model trace of $\wproc$. This sum is, in turn, equal to $\prob{\aut_\varphi}{\wproc}$, which, by definition, is the output of $\verifprob(\wproc,\aut_\varphi)$.
\end{proof}

\begin{example}
Consider $\runningwproc$ under the assumption that all weights are equal (i.e., the working hypothesis of Example~\ref{ex:outcome}). We now show that under this assumption $\runningwproc$ stochastically complies with the \pdeclare specification $\mypdm$ from Example~\ref{ex:pdeclare}. We can do so by simply applying Theorem~\ref{thm:semantic-compliance}. To witness semantic compliance, we notice the following:
\begin{compactitem}
\item Every model trace of $\runningwproc$  contains either \taskname{pay}, or \taskname{ack reject}, or none of the two, and thus satisfies the \emph{not coexistence} between \taskname{pay} and \taskname{ack reject}; by \ref{thm:total-absorption-one}, we know that the sum of the probability values of all model traces of \runningwproc is indeed 1 (as requested by the probabilistic constraint $\Phi_{pr} \in \mypdm$).
\item Considering the structure of $\rg{\runningwproc}$, we have that model traces satisfying the \emph{response} from \taskname{open} to \taskname{pay} are simply those that culminate in final marking $[\pname{h}]$. By Example~\ref{ex:outcome}, we know that the overall probability of such model traces is $\frac{1}{11}$, which is indeed $\geq \frac{1}{20}$ (as requested by the probabilistic constraint $\Phi_{op} \in \mypdm$); 
\item Considering the structure of $\rg{\runningwproc}$, we have that model traces satisfying the \emph{response} from \taskname{open} to \taskname{reject} are simply those that culminate in final marking $[\pname{r}]$. By Example~\ref{ex:outcome}, we know that the overall probability of such model traces is $\frac{6}{11}$, which is indeed $\leq \frac{1}{4}$ (as requested by the probabilistic constraint $\Phi_{or} \in \mypdm$). 
\end{compactitem}

\end{example}

In case $\scompliance$ returns a negative answer, the standard Earth Mover's Distance can be used to measure the deviation between the probability distribution by the livelock-free \bpwp under scrutiny, and what expected by the probabilistic constraint in the reference \pdeclare specification. To do so, one has to follow the approach in \cite{AMMP22}, computing the different \emph{scenarios} induced by the \pdeclare specification and their probabilities. This realises a form of \emph{stochastic delta analysis}.

%% file: evaluation.tex
\section{A Practical Application \& its Evaluation}
\label{sec:evaluation}
    To illustrate the practical applicability and feasibility of the verification problem, we implemented it, and used this implementation in a previously incomplete stochastic conformance measure. 
    
    \subsection{Implementation}
    \label{sec:implementation}
        Our solutions to the trace and temporal properties probability have been implemented as plug-ins of the ProM framework~\cite{DBLP:conf/apn/DongenMVWA05}, as part of the \verb=StochasticLabelledPetriNets= package.
        This package provides stochastic PNPs as first-class objects in ProM, an import plug-in, an export plug-in and a visualiser.
        The package further contains an implementation of the trace probability approach (cf.~Section~\ref{sec:computing-trace-probability}): given a trace, this approach constructs an explicit product stochastic transition system, after which a linear problem is solved using Lpsolve~\cite{Berkelaar2004} to obtain the trace probability.
        Besides a trace, the implementation provides the machinery to accept any type of DFA.
        
        This code is intended for developers, however we used it to complete an existing method: Unit Earth Movers' Stochastic Conformance (uEMSC)~\cite{LeeSA19}\footnote{Notice that uEMSC is then complete for bounded models only.}.
        This stochastic conformance checking measure expresses the difference between a stochastic process model and an event log, which are for our purposes seen as two distributions of traces.
        Intuitively, the measure indicates the amount of earth (i.e. probability mass) that needs to be moved to transform one of the distributions into the other.
        From~\cite{LeeSA19}, uEMSC of a log $L$ and a stochastic process model $\wproc$ is $1 - \sum_{\trace \in L} \max(L(\trace) - \prob{\trace}{\wproc}, 0)$, in which $L(\trace)$ denotes the likelihood of trace $\trace$ in $L$.
        Before, approximations of $\prob{\trace}{\wproc}$ were necessary to compute this otherwise rather simple function \cite{LABP21,BMMP21b}.
        
    \subsection{Applicability and Feasibility}
        \begin{figure}
            \centering
            \begin{tikzpicture}[font=\footnotesize]
                \def\dx{1.3cm}
                \def\dy{0.15cm}
                \node (log) {log};
                \node (discovery) [right=\dx of log, align=center, draw, rounded corners] {discovery};
                \node (stochdiscovery) [right=\dx of discovery, align=center, draw, rounded corners] {stochastic\\discovery};
                \node (measure) [right=\dx of stochdiscovery, draw, align=center, rounded corners] {measure};
                
                \draw [->] (log) to (discovery);
                \draw [->] (discovery) to node[above] {model\strut} (stochdiscovery);
                \draw [->] (log) to ++(0,-0.65cm) -| node[above, pos=0.25, yshift=-1mm] {resample} (stochdiscovery);
                \draw [->] (stochdiscovery) to node[above, align=center] {stochastic\strut} node[below] {model} (measure);
                \draw [->] (log) |- ++(1cm,0.65cm) -| (measure);
            \end{tikzpicture}
            \caption{Feasibility evaluation set-up.}
            \label{fig:eval:setup}
        \end{figure}
        
        \begin{table}
            \scriptsize
            \caption{Details of the set-up of the feasibility evaluation.}
            \label{tbl:eva}
            \begin{minipage}[t]{0.5\linewidth}
                \begin{subtable}[t]{\linewidth}
                    \caption{Logs.}
                    \begin{tabularx}{\linewidth}{@{}Xl}
                        \toprule
                        BPIC11~\cite{bpic2011} & BPIC11\\
                        BPIC12-approvals~\cite{bpic2012} & BPIC12-a\\
                        BPIC12-approvals and work~\cite{bpic2012} & BPIC12-aw\\
                        BPIC13-closed problems~\cite{bpic2013:closed} & BPIC13-cp\\
                        BPIC13-incidents~\cite{bpic2013:incidents} & BPIC13-i\\
                        BPIC13-open problems~\cite{bpic2013:open} & BPIC13-op\\
                        BPIC17-offers~\cite{bpic17:offer} & BPIC17-o\\
                        BPIC18-parcel document~\cite{bpic2018} & BPIC18-6\\                    
                        BPIC20-domestic declarations~\cite{bpic2020} & BPIC20-dd\\
                        BPIC20-international declarations~\cite{bpic2020} & BPIC20-id\\
                        BPIC20-prepaid travel costs~\cite{bpic2020} & BPIC20-pt\\
                        BPIC20-request for payment~\cite{bpic2020} & BPIC20-rf\\
                        MIMIC services~\cite{johnson2016mimic}\ref{app:logs} & mimic-serv\\
                        MIMIC transfers~\cite{johnson2016mimic}\ref{app:logs} & mimic-trans\\
                        Road fines~\cite{roadfines} & roadfines\\
                        Sepsis cases~\cite{sepsis} & sepsis\\
                        \bottomrule
                    \end{tabularx}
                \end{subtable}
            \end{minipage}%
            \begin{minipage}[t]{0.5\linewidth}
                \begin{subtable}[t]{\linewidth}
                    \caption{Discovery techniques for control flow.}
                    \begin{tabularx}{\linewidth}{Xl@{}}
                    \toprule
                        Inductive Miner - infrequent (0.8)~\cite{DBLP:conf/bpm/LeemansFA13} transformed to a Petri net & IMf\\
                        Directly Follows Model Miner (0.8)~\cite{DBLP:conf/icpm/LeemansPW19} transformed to a Petri net & DFM\\
                        Flower Model -- a model consisting of any behaviour of the observed alphabet & FM\\
                    \bottomrule
                    \end{tabularx}
                \end{subtable}\\
                \begin{subtable}[t]{\linewidth}
                    \caption{Stochastic discovery techniques.}
                    \begin{tabularx}{\linewidth}{Xl@{}}
                        \toprule
                        Frequency-based estimator~\cite{DBLP:conf/icpm/BurkeLW20} & FBE\\
                        Bill Clinton estimator~\cite{DBLP:conf/icpm/BurkeLW20} & BCE\\
                        Alignment-based estimator~\cite{DBLP:conf/icpm/BurkeLW20} & ABE\\
                        \bottomrule
                    \end{tabularx}
                \end{subtable}\\
                \begin{subtable}[t]{\linewidth}
                    \caption{Measures.}
                    \begin{tabularx}{\linewidth}{Xl@{}}
                    \toprule
                        Number of traces, events, activities in log & log\\
                        Number of transitions, silent transitions in control-flow model & transitions\\
                        Earth Movers' Conformance Checking~\cite{LABP21} model sample up to 50\,000 traces & EMSC\\
                        Unit Earth Movers' Conformance Checking (Sec.~\ref{sec:implementation}) & uEMSC\\
                    \bottomrule
                    \end{tabularx}
                \end{subtable}
            \end{minipage}
        \end{table}
        
        In this section, we empirically show the applicability of our solution of the verification problem.
        We do this by using our solution as a trace-probability predictor method in the uEMSC method, which was previously incomplete.
        Unit Earth Movers' Stochastic Conformance (uEMSC) takes an event log and a stochastic process model (in our cases, a labelled stochastic Petri net).
        uEMSC considers both the log and the stochastic model as a distribution over traces.
        Then, it considers the earth movers' distance of these distributions: how much effort does it take to transform one distribution into the other, in terms of moved probability mass and the distance over which it has to be moved.
        The uEMSC measure considers a unit distance between traces -- traces are either equivalent or different -- and thus is a measure of the agreement in stochastic languages of log and model.
        
        In uEMSC, our solution of the verification problem was included to assess the probability of a given trace in the model.
        In this section, we illustrate the practical applicability of this implementation, using an experiment, summarised in Figure~\ref{fig:eval:setup} and detailed in Table~\ref{tbl:eva}.
        We take a number of publicly available real-life event logs, and apply a control flow discovery technique to obtain a process model (guaranteed bounded and without stochastic perspective).
        To avoid measuring on training data, the event log is resampled 10 times.
        To each resampled log and the discovered control-flow model, we apply a stochastic process discovery technique to obtain a stochastic process model.
        Then, the quality of this model is assessed with respect to the original log using several measures.

        Figure~\ref{fig:eva} shows a summary of the results; Table~\ref{tbl:results} in~\ref{app:results} shows the results in detail.

        \begin{figure}[t]
            \centering
            \begin{subfigure}[b]{0.5\linewidth}
                \begin{tikzpicture}
                    \begin{axis}[
                        xmin = 0, xmax = 1,
                        ymin = 0, ymax = 1,
                        xlabel = uEMSC,
                        ylabel = EMSC,
                        legend pos=north west,
                        width=\linewidth,
                        mark=o,
                        point meta=explicit symbolic,
                        scatter/classes={
                            IMf={blue},%
                            DFM={red},%
                            FM={black}}
                    ]
    
                        \addplot [scatter] table [x=uemsc, y=emsc, meta=algorithm, only marks] {results-measures.csv};

                        \addplot[mark=none, black, dashed] {x};
                        
                    \end{axis}
                \end{tikzpicture}
                \caption{Measured values.}
                \label{fig:eva:measures}
            \end{subfigure}%
            \begin{subfigure}[b]{0.5\linewidth}
                \begin{tikzpicture}
                    \begin{loglogaxis}[
                        xmin=5, xmax=20000,
                        ymin=5, ymax=100000000,
                        xlabel = uEMSC run time (ms),
                        ylabel = EMSC run time (ms),
                        legend pos=north west,
                        width=\linewidth,
                        mark=o,
                        point meta=explicit symbolic,
                        scatter/classes={
                            IMf={blue},%
                            DFM={red},%
                            FM={black}}
                    ]
    
                        \addplot [scatter] table [x=uemsc, y=emsc, meta=algorithm, only marks] {times-measures.csv};
                        
                        \addplot[mark=none, black, domain=5:20000, dashed] {x};
                        
                    \end{loglogaxis}
                \end{tikzpicture}
                \caption{Run time.}
                \label{fig:eva:times}
            \end{subfigure}\\
            \definecolor{colfreq}{RGB}{27,158,119}
            \definecolor{colbicl}{RGB}{217,95,2}
            \definecolor{colalig}{RGB}{117,112,179}
            \begin{subfigure}[b]{0.5\linewidth}
                \begin{tikzpicture}
                    \begin{axis}[
                        xmin = 0, xmax = 1,
                        ymin = 0, ymax = 1,
                        xlabel = uEMSC,
                        ylabel = EMSC,
                        legend pos=north west,
                        width=\linewidth,
                        mark=o,
                        point meta=explicit symbolic,
                        scatter/classes={
                            freq={colfreq},%
                            bicl={colbicl},%
                            alig={colalig}}
                    ]
    
                        \addplot [scatter] table [x=uemsc, y=emsc, meta=stochastic-algorithm, only marks] {results-measures.csv};

                        \addplot[mark=none, black, dashed] {x};
                        
                    \end{axis}
                \end{tikzpicture}
                \caption{Measured values.}
                \label{fig:eva:measures2}
            \end{subfigure}%
            \begin{subfigure}[b]{0.5\linewidth}
                \begin{tikzpicture}
                    \begin{loglogaxis}[
                        xmin=5, xmax=20000,
                        ymin=5, ymax=100000000,
                        xlabel = uEMSC run time (ms),
                        ylabel = EMSC run time (ms),
                        legend pos=north west,
                        width=\linewidth,
                        mark=o,
                        point meta=explicit symbolic,
                        scatter/classes={
                            freq={colfreq},%
                            bicl={colbicl},%
                            alig={colalig}}
                    ]
    
                        \addplot [scatter] table [x=uemsc, y=emsc, meta=stochastic-algorithm, only marks] {times-measures.csv};
                        
                        \addplot[mark=none, black, domain=5:20000, dashed] {x};
                        
                    \end{loglogaxis}
                \end{tikzpicture}
                \caption{Run time.}
                \label{fig:eva:times2}
            \end{subfigure}
            \caption{Results of our feasibility experiment (only complete logs included). Dashed lines denote equality. Colours indicate \textcolor{blue}{IMf}, \textcolor{red}{DFM} and FM; or \textcolor{colfreq}{FBE}, \textcolor{colbicl}{BCE} and \textcolor{colalig}{ABE}.}
            \label{fig:eva}
        \end{figure}

        Both uEMSC and EMSC were feasible for the vast majority of log-model combinations of this experiment: as can be seen in Figure~\ref{fig:eva:times}, for most logs uEMSC's run time was at most a few seconds, while EMSC required up to a day.
            
        There are exceptions: BPIC11, for which uEMSC required up to half an hour or did not return a result in our main memory available (40GB).
        In particular for BPIC11 IMf, there were 345 silent transitions, leading to a too-large state space to be explored in the product stochastic transition system.
        EMSC takes in general much longer, with BPIC12-AW IMf FBE and BCE running out of memory.
            
        A notable case where uEMSC ran longer than EMSC is mimic-serv IMf FBE and mimic-serv IMf BCE, where uEMSC takes 5-6 seconds, while EMSC only needs a few miliseconds. 
        Figure~\ref{fig:eva:mimic-serv} in~\ref{app:results} shows this model.
        In this model, there is firstly a choice between on the one hand executing the visible transition NB with weight 7819 and on the other hand a large block of parallel behaviour, which starts with a silent transition of weight 1.
        EMSC samples only one trace from this model, with probability mass $\frac{7819}{7820}$.
        This leads to an answer quickly, whereas uEMSC performs an analysis on the large product stochastic transition system due to the 55 silent transitions, which also appear in concurrent settings.
        
        In Figure~\ref{fig:eva:measures}, we observe the trace-based nature of uEMSC vs. the event-based nature of EMSC: for a too-low EMSC value, uEMSC stays close to 0, as only events can be mapped but the probabilistic overlap in traces is negligible.
            
        An exception, visible in Figure~\ref{fig:eva:measures}, are some points where EMSC is high (over 0.8) but uEMSC is still 0; these points are from
        BPIC2018-6 DFM FBE, BPIC2018-6 DFM BCE, BPIC20-id DFM FBE and BPIC20-id DFM BCE.
        All of these involve the FBE and BCE stochastic discovery techniques, while ABE does not exhibit this pattern. A reason can be found in that FBE puts a weight of 1 for each silent transition and BCE puts, in these logs, a value close to 1 on the silent transitions.
        In DFM models, silent transitions are only used for termination, such that these transitions with a weight close to 1 are competing against visible transitions with a weight close to their execution frequency (which is in the 10-thousands).
        Thus, on the one hand uEMSC assigns a low score to any trace in these models, which gets rounded to 0 eventually.
        EMSC on the other hand takes, is much less sensitive to this single weight as it takes a sample of behaviour of the model.
        

%% file: conclusion.tex
\section{Conclusion}
\label{sec:conclusion}
We have introduced (bounded) labelled stochastic processes, an extension of generalized stochastic Petri nets with (possibly duplicated) labels and silent transitions, and defined their finite-trace execution semantics. We have shown how different analysis techniques combining behaviors and their probabilities can be solved through a combination of techniques from Markov chains and their qualitative verification, suitably adapted and refined to deal with this challenging setting. All techniques are implemented in the \verb=StochasticLabelledPetriNets= plug-in of ProM. 

Our approach lazily handles silent transitions when intersecting the set of runs characterised by a labelled stochastic process (including silent steps), with the traces of interest (predicating on visible activities only). To specify traces of interest, we employ deterministic finite-state automata. By recalling that temporal logics on finite traces such as \LTLf and \LDLf have an automata-theoretic characterization based on deterministic finite-state automata \cite{DeVa13,DDMM22}, our approach provides, for the first time, an technique for model checking stochastic processes against such finite-trace \LTLf/\LDLf specifications.

When computing the probability of model traces, this is in contrast with ad-hoc probability computation techniques developed in the context of stochastic automata with $\varepsilon$-moves in natural language processing \cite{hanneforth2010epsilon}. There, silent steps are filtered out in an eager way, using ad-hoc algorithms. A natural follow-up of this work is then to comparatively assess our lazy approach with such eager algorithm. On a similar note, we intend to conduct a more extensive experimental evaluation, with the aim of assessing the usage of well-established state-of-the-art probabilistic model checkers (such as Prism\footnote{\url{https://www.prismmodelchecker.org}} and Storm\footnote{\url{https://www.stormchecker.org}}) to deal with our analysis tasks.

Notably, our approach seamlessly generalise to the case where weights/rates are not exact but come with a confidence interval: one just needs to transform the system of equations used for computing outcome probabilities into a corresponding system of inequalities. The next step is then to consider this and other extensions of the model, in particular looking into richer nets with nondeterministic transition effects and their mapping to Markov decision processes, nets with non-Markovian behaviors, and nets equipped with different distributions on durations/delays. This calls for a fine-grained analysis of the same problems studied here in such richer settings, where often analytical methods are out of reach and have to be replaced by approximate/numerical techniques or simulation.




\section*{Acknowledgments}
We thank Werner Nutt for interesting discussions on exponential random variables and their interplay.

%% file: results.tex
\begin{longtable}{lllrrrrrrrrr}
\caption{Results of the evaluation.}
\label{tbl:results}
\\
\toprule
Log & discovery & sto.dis. & \multicolumn{3}{l}{log}& \multicolumn{2}{l}{transitions}& \multicolumn{2}{l}{EMSC}& \multicolumn{2}{l}{uEMSC}\\
&&&traces&events&activities&&silent&&time (ms)&&time (ms)\\
\midrule
\endhead
\bottomrule
\endfoot
BPIC11&IMf&FBE&1143&150291&624&580&345&-&-&-&-\\*&&BCE&1143&150291&624&580&345&-&-&-&-\\*&&ABE&-&-&-&-&-&-&-&-&-\\*[0.1cm]
&DFM&FBE&1143&150291&624&4295&35&0.22&307\,002&0.0000&69\,541\\*&&BCE&1143&150291&624&4295&35&0.22&292\,145&0.0000&69\,110\\*&&ABE&-&-&-&-&-&-&-&-&-\\*[0.1cm]
&FM&FBE&1143&150291&624&627&3&0.19&11\,724\,000&0.0032&1\,809\,636\\*&&BCE&1143&150291&624&627&3&0.19&12\,497\,634&0.0032&1\,762\,811\\*&&ABE&1143&150291&624&627&3&0.14&104\,537&0.0003&1\,801\,286\\[0.1cm]
BPIC12-a&IMf&FBE&13087&60849&10&15&5&0.65&53&0.1870&42\\*&&BCE&13087&60849&10&15&5&0.65&63&0.1870&42\\*&&ABE&13087&60849&10&15&5&0.86&62&0.4801&51\\*[0.1cm]
&DFM&FBE&13087&60849&10&16&3&0.85&53&0.7295&42\\*&&BCE&13087&60849&10&16&3&0.85&57&0.7295&41\\*&&ABE&13087&60849&10&16&3&0.96&56&0.8327&55\\*[0.1cm]
&FM&FBE&13087&60849&10&13&3&0.33&8\,135&0.0016&59\\*&&BCE&13087&60849&10&13&3&0.33&6\,131&0.0016&46\\*&&ABE&13087&60849&10&13&3&0.53&13\,871&0.0017&42\\[0.1cm]
BPIC12-aw&IMf&FBE&13087&204638&16&44&28&-&-&0.0000&1\,154\\*&&BCE&13087&204638&16&44&28&-&-&0.0000&1\,084\\*&&ABE&13087&204638&16&44&28&-&-&0.0025&1\,794\\*[0.1cm]
&DFM&FBE&13087&204638&16&48&6&0.84&24\,529&0.0000&723\\*&&BCE&13087&204638&16&48&6&0.84&24\,714&0.0000&685\\*&&ABE&13087&204638&16&48&6&0.79&3\,249\,424&0.1395&732\\*[0.1cm]
&FM&FBE&13087&204638&16&19&3&-&-&0.0000&2\,547\\*&&BCE&13087&204638&16&19&3&-&-&0.0000&2\,612\\*&&ABE&13087&204638&16&19&3&-&-&0.0000&2\,591\\[0.1cm]
BPIC13-cp&IMf&FBE&1487&6660&7&15&9&0.51&25&0.0143&28\\*&&BCE&1487&6660&7&15&9&0.55&42&0.0142&17\\*&&ABE&1487&6660&7&15&9&0.76&130&0.4175&21\\*[0.1cm]
&DFM&FBE&1487&6660&7&13&1&0.95&26\,983&0.0002&15\\*&&BCE&1487&6660&7&13&1&0.95&24\,674&0.0002&18\\*&&ABE&1487&6660&7&13&1&0.91&137&0.6547&15\\*[0.1cm]
&FM&FBE&1487&6660&7&10&3&0.49&708&0.0517&59\\*&&BCE&1487&6660&7&10&3&0.49&671&0.0517&66\\*&&ABE&1487&6660&7&10&3&0.63&1\,223\,077&0.0491&76\\[0.1cm]
BPIC13-i&IMf&FBE&7554&65533&13&31&20&0.39&5\,605\,563&-0.0000&1\,523\\*&&BCE&7554&65533&13&31&20&0.46&16\,199\,640&0.0000&1\,621\\*&&ABE&7554&65533&13&31&20&0.79&8\,709\,087&0.0527&1\,871\\*[0.1cm]
&DFM&FBE&7554&65533&13&24&2&0.88&3\,931\,804&0.0442&253\\*&&BCE&7554&65533&13&24&2&0.88&4\,931\,212&0.0442&264\\*&&ABE&7554&65533&13&24&2&0.80&6\,735\,050&0.1500&257\\*[0.1cm]
&FM&FBE&7554&65533&13&16&3&0.34&1\,139\,508&0.0048&843\\*&&BCE&7554&65533&13&16&3&0.34&1\,782\,487&0.0048&875\\*&&ABE&7554&65533&13&16&3&0.67&25\,740\,542&0.0050&895\\[0.1cm]
BPIC13-op&IMf&FBE&819&2351&5&16&11&0.39&24&0.0056&37\\*&&BCE&819&2351&5&16&11&0.52&53&0.0709&30\\*&&ABE&819&2351&5&16&11&0.75&845&0.4269&29\\*[0.1cm]
&DFM&FBE&819&2351&5&17&5&0.70&185\,299&0.2242&11\\*&&BCE&819&2351&5&17&5&0.71&169\,801&0.2242&13\\*&&ABE&819&2351&5&17&5&0.92&173&0.7456&10\\*[0.1cm]
&FM&FBE&819&2351&5&8&3&0.70&127&0.3383&24\\*&&BCE&819&2351&5&8&3&0.70&145&0.3383&26\\*&&ABE&819&2351&5&8&3&0.73&209\,141&0.3440&33\\[0.1cm]
BPIC17-o&IMf&FBE&42995&193849&8&9&1&0.90&186&0.5390&170\\*&&BCE&42995&193849&8&9&1&0.90&185&0.5390&162\\*&&ABE&42995&193849&8&9&1&0.91&192&0.5811&162\\*[0.1cm]
&DFM&FBE&42995&193849&8&10&3&0.96&183&0.8419&168\\*&&BCE&42995&193849&8&10&3&0.96&181&0.8419&171\\*&&ABE&42995&193849&8&10&3&0.96&186&0.8419&168\\*[0.1cm]
&FM&FBE&42995&193849&8&11&3&0.31&971&0.0023&162\\*&&BCE&42995&193849&8&11&3&0.31&1\,693&0.0023&150\\*&&ABE&42995&193849&8&11&3&0.50&14\,026&0.0026&172\\[0.1cm]
BPIC18-6&IMf&FBE&14750&132963&10&24&15&0.55&1\,225\,049&0.0003&2\,394\\*&&BCE&14750&132963&10&24&15&0.59&14\,419\,026&0.0003&2\,556\\*&&ABE&14750&132963&10&24&15&0.83&5\,879\,827&0.0684&2\,729\\*[0.1cm]
&DFM&FBE&14750&132963&10&16&1&0.88&20\,288&0.0000&410\\*&&BCE&14750&132963&10&16&1&0.89&18\,087&0.0000&402\\*&&ABE&14750&132963&10&16&1&0.85&3\,547\,940&0.2025&21\,579\\*[0.1cm]
&FM&FBE&14750&132963&10&13&3&0.28&4\,074\,955&0.0008&1\,186\\*&&BCE&14750&132963&10&13&3&0.28&3\,589\,177&0.0008&1\,181\\*&&ABE&14750&132963&10&13&3&0.64&85\,640\,171&0.0009&1\,237\\[0.1cm]
BPIC20-dd&IMf&FBE&10500&56437&17&33&17&0.40&25\,320&0.0001&41\\*&&BCE&10500&56437&17&33&17&0.49&72\,573&0.0001&46\\*&&ABE&10500&56437&17&33&17&0.93&53&0.8049&41\\*[0.1cm]
&DFM&FBE&10500&56437&17&9&1&0.84&50&0.5807&35\\*&&BCE&10500&56437&17&9&1&0.84&50&0.5807&37\\*&&ABE&10500&56437&17&9&1&0.92&61&0.8079&36\\*[0.1cm]
&FM&FBE&10500&56437&17&20&3&0.29&16\,330&0.0019&60\\*&&BCE&10500&56437&17&20&3&0.29&15\,637&0.0019&64\\*&&ABE&10500&56437&17&20&3&0.55&169\,561&0.0012&65\\[0.1cm]
BPIC20-id&IMf&FBE&6449&72151&34&66&36&0.51&14\,278\,269&-0.0000&99\\*&&BCE&6449&72151&34&66&36&0.53&8\,095\,583&-0.0000&80\\*&&ABE&6449&72151&34&66&36&0.78&1\,494&0.2690&108\\*[0.1cm]
&DFM&FBE&6449&72151&34&51&3&0.86&22\,502&0.0000&86\\*&&BCE&6449&72151&34&51&3&0.87&21\,782&0.0000&75\\*&&ABE&6449&72151&34&51&3&0.82&41\,291&0.1830&96\\*[0.1cm]
&FM&FBE&6449&72151&34&37&3&0.18&10\,800\,412&0.0000&622\\*&&BCE&6449&72151&34&37&3&0.18&11\,029\,039&0.0000&618\\*&&ABE&6449&72151&34&37&3&0.42&19\,241\,284&0.0000&665\\[0.1cm]
BPIC20-pt&IMf&FBE&2099&18246&29&63&37&0.46&2\,710\,449&0.0000&30\\*&&BCE&2099&18246&29&63&37&0.46&2\,975\,697&0.0000&27\\*&&ABE&2099&18246&29&63&37&0.80&2\,770&0.2248&44\\*[0.1cm]
&DFM&FBE&2099&18246&29&31&2&0.67&61&0.1624&13\\*&&BCE&2099&18246&29&31&2&0.67&63&0.1624&12\\*&&ABE&2099&18246&29&31&2&0.89&62&0.3910&13\\*[0.1cm]
&FM&FBE&2099&18246&29&32&3&0.22&1\,151\,046&0.0006&109\\*&&BCE&2099&18246&29&32&3&0.22&1\,250\,116&0.0006&118\\*&&ABE&2099&18246&29&32&3&0.51&1\,637\,301&0.0002&127\\[0.1cm]
BPIC20-rf&IMf&FBE&6886&36796&19&29&13&0.37&55&-0.0000&26\\*&&BCE&6886&36796&19&29&13&0.37&80&-0.0000&22\\*&&ABE&6886&36796&19&29&13&0.89&39&0.6262&32\\*[0.1cm]
&DFM&FBE&6886&36796&19&12&1&0.83&37&0.5483&24\\*&&BCE&6886&36796&19&12&1&0.83&51&0.5483&21\\*&&ABE&6886&36796&19&12&1&0.92&44&0.8118&23\\*[0.1cm]
&FM&FBE&6886&36796&19&22&3&0.30&16\,568&0.0024&44\\*&&BCE&6886&36796&19&22&3&0.30&15\,904&0.0024&47\\*&&ABE&6886&36796&19&22&3&0.55&116\,739&0.0018&51\\[0.1cm]
mimic-serv&IMf&FBE&58926&73343&20&55&35&0.07&71&0.0038&5\,433\\*&&BCE&58926&73343&20&55&35&0.07&94&0.0038&5\,784\\*&&ABE&58926&73343&20&55&35&0.74&1\,697\,889&0.5318&9\,265\\*[0.1cm]
&DFM&FBE&58926&73343&20&19&9&0.83&65&0.7241&63\\*&&BCE&58926&73343&20&19&9&0.83&64&0.7241&57\\*&&ABE&58926&73343&20&19&9&0.86&66&0.8150&79\\*[0.1cm]
&FM&FBE&58926&73343&20&23&3&0.79&49\,588&0.5790&204\\*&&BCE&58926&73343&20&23&3&0.79&50\,026&0.5790&193\\*&&ABE&58926&73343&20&23&3&0.92&114&0.8051&185\\[0.1cm]
mimic-trans&IMf&FBE&3019&13411&18&37&19&0.21&10\,231\,861&0.0003&53\\*&&BCE&3019&13411&18&37&19&0.22&3\,675\,148&0.0003&60\\*&&ABE&3019&13411&18&37&19&0.66&74\,478&0.1681&66\\*[0.1cm]
&DFM&FBE&3019&13411&18&37&4&0.87&219\,124&0.5112&49\\*&&BCE&3019&13411&18&37&4&0.87&106\,450&0.5112&46\\*&&ABE&3019&13411&18&37&4&0.89&184&0.5877&52\\*[0.1cm]
&FM&FBE&3019&13411&18&21&3&0.39&620\,518&0.0100&145\\*&&BCE&3019&13411&18&21&3&0.39&593\,890&0.0100&145\\*&&ABE&3019&13411&18&21&3&0.59&1\,309\,639&0.0098&145\\[0.1cm]
Roadfines&IMf&FBE&150370&561470&11&24&13&0.54&449&0.0139&394\\*&&BCE&150370&561470&11&24&13&0.54&443&0.0104&361\\*&&ABE&150370&561470&11&24&13&0.77&422&0.2940&392\\*[0.1cm]
&DFM&FBE&150370&561470&11&9&3&0.88&427&0.6840&408\\*&&BCE&150370&561470&11&9&3&0.88&407&0.6840&388\\*&&ABE&150370&561470&11&9&3&0.95&419&0.8196&417\\*[0.1cm]
&FM&FBE&150370&561470&11&14&3&0.42&13\,637&0.0236&408\\*&&BCE&150370&561470&11&14&3&0.42&27\,969&0.0236&426\\*&&ABE&150370&561470&11&14&3&0.52&146\,820&0.0193&502\\[0.1cm]
Sepsis&IMf&FBE&1050&15214&16&29&15&0.51&8\,246\,699&-0.0000&406\\*&&BCE&1050&15214&16&29&15&0.60&4\,252\,316&0.0000&408\\*&&ABE&1050&15214&16&29&15&0.70&1\,926\,387&0.0003&420\\*[0.1cm]
&DFM&FBE&1050&15214&16&77&10&0.63&6\,517\,108&0.0000&188\\*&&BCE&1050&15214&16&77&10&0.64&7\,041\,945&0.0000&188\\*&&ABE&1050&15214&16&77&10&0.78&1\,253\,417&0.0451&191\\*[0.1cm]
&FM&FBE&1050&15214&16&19&3&0.21&18\,589\,789&0.0000&351\\*&&BCE&1050&15214&16&19&3&0.21&20\,377\,324&0.0000&366\\*&&ABE&1050&15214&16&19&3&0.50&8\,652\,211&0.0000&364\\[0.1cm]
\end{longtable}